\documentclass[12pt]{article}

\usepackage[letterpaper, margin=1in]{geometry}
\usepackage{fullpage}

\usepackage{times}

\usepackage{amsmath,amssymb}
\usepackage{verbatim}
\usepackage{amsthm}
\usepackage{amsfonts}
\usepackage{xspace}
\usepackage{subcaption}
\usepackage{setspace}
\usepackage{url}
\usepackage{mathrsfs}
\usepackage{hyperref}

\usepackage{algorithm}
\usepackage[noend]{algpseudocode}
\usepackage[framemethod=tikz]{mdframed}
\usepackage{soul}
\usepackage{cancel}

\theoremstyle{plain}
\newtheorem{theorem}{Theorem}
\newtheorem{lemma}[theorem]{Lemma}
\newtheorem{claim}[theorem]{Claim}

\newtheorem{corollary}[theorem]{Corollary}

\newtheorem{definition}{Definition}

\newtheorem{hypothesis}{Hypothesis}

\newenvironment{proofof}[1]{ {\noindent \em Proof of #1.}\/}{\hfill\qedsymbol\bigskip}

\newcommand{\remove}[1]{}
\newcommand{\suppress}[1]{}

\newcommand{\NN}{{\mathbb{N}}}
\newcommand{\RR}{{\mathbb{R}}}

\newcommand{\E}{{\mathrm{E}}}
\newcommand{\eps}{\varepsilon}

\DeclareMathOperator{\cost}{cost}

\DeclareMathOperator{\OPT}{OPT}

\DeclareMathOperator{\poly}{poly}

\DeclareMathOperator{\conv}{conv}
\DeclareMathOperator{\core}{core}
\DeclareMathOperator{\Cost}{cost}
\DeclareMathOperator{\Val}{val}
\DeclareMathOperator{\Far}{far}
\DeclareMathOperator{\Close}{cls}
\DeclareMathOperator{\Mid}{mid}

\renewcommand{\tilde}{\widetilde}
\newcommand{\maxcover}{\mbox{\sf MaxCover}\xspace}

\newcommand{\SETH}{\mbox{\sf SETH}\xspace}

\newcommand{\N}{\mathbb{N}}

\usepackage{nicefrac}
\usepackage[utf8]{inputenc}
\usepackage{microtype}

\usepackage{tikz}
\usetikzlibrary{arrows.meta}

\usepackage{comment}

\usepackage{color}

\newif\ifnotes\notestrue 
\ifnotes
\newcommand{\chris}[1]{\textcolor{green}{{\bf (Chris:} {#1}{\bf ) }} \marginpar{\tiny\bf
             \begin{minipage}[t]{0.5in}
               \raggedright S:
            \end{minipage}}}
\newcommand{\karthik}[1]{\textcolor{blue}{{\bf (Karthik:} {#1}{\bf ) }} \marginpar{\tiny\bf
             \begin{minipage}[t]{0.5in}
               \raggedright K:
            \end{minipage}}}
\newcommand{\samson}[1]{\textcolor{purple}{{\bf (Samson:} {#1}{\bf ) }} \marginpar{\tiny\bf
             \begin{minipage}[t]{0.5in}
               \raggedright S:
            \end{minipage}}} 
\newcommand{\euiwoong}[1]{\textcolor{purple}{{\bf (Euiwoong:} {#1}{\bf ) }} \marginpar{\tiny\bf
             \begin{minipage}[t]{0.5in}
               \raggedright E:
            \end{minipage}}} 
\else
\newcommand{\chris}[1]{}
\newcommand{\karthik}[1]{}
\newcommand{\samson}[1]{}
\newcommand{\euiwoong}[1]{}
\fi

\usepackage{graphicx} 
\usepackage{MnSymbol}

\DeclareRobustCommand{\mapsfrom}{\text{\reflectbox{$\mapsto$}}}

\allowdisplaybreaks

\begin{document}

\title{A Fixed Parameter Sublinear Time Approximation Scheme for Metric Min-Sum Clustering}

\author{
 \hspace{0.4in}
 Karthik C. S.\thanks{Rutgers University.
 E-mail: {\tt karthik.cs@rutgers.edu}. This work was supported by the National Science Foundation under Grants CCF-2313372 and CCF-2443697.}
 \hspace{0.4in}
 \and 
 Euiwoong Lee\thanks{University of Michigan. 
 E-mail: {\tt euiwoong@umich.edu}. Work partly supported by NSF grant CCF-2236669.}
 \hspace{0.4in}
 \and 
 Yuval Rabani\thanks{The Hebrew University of Jerusalem. 
 E-mail: {\tt yuval.rabani@mail.huji.ac.il}. Work partly supported by ISF grant 389-22 and 
 by BSF grant 2023607.}
 \hspace{0.4in}
 \vspace{-0.2in}
 \and 
 Chris Schwiegelshohn
 \thanks{Aarhus University. 
 E-mail: {\tt cschwiegelshohn@gmail.com}} 
 \and 
 Samson Zhou\thanks{Texas A\&M University. 
 E-mail: {\tt samsonzhou@gmail.com}}
}

\date{\today}

\maketitle

\begin{abstract}
Metric $k$-minsum clustering is a simple, natural, appealing density-based clustering criterion whose computational complexity remains poorly understood, despite the progress achieved for other clustering criteria in the past three decades. In this paper we give an efficient approximation scheme for this problem that runs in sublinear time for mildly increasing $k$, and thus is fixed parameter tractable with respect to the parameter $k$. 
We complement this result by showing that similar algorithms are unlikely to exist for the closely related problems $k$-median and balanced $k$-median, assuming popular fine-grained complexity conjectures. Specifically, we show that no sublinear time algorithm with small approximation factors can exist even for $2$-median and balanced $2$-median. \vspace{0.1cm}

It was previously known that no efficient PTAS can exist for $k$-median and our results show a strong and perhaps surprising separation between the $k$-minsum clustering problem and $k$-median even in the sublinear regime. Our upper bounds improve qualitatively upon decades old results. Our lower bounds establish the first use of the Distributed PCP framework to clustering problems.
\bigskip\bigskip

\noindent {\em AI apology: AI was used neither to derive these results nor to produce their presentation.}
\end{abstract}

\thispagestyle{empty}
\newpage
\setcounter{page}{1}

\section{Introduction}\label{sec:intro}

\paragraph{Problem statement and motivation.}
The worst-case approximability of many hard clustering criteria, most notably center-based ones 
such as $k$-center, $k$-median, or $k$-means, is quite well-understood. The gaps between upper 
bounds and lower bounds, including the dependence on $k$ or the data dimension, are relatively 
small, definitely with little room for qualitative improvement. This is not the case with $k$-minsum 
clustering, a simple, natural, appealing density-based clustering criterion. The setting is a finite set 
$A$ of $n=|A|$ points, endowed with an explicit or implicit semi-metric $d$.\footnote{A semi-metric 
$d:A\times A\rightarrow\RR$ is non-negative, symmetric, and satisfies the triangle inequality, but 
pairs of points can have distance $0$, indicating that they are co-located.} Given $(A,d)$ and 
$k\in\NN$, the $k$-minsum clustering problem is to compute a clustering 
$\mathcal{C}=\{C_1\sqcup C_2\sqcup\cdots\sqcup C_k\}$ minimizing
$$
\text{cost}(\mathcal{C})=\sum_{i=1}^k \sum_{x,y\in C_i} d(x,y).
$$
In other words, we seek to partition $A$ into $k$ disjoint clusters so as to minimize the sum of 
intra-cluster distances. Fernandez de la Vega and Mathieu~\cite{FK98} showed that this problem 
is NP-hard, even for $k=2$; the complement objective is metric MAX CUT (and more generally for 
$k > 2$, metric MAX $k$-CUT).

\paragraph{Our results.}
Our contributions in this paper are the following results. Firstly, we give, for $k$-minsum clustering 
parametrized by the number of clusters $k$, a fixed parameter\footnote{Fixed parameter tractability 
means here running time bounded by $f(k,\varepsilon)\cdot n^c$, for some function $f$, and an
absolute constant $c$ that is independent of $n$, $k$, and $\varepsilon$.} approximation scheme 
that runs in time nearly linear in the number of points $n$. Note that if the metric $d$ is given explicitly, 
then the input size is $\Omega(n^2)$. Therefore, our algorithm runs in sublinear time, and in particular 
queries a negligible fraction of pairwise distances. More concretely, we prove the following theorem.
\begin{theorem}\label{thm: main}
There exists an algorithm that on every input $(A,d)$ computes a $(1+\varepsilon)$-approximately optimal 
$k$-minsum clustering of $A$ using $n^{1+o(1)}\cdot 2^{2^{O(k\log(k/\varepsilon))}}$ 
arithmetic operations, assuming that for any $x,y\in A$ computing $d(x,y)$ is a single 
operation.
\end{theorem}
This theorem improves upon the 23 years old result of Fernandez de la Vega et al.~\cite{FKKR03}, who 
gave an approximation scheme with running time $n^{3k}\cdot 2^{O(\eps^{-k^2})}$ (so their polynomial 
dependence on the number of points degrades with $k$). It also extends the 27 years old result of 
Indyk~\cite{Ind99b} that gave an $n^{1+o(1)}\cdot 2^{(1/\eps)^{O(1)}}$ time approximation scheme for the 
special case of $k=2$. In particular, the algorithm of Theorem~\ref{thm: main} runs in sublinear time for 
all $k = o\left(\frac{\log\log n}{\log\log\log n}\right)$, whereas the algorithm 
in~\cite{FKKR03} runs in (far from linear, and quite prohibitive) polynomial time only when $k$ is an 
absolute constant (assuming in both cases that $\varepsilon > 0$ is an absolute constant). Moreover, 
for $k = o\left(\frac{\log\log n}{\log\log\log n}\right)$, our result also improves (qualitatively) upon the $4+\varepsilon$ 
approximation guarantee of the $(k/\varepsilon)^{O(k)}\cdot n^{1+o(1)}$ time algorithm of Czumaj and 
Sohler~\cite{CS07} (see also their earlier work~\cite{CS04}). Note that their sublinear time guarantee 
holds for a larger range of all $k = o\left(\frac{\log n}{\log\log n}\right)$. Finally, our algorithm is also 
easily parallelizable, for instance in the MPC model (see the survey~\cite{IKLMV23}).

We note that as part of the proof of Theorem~\ref{thm: main} we also design a sublinear time
approximation scheme for metric MAX $k$-CUT. This improves upon the linear time approximation
scheme for the same problem, outlined in~\cite{FKKR03}.

We observe that combining ideas in literature, such as the NP-hardness for approximating $k$-minsum 
clustering of~\cite{GI03}, and combining it with the Gap Exponential Time Hypothesis~\cite{D16,MR17}, 
one can prove a conditional parametrized lower bound for $k$-minsum clustering by showing that there 
is no $2^{o(k/\varepsilon)}\cdot 2^{o(n)}$ time $(1+\varepsilon)$-approximation algorithm for metric 
$k$-minsum clustering (see Theorem~\ref{thm:lbminsum} for a formal statement).

Given Theorem \ref{thm: main}, one might ask whether there exist other clustering problems for which a 
sublinear time PTAS is possible? A likely candidate is the balanced $k$-median problem, which has been 
previously used to design approximation algorithms for the $k$-minsum objective. Specifically, given a 
cluster and the center, the cluster cost is the sum of distances of points to the center, weighted by the size 
of the cluster. Unfortunately, both the balanced $k$-median problem and the standard $k$-median problem, 
where the cluster cost is not weighted by its size, likely do not admit algorithms similar to the one presented 
in this work. Specifically, we show that even for (balanced) $2$-median, no subquadratic time algorithm that 
achieves a small approximation factor is possible, under popular conjectures. The connection between these 
problems will be discussed in greater detail later.
\begin{theorem}[Informal version of Theorem~\ref{thm:median}]\label{thm:SETHinformal}
Assuming the Strong Exponential Time Hypothesis, for every $\eps>0$ there exists 
$\delta>0$ such that no $n^{2-\eps}$ time algorithm that takes as input a set $F$ of $n$ facilities, a set 
$C$ of $m=O_\eps(\log n)$ clients, and a metric $d$ on $F\cup C$, can distinguish between the following 
two cases:
\begin{description}
    \item[Completeness:] The $2$-median cost with balanced clusters is $\le m$.
    \item[Soundness:] The same $2$-median cost is $\ge (1+\delta)\cdot m$.
\end{description}
\end{theorem}

Theorem \ref{thm:median} complements Manurangsi's lower bound~\cite{Man20} for balanced $k$-median 
when $k$ is large. It implies in particular that even for $2$-median, very accurate approximation algorithms 
likely cannot run in sublinear time. A particular notable feature of this lower bound is that if the metric space 
is given by distances between $m =\text{polylog}(n)$ clients and $n$ facilities, then the input consists of 
$mn = n\text{polylog}(n)$ distances. Thus, for certain ranges of parameters, the lower bound therefore implies 
that not only linear time, but even quadratic running times can be necessary.

\paragraph{Technical contribution for the PTAS.}
In this discussion it is worth drawing a parallel between the metric case that is the focus of this paper and 
the $\ell_2^2$ variant in which $A$ is realized as a set of real-valued vectors and $d$ is the squared-Euclidean 
distance (hence, the triangle inequality is violated). The aforementioned~\cite{FKKR03} gave, for fixed $k$, a 
polynomial time approximation scheme for $\ell_2^2$ $k$-minsum clustering, and this was very recently improved 
to give running time similar to Theorem~\ref{thm: main} by Karthik et al.~\cite{KLRSZ25}. These and other results 
on the $\ell_2^2$ variant rely on a fundamental connection in the $\ell_2^2$ setting between the minsum criterion 
and the center-based $k$-means criterion, so we must discuss briefly $k$-means clustering. 

In $k$-means clustering, we must choose a center for each cluster, the cost of a cluster is the sum over the cluster 
points of squared Euclidean distance to the center, and the objective is to minimize the sum of cluster costs. If we 
know the cluster centers, the partition is induced by their Voronoi diagram; each data point connects to the closest 
center. Algorithmically, it helps immensely that, given a cluster, its optimal cluster center is uniquely the mean (center 
of mass) of the cluster. This is often exploited through a good approximation of the location of each cluster center 
using a small sample of points from the cluster's {\em core}.\footnote{The core is the subset of points within a constant 
factor times the average distance to the center. Depending on the constant, it contains most of the points of the cluster, 
thanks to Markov's inequality.} To guarantee that every cluster core is ``hit'' (they may differ substantially in size) a 
distribution ensuring diversification is required, for instance sampling in proportion to the distance to previously sampled 
points. Put in simpler terms, for small $k$ there is an efficient construction of a {\em coreset} for 
$k$-means~\cite{Cohen-AddadLSS22,FMS07}.

The relation to $\ell_2^2$ $k$-minsum clustering stems from the fact that the cost of an $\ell_2^2$ $k$-minsum cluster 
can be expressed exactly as the $k$-means cost of the same cluster times its size. Thus, the starting idea of results 
such as the aforementioned~\cite{FKKR03,KLRSZ25} is to approximate the cluster centers in much the same way as 
$k$-means results do. This is more challenging than in the $k$-means case, because the convex hulls of clusters can 
intersect, and cluster centers may even be co-located. In fact, this difficulty is resolved very differently in the two 
above-mentioned papers, and more efficiently in the latter of the two. Another issue is that even if we have accurate 
cluster centers, assigning points to clusters is non-trivial, and in fact is likely to be computationally hard. The solution 
here is to enumerate over (approximate) cluster sizes. If center-size pairs are known, partitioning optimally into clusters 
becomes a simple assignment problem. The enumeration incurs a factor proportional to $\log^k n$ in the running 
time.\footnote{This is tolerable, as for small $k$ it is $n^{o(1)}$, and for large $k$ it is $2^{O(k\log k)}$.}
In essence, the state-of-the-art~\cite{KLRSZ25} relies, as with $k$-means, on the efficient construction of a coreset.

Now, in the case of metric $k$-minsum clustering, these ideas are insufficient. In particular, the input data is discrete 
and there is no natural definition of a cluster center. Of course, one may resort to the isometric Fr\'echet embedding 
into $\ell_\infty^n$. However, this offers little geometric insight. Even if the dataset is representable in, say, Euclidean 
space, there is no cluster center that would satisfy with respect to the metric objective all properties that the center of 
mass satisfies with respect to the $\ell_2^2$ objective: ($i$) the cost of the cluster is its size times the sum of distances
to the center (or some other nice and easily approximated formula); ($ii$) the center is a minimizer of this formula;
($iii$) a good approximate center can be computed within the desired efficiency. The total cluster distance to its center 
can be optimized in this case by the median, the solution to the Fermat-Weber location problem. However, the median 
neither is approximated with ease, nor approximates well the minsum objective.

Our main technical contribution in the proof of Theorem~\ref{thm: main} is a solution to overcome this obstacle,
applying a toolkit that includes known and new techniques. We will explain briefly how this is achieved, comparing 
our approach to that leading to the (non-parametric) result of~\cite{FKKR03}, and to the $\ell_2^2$ setting 
that was solved in~\cite{KLRSZ25}.

In the natural embedding of $(A,d)$ into the Wasserstein-$1$ space ${\cal W}_1(A,d)$, the set of measures 
on $A$ with total measure $0$, endowed with the $d$-induced transportation norm, the sum of (transportation) 
distances to a cluster's center of mass, scaled by the cluster size, is exactly the cluster cost. (The distance of
a datapoint to this ${\cal W}_1$ center of mass is simply its average distance to all the points in the cluster; 
we use mostly this property and the fact that ${\cal W}_1(A,d)$ is a metric space that includes an isometric copy 
of $(A,d)$.) Unfortunately, the center of mass of the cluster's core does not well-approximate the entire cluster's 
center of mass. The center of mass is not a minimizer of the sum of distances, so the error can go either way.  
If we do not focus on the core, we cannot bound the probability of large deviation of a uniformly-at-random 
sample. Hence, the machinery that works in the $\ell_2^2$ case fails here; the assignment of points to cluster 
centers may produce a clustering of cost far from optimal. In simpler terms, we currently do not know how to 
compute a coreset in fixed-parameter tractable time.

But not all is lost. Even a relatively poor approximation of the ${\cal W}_1$ center of mass suffices to give a
good approximation of the cost due to cluster outliers (cluster points that are far from the ${\cal W}_1$ center 
of mass of the cluster, relative to the average). Moreover, it is also sufficient to identify well separated clusters, 
but not all clusters. Luckily within each union of non-well separated clusters, the clusters can be identified if 
we have an efficient approximation scheme for MAX $h$-CUT (where $h$ is the number of clusters in the 
union). The idea of assigning points to cluster unions is the core of the algorithm in~\cite{FKKR03}, who use 
uniform sampling to ``hit" proxy centers. The issue with uniform sampling is that it helps only with large clusters. 
Therefore, the algorithm in~\cite{FKKR03} iterates over decreasing size scales, constructing the clusters by 
order of size. This is an approach whose time complexity we cannot afford. In particular, it requires enumeration 
over an exact guess of the cluster sizes, and this is a major reason for the non-parametric time complexity. On 
the other hand, guessing exact cluster sizes simplifies the assignment of outlier points to clusters, because
alternative choices are exactly interchangeable.

In contrast, we design an adaptive non-uniform sampling distribution that generates a relatively small number 
of candidates, at least one of which ``hits" all clusters, large and small alike, appropriately. Armed with the
adaptive sampling process and enumeration to identify cluster center proxies, and an enumeration over 
approximate cluster sizes and costs, we can solve for every option the problem of assigning points to cluster
unions, and then apply a MAX $h$-CUT approximation scheme to each cluster union. Cluster cores (points
relatively close to their cluster's ${\cal W}_1$ center of mass) are quite easily assigned to cluster unions. 
However, the assignment of cluster outliers poses a serious challenge that requires a novel cross-cluster 
charging argument, as they can get mixed with the cores of smaller clusters. As we do not enumerate over 
exact cluster sizes, we cannot locally optimize our choices. In fact, we need to enumerate over alternative 
assignments, each leading to a solution that can be checked only after the entire clustering is computed. On 
the other hand, we simplify the task of partitioning the clusters into cluster unions, which in~\cite{FKKR03} is 
done by complete enumeration. We enumerate over approximate cluster costs, also for other purposes, and 
this coupled with the other parameters simplifies the detection of cluster separation. Finally, for every candidate 
clustering thus computed we compute its approximate cost, and return the (nearly) cheapest candidate. 

The stated running time is the result of ($i$) carefully controlling the number of candidate clusterings; ($ii$) efficiently 
computing the assignment of points (in particular outliers) to clusters (or cluster unions); ($iii$) designing new 
approximation schemes for metric MAX $k$-CUT and for estimating the minsum cost of a clustering, both in 
sublinear time. A linear time approximation scheme was previously known for metric MAX $k$-CUT, and is trivial 
for clustering cost estimation. The new metric MAX $k$-CUT sublinear time approximation scheme is inspired 
by results for dense graphs of Andoni et al.~\cite{ANOY14} and of Yaroslavtsev~\cite{Yaroslavtsev14}.

\paragraph{Technical contribution for the lower bounds.}
The proof of Theorem~\ref{thm:SETHinformal} builds on results stemming from the {\em distributed PCP} 
framework of Abboud et al.~\cite{ARW17}, Rubinstein~\cite{Rub18}, and Karthik et al.~\cite{KLM18}. While 
this framework has been used in literature to prove conditional lower bounds against subquadratic approximation 
algorithms for nearest-neighbor search type problems~\cite{ARW17,Rub18}, we show that when interpreted 
along the lines of set systems as in~\cite{KLM18}, we can actually obtain conditional lower bounds against 
subquadratic approximation algorithms for various center-based clustering objectives such as $2$-median, 
balanced $2$-median, $2$-means, and so on. The proof proceeds in two steps. Firstly, we start from the 
inapproximability under SETH against specialized label cover instances proven in~\cite{KLM18} and compose 
it with the ``hypercube partition gadgets'' given by Feige~\cite{Fei98} to obtain gap maxcoverage instances 
where in the completeness case only two sets suffice to cover the entire universe. However, we need to be 
careful in our analysis of this reduction to obtain the additional guarantees as specified in Theorem~\ref{thm:maxcoverage}. 
Secondly, we reduce this maxcoverage instance using the now standard reduction of Guha and 
Khuller~\cite{GK98} to obtain the clustering instances of Theorem~\ref{thm:SETHinformal}.

\paragraph{Further motivation.}
Aside from improving our theoretical understanding of a fundamental problem in unsupervised learning, minsum clustering also exhibits strong practical and conceptual advantages over more widely studied center-based objectives such as $k$-means and $k$-median. 
Whereas classical centroid-based clustering problems identify clusters by minimizing distances to centroids, thereby favoring convex and roughly spherical shapes, minsum clustering minimizes pairwise intra-cluster distances. 
Consequently, it can capture clusters with irregular, non-convex, or even overlapping shapes that are not well represented by a single centroid. 

A simple geometric example highlights this distinction. Consider two concentric rings of points in the plane. 
The optimal minsum clustering partitions the dataset into the two rings, reflecting their true underlying structure. 
In contrast, both $k$-means and $k$-median attempt to separate the points by a hyperplane, producing an unnatural division that mixes points from different rings. 
See Figure~\ref{fig:cluster} for an illustration of this difference. 

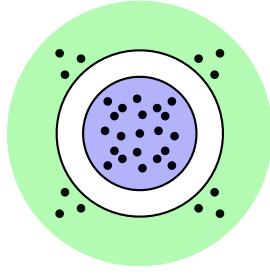
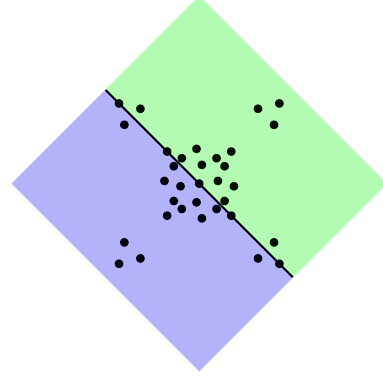
\begin{figure*}[!htb]
\centering
\begin{subfigure}[b]{0.45\textwidth} 
\centering
\begin{tikzpicture}[scale=0.5]
\begin{scope}[rotate=45]
\filldraw[green!95!black, opacity=0.3]  (0,0) circle (3.5);
\filldraw[white]  (0,0) circle (2.2);
\draw[black, thick] (0,0) circle (2.2); 
\filldraw[blue!95!black, opacity=0.3]  (0,0) circle (1.5);
\draw[black, thick] (0,0) circle (1.5); 

\filldraw[blue!5!black] (0,0) circle (0.1);
\filldraw[blue!5!black] (1.2,0) circle (0.1);
\filldraw[blue!5!black] (-1.2,0) circle (0.1);
\filldraw[blue!5!black] (0,1.2) circle (0.1);
\filldraw[blue!5!black] (0,-1.2) circle (0.1);
\filldraw[blue!5!black] (0.6,0.7) circle (0.1);
\filldraw[blue!5!black] (-0.6,-0.7) circle (0.1);
\filldraw[blue!5!black] (0.6,-0.7) circle (0.1);
\filldraw[blue!5!black] (-0.6,0.7) circle (0.1);
\filldraw[blue!5!black] (0.4,0.3) circle (0.1);
\filldraw[blue!5!black] (0.4,-0.3) circle (0.1);
\filldraw[blue!5!black] (-0.4,-0.3) circle (0.1);
\filldraw[blue!5!black] (-0.4,0.3) circle (0.1);
\filldraw[blue!5!black] (0.8,0.15) circle (0.1);
\filldraw[blue!5!black] (0.8,-0.15) circle (0.1);
\filldraw[blue!5!black] (-0.8,0.15) circle (0.1);
\filldraw[blue!5!black] (-0.8,-0.15) circle (0.1);
\filldraw[blue!5!black] (0.15,0.8) circle (0.1);
\filldraw[blue!5!black] (0.15,-0.8) circle (0.1);
\filldraw[blue!5!black] (-0.15,0.8) circle (0.1);
\filldraw[blue!5!black] (-0.15,-0.8) circle (0.1);

\filldraw[green!5!black] (2.5,0.3) circle (0.1);
\filldraw[green!5!black] (-2.5,0.3) circle (0.1);
\filldraw[green!5!black] (2.5,-0.3) circle (0.1);
\filldraw[green!5!black] (-2.5,-0.3) circle (0.1);
\filldraw[green!5!black] (3,0) circle (0.1);
\filldraw[green!5!black] (-3,0) circle (0.1);
\filldraw[green!5!black] (0.3,2.5) circle (0.1);
\filldraw[green!5!black] (0.3,-2.5) circle (0.1);
\filldraw[green!5!black] (-0.3,2.5) circle (0.1);
\filldraw[green!5!black] (-0.3,-2.5) circle (0.1);
\filldraw[green!5!black] (0,3) circle (0.1);
\filldraw[green!5!black] (0,-3) circle (0.1);
\end{scope}
\end{tikzpicture}
\caption{Minsum clustering.}
\label{fig:cluster:a}
\end{subfigure}
\hfill
\begin{subfigure}[b]{0.45\textwidth} 
\centering
\begin{tikzpicture}[scale=0.5]
\begin{scope}[rotate=45]

\filldraw[green!95!black, opacity=0.3] (0,-3.5) rectangle +(3.5,7);
\filldraw[blue!95!black, opacity=0.3] (0,-3.5) rectangle +(-3.5,7);

\filldraw[blue!5!black] (0,0) circle (0.1);
\filldraw[green!5!black] (1.2,0) circle (0.1);
\filldraw[blue!5!black] (-1.2,0) circle (0.1);
\filldraw[green!5!black] (0,1.2) circle (0.1);
\filldraw[blue!5!black] (0,-1.2) circle (0.1);
\filldraw[green!5!black] (0.6,0.7) circle (0.1);
\filldraw[blue!5!black] (-0.6,-0.7) circle (0.1);
\filldraw[green!5!black] (0.6,-0.7) circle (0.1);
\filldraw[blue!5!black] (-0.6,0.7) circle (0.1);
\filldraw[green!5!black] (0.4,0.3) circle (0.1);
\filldraw[green!5!black] (0.4,-0.3) circle (0.1);
\filldraw[blue!5!black] (-0.4,-0.3) circle (0.1);
\filldraw[blue!5!black] (-0.4,0.3) circle (0.1);
\filldraw[green!5!black] (0.8,0.15) circle (0.1);
\filldraw[green!5!black] (0.8,-0.15) circle (0.1);
\filldraw[blue!5!black] (-0.8,0.15) circle (0.1);
\filldraw[blue!5!black] (-0.8,-0.15) circle (0.1);
\filldraw[green!5!black] (0.15,0.8) circle (0.1);
\filldraw[green!5!black] (0.15,-0.8) circle (0.1);
\filldraw[blue!5!black] (-0.15,0.8) circle (0.1);
\filldraw[blue!5!black] (-0.15,-0.8) circle (0.1);
\filldraw[green!5!black] (2.5,0.3) circle (0.1);
\filldraw[blue!5!black] (-2.5,0.3) circle (0.1);
\filldraw[green!5!black] (2.5,-0.3) circle (0.1);
\filldraw[blue!5!black] (-2.5,-0.3) circle (0.1);
\filldraw[green!5!black] (3,0) circle (0.1);
\filldraw[blue!5!black] (-3,0) circle (0.1);
\filldraw[green!5!black] (0.3,2.5) circle (0.1);
\filldraw[green!5!black] (0.3,-2.5) circle (0.1);
\filldraw[blue!5!black] (-0.3,2.5) circle (0.1);
\filldraw[blue!5!black] (-0.3,-2.5) circle (0.1);
\filldraw[green!5!black] (0,3) circle (0.1);
\filldraw[blue!5!black] (0,-3) circle (0.1);

\draw (0,3.5)[thick] -- (0,-3.5);
\end{scope}
\end{tikzpicture}
\caption{Center-based clustering.}
\label{fig:cluster:b}
\end{subfigure}
\caption{Minsum clustering in Figure~\ref{fig:cluster:a} captures the natural structure of the dataset where centroid-based clustering in Figure~\ref{fig:cluster:b} fails.}
\label{fig:cluster}
\end{figure*}

From an axiomatic perspective, minsum clustering satisfies Kleinberg’s \emph{consistency} axiom~\cite{Kle02}, which informally requires that the optimal clustering remain unchanged when intra-cluster distances are contracted and inter-cluster distances are expanded. 
Most center-based criteria, including $k$-means and $k$-median, fail to satisfy this property~\cite{ZB09}. 
Hence, the partition induced by minsum clustering is provably more robust to perturbations of the metric and more faithful to the underlying structure of the data.

Beyond the theoretical appeal of minsum clustering, the references herein document several notable applications of the problem in domains where robustness and density-awareness are essential, further motivating the need for efficient and provably accurate algorithms for this fundamental problem.

\paragraph{Related work.}
Guttman-Beck and Hassin~\cite{GH98} pioneered approximation algorithms for metric $k$-minsum clustering. 
Behsaz et al.~\cite{BFSS15}, improving slightly the seminal result of Bartal et al.~\cite{BCR01}, showed that the general case, for arbitrary $k$, can be approximated to within a factor of $O(\log n)$. 
Guruswami and Indyk~\cite{GI03} showed that the problem is NP-hard to approximate within a constant factor very close to $1$, and Cohen-Addad et al.~\cite{CKL21} showed that it is hard to approximate within a factor of $1.415$ under the Johnson Coverage Hypothesis. 

Naderi et al.~\cite{NRS23} recently gave a quasi-polynomial time approximation scheme for low-dimensional instances, for various notions metric dimension. 
Banerjee et al.~\cite{BOR21} showed how to cluster most of the points, excluding a small constant fraction of outliers, at cost within a (large) constant factor times the optimal cost to cluster all the points. 
Hassin and Or~\cite{HO10} considered a version of the problem with exclusion penalties. 
Balcan et al.~\cite{BBG09,BB09,BL12} gave algorithms with provable guarantees for approximation-stable and perturbation-resilient instances. 
In particular, they showed that if $k = o(\log n/\log\log n)$ and the instance is approximation-stable, then there is a polynomial time algorithm to compute a partition that misclassifies a small fraction of the points, in comparison with the optimal partition. 
Voevodski et al.~\cite{VBRTX11} discussed some applications and experimental results in this vein.

The above technical discussion reflects on a loose connection between metric $k$-minsum clustering and metric $k$-median clustering. 
Most of the state-of-the-art for the latter problem can be found
in~\cite{MS84,GK98,FL11,CS17,CFS19,NCN23,CGLSS25}. 
The (non-uniform) capacitated version of metric $k$-median (see, for instance~\cite{CR05,DL16}) was used as inspiration to design metric 
$k$-minsum clustering algorithms; see~\cite{BOR21}. 
The balanced version of $k$-median was used explicitly for the same purpose; see~\cite{BFSS15} and the references therein. 
It is worth noting that Manurangsi~\cite{Man20} showed that it is gap-ETH-hard to $\left(1+\frac{2}{e}-\eps\right)$-approximate metric $k$-median in time better than $n^{o(k)}$ (see also~\cite{CGKLL19}). 
This hardness result applies to balanced $k$-median as well.\footnote{This is because the inapproximability for $k$-median in~\cite{Man20} is proven via a standard reduction from the $k$-maxcoverage problem (similar to our proof of Theorem~\ref{thm:median}). In the hard instances constructed in~\cite{Man20} for the $k$-maxcoverage problem, the completeness case solution consists of $k$ disjoint sets that cover the universe, which inherently maps to a balanced solution in the resulting $k$-median instance.} 
Therefore, the $k$-median path to fixed parameter tractable approximation schemes, let alone sublinear time ones, seems challenging, and has discouraged researchers in the past from attempting our result.\footnote{In fact, the aforementioned result for low dimensional instances~\cite{NRS23} used a generalization of balanced $k$-median, which they called min-hub clustering, to design their algorithms.} 
It is a pleasant surprise that in our particular context, $k$-minsum clustering appears to be computationally easier than its natural center-based ``relatives".\footnote{Note that, in contrast with the hardness result of~\cite{CGKLL19}, in the special case that $d$ is Euclidean distance, represented implicitly through $X\subset\RR^m$, then Feldman and Langberg~\cite{FL11} gave an $O(nm k) + m(k/\eps)^{O(1)} + 2^{\tilde O(k/\eps)}$ time approximation scheme for $k$-median (allowing medians outside $X$, in the ambient space).}

Sublinear time algorithms (as well as fixed parameter tractability) is too vast an area to do justice to within the scope of this paper. 
The seminal work of Indyk~\cite{Ind99a} deals with the more relevant context of computations on finite metric spaces. 
Esfandiari and Mitzenmacher~\cite{EM18} and Abbasi et al.~\cite{ABBCGKMSS23} are examples of recent work in this vein; see also the the references therein and~\cite{CS24}. 
Metric MAX $k$-CUT is the problem of maximizing the complement of metric $k$-minsum clustering, namely, the sum of pairwise inter-cluster distances. 
While metric MAX $k$-CUT had an $O(n^2) + nk2^{O(1/\eps^3)}$ time approximation scheme (see the aforementioned~\cite{FKKR03}), no sublinear time approximation scheme was known, except for Indyk's result~\cite{Ind99b} that applies to the special case of $k=2$ (metric MAX CUT). 
Here we adapt the aforementioned result of~\cite{Yaroslavtsev14} on dense graphs to derive a sublinear time approximation guarantee for larger $k$, as one of the building blocks of our algorithm.

It is not known if the general case of the non-metric $\ell_2^2$ variant admits any non-trivial approximation guarantee. 
Past work and results can be found in some of the above-referenced papers, as well as in~\cite{IKI94,Mat00,Sch00,ADHP09}.

\setcounter{section}{1}
\section{The Approximation Scheme}\label{sec:PTAS}
\subsection{Preliminaries}\label{sec:prelim}

Let ${\cal M} = (A,d)$ be the input finite metric space. 
It will be convenient to consider the Banach space ${\cal W}_1(A,d) = \left(\RR_0^A,\|\cdot\|_{{\cal W}_1(A,d)}\right)$, where $\RR_0^A = \{f:A\rightarrow\RR\hbox{ such that }\sum_{x\in A} f(x) = 0\}$, and the norm $\|\cdot\|_{{\cal W}_1(A,d)}$ (we will abbreviate $\|\cdot\|_{{\cal W}_1}$, when the context $(A,d)$ is clear) is the $L^1$ transportation norm whose unit ball is $\conv\left\{\frac{e_x - e_y}{d(x,y)}\colon x,y\in A \wedge x\ne y\right\}$. 
Note that the mapping $\phi:A\rightarrow\RR_0^A$ given by $\phi(x) = e_x - \frac{1}{|A|} \cdot {\mathbf 1}$, where ${\mathbf 1} = \sum_{y\in A} e_y$, is an isometry, as for $x,y\in A$ we have that $\|e_x - e_y\|_{{\cal W}_1} = d(x,y)$. 
More generally, if $x\in A$ and $\pi\in \RR^A$ is a probability distribution on $A$, then for $\phi(\pi) = \pi-\frac{1}{|A|} \cdot {\mathbf 1} = \sum_{y\in A} \pi(y)\cdot\phi(y)$, we have $\|e_x - \pi\|_{{\cal W}_1} = \sum_{y\in A} \pi(y)\cdot d(x,y)$.

\begin{lemma}\label{lm: center based}
Consider a cluster $C\subset A$. Let $\mu_C\in\RR^A$ be the uniform distribution on $C$, and let $\Delta_C = \frac{1}{|C|} \sum_{x\in C} \| e_x - \mu_C \|_{{\cal W}_1}$. 
Then
$$
\cost(C) = \sum_{x,y\in C} d(x,y) = |C|\cdot\sum_{x\in C} \| e_x - \mu_C \|_{{\cal W}_1} =
|C|^2\cdot \Delta_C.
$$
\end{lemma}

\begin{proof}
By definition, $\sum_{x\in C} \| e_x - \mu_C \|_{{\cal W}_1} = 
\sum_{x\in C}\sum_{y\in X} \mu_C(y) \cdot d(x,y) = \sum_{x\in C}\sum_{y\in C} \frac{1}{|C|}\cdot d(x,y)$,
where the last equation follows from the definition of $\mu_C$.
\end{proof}

\noindent
For $\alpha > 1$, let $\core_{\alpha}(C) = \left\{x\in C\colon \| e_x - \mu_C \|_{{\cal W}_1}\le \alpha\cdot\Delta_C\right\}$.
Note that by Markov's Inequality, $|\core_{\alpha}(C)| > \left(1 - \alpha^{-1}\right)\cdot |C|$.
%
\subsection{Overview}\label{sec:overview}

Let ${\cal C} = \{C_1,C_2,\dots,C_k\}$ denote an optimal clustering of $A$. For convenience, we assume 
that the clusters are sorted by non-increasing size, that is $|C_1|\geq |C_2|\geq \ldots |C_k|$. Throughout 
this section, we will assume that we can enumerate efficiently over estimates $\widehat{|C_i|}$, for all 
$i=1,2,\dots, k$ of the optimal cluster cardinalities $|C_i|$, and similar estimates $\widehat{\Delta_{C_i}}$ 
of the optimal cluster cost parameters $\Delta_{C_i}$, roughly up to factors $1+\gamma$ for $\gamma\in (0,1)$ 
(we'll use $\gamma^{-1} = \poly(k,\varepsilon^{-1})$). The former can be obtained by enumerating 
for each cluster over a value of the form $\lceil(1+\gamma)^j\rceil$ in the range $\{1,2,\dots n\}$. There are 
at most $O(\gamma^{-1} \log n)$ such values. The latter can be obtained, for example, by computing a constant 
factor approximation to $k$-median in time $2^{O(k)}\cdot n$ via~\cite{Ind99a}. As the optimal $k$-minsum 
cost $\OPT$ is lower bounded by the value of an optimal $k$-median clustering and is upper bounded by $2n$ 
times the value of an optimal $k$-median clustering, we can then enumerate over all possible values of the 
cluster costs in a similar manner as the enumeration over cluster cardinalities. We do not have a lower bound 
for cluster costs. However, the cost of clusters contributing less than 
$\frac{\varepsilon}{k}\cdot \OPT$ to the total cost can be rounded up to 
$\frac{\varepsilon}{k}\cdot \OPT$, without affecting the total cost significantly. Therefore, we only need to
consider $O(\gamma^{-1} \log(nk/\varepsilon))$ cost values per cluster.

Our algorithm works in four stages that we now outline. The first stage produces a relatively small
collection of candidate sets of cluster centers to check. At least one of the candidates is guaranteed
to enable a good approximation. The second and third stage compute, given candidate centers
and approximate cluster sizes and costs (also enumerated upon), a clustering. If the input parameters
are good approximations, the computed clustering is guaranteed to be a good approximation. The
last stage approximates quickly the cost of each computed clustering, enabling the selection of
the (nearly) best result.

\paragraph{Stage 1: center approximation.}
In this stage, we determine an approximate center $\widehat{\mu_C}$ for every optimal 
cluster $C\in {\cal C}$. To formalize the guarantee of this stage, we require the following 
definition.
\begin{definition}\label{def:cover}
Let $M\subset A$ be a set of points. A cluster $C$ is ($\varepsilon,\alpha$)-covered (or
simply covered, if the parameters are obvious from the context) by $M$ if there exists a
 point $p\in M$ such that 
\begin{equation}\label{eq:cover}
|C|\cdot \sum_{x\in C}\|e_x-e_p\|_{{\cal W}_1}\leq 
    \frac{\varepsilon}{4\alpha}\cdot \frac{\OPT}{k} + \alpha\cdot\Cost(C).
\end{equation}
A point $p\in M$ that satisfies Equation~\eqref{eq:cover} is called an 
$(\varepsilon,\alpha)$-approximate center of $C$.
\end{definition}

\begin{claim}\label{cl:approx dist}
If $p$ is an $(\varepsilon,\alpha)$-approximate center of $C$, then
$$
\|e_p-\mu_C\|_{{\cal W}_1}\le 
\frac{\varepsilon}{4\alpha}\cdot \frac{\OPT}{k\cdot |C|^2} + \alpha\cdot\Delta_C.
$$
\end{claim}

\begin{proof}
Note that $\|e_p-\mu_C\|_{{\cal W}_1} = \frac{1}{|C|}\cdot\sum_{x\in C} \|e_p-e_x\|_{{\cal W}_1}$,
and use Equation~\eqref{eq:cover}.
\end{proof}

Notice that if $C$ is $(\varepsilon,\alpha)$-covered by $M$, then for $\alpha'\ge \alpha$ and
$\varepsilon'\ge \varepsilon\cdot\frac{\alpha'}{\alpha}$, it is also $(\varepsilon', \alpha')$-covered by $M$. 
Stage 1 realizes the algorithm of the following theorem.\footnote{Note that Theorem~\ref{thm:approxcenters}
is stated with respect to any clustering, not just an optimal clustering $\mathcal{C}$. This will become 
important at Stage 3 of the algorithm, where we shall need to invoke the guarantee of 
Theorem~\ref{thm:approxcenters} for a non-optimal clustering.}
\begin{theorem}\label{thm:approxcenters}
There exist absolute constants $\alpha, \eta > 0$ such that for every $\delta, \varepsilon > 0$, there
is a randomized $\left(n^{1+o(1)}+2^{O(k\log k)}\right)\cdot ((k/\varepsilon)\cdot\log(1/\delta))^{O(1)}$-time
algorithm that outputs a collection of at most $(n/\varepsilon)^{o(1)}+2^{O(k\log k)}$ sets, each containing 
at most $\eta\cdot (k^2/\varepsilon) \cdot \log(k/\delta)$ points, such that for any partition of $A$ into $k$
clusters, with probability at least $1-\delta$, at least one of these sets ($\varepsilon,\alpha$)-covers all 
clusters.
\end{theorem}

The following stages are applied to each candidate set, to each possible choice of $k$
points in the set as approximate centers $\widehat{\mu_{C_i}}$, and to each choice of
approximate cluster sizes $\widehat{|C_i|}$, approximate cluster costs $\widehat{\Delta_{C_i}}$,
and approximate optimal cost $\widehat{\OPT}$. Note that the total number of options to 
enumerate over is $(n/\varepsilon)^{o(1)}\cdot \gamma^{-2}\cdot 2^{O(k\log k)}$ (see 
the definitions and details below).

\paragraph{Stage 2: assignment of far points.}
In this stage, we assume to have, for every cluster $C\in {\cal C}$, an approximate 
center $\widehat{\mu_C}$, an approximation $\widehat{|C|}$ of $|C|$, an 
approximation $\widehat{\Delta_C}$ of $\Delta_C$, and an approximation $\widehat{\OPT}$
of $\OPT$. We require the following definitions.
\begin{definition}\label{def:card}
We say that $\widehat{|C|}$ $\gamma$-approximates $|C|$ iff
\begin{equation}\label{eq:card}
|C|\leq \widehat{|C|}\leq (1+\gamma)\cdot |C|.
\end{equation}
Similarly, $\widehat{\OPT}$ $\gamma$-approximates $\OPT$ iff
\begin{equation}\label{eq:opt}
\OPT\le \widehat{\OPT}\le (1+\gamma)\cdot \OPT.
\end{equation}
\end{definition}
\begin{definition}\label{def:size}
We say that $\widehat{\Delta_{C}}$ $\gamma$-approximates $\Delta_C$ iff
\begin{equation}\label{eq:size}
\Delta_{C}\leq \widehat{\Delta_{C}}\leq (1+\gamma)\cdot \Delta_{C}\ \ \vee\ \ 
\Delta_{C}\le \gamma\cdot \frac{\OPT}{k\cdot |C|^2}\le \widehat{\Delta_{C}}\le 
(1+\gamma)\cdot\gamma\cdot\frac{\OPT}{k\cdot |C|^2}.
\end{equation}
\end{definition}
\begin{definition}\label{def:far}
For $C\in {\cal C}$ and $\beta\in (0,1)$, given $\varepsilon$, $\alpha$, $\gamma$, 
$\widehat{\mu_C}$, $\widehat{\Delta_C}$, and $\widehat{\OPT}$, we say that a point 
$p$ is $\beta^{-1}$-far for $C$ iff
\begin{equation}\label{eq:far}
\|e_p-\widehat{\mu_C}\|_{{\cal W}_1}\geq \beta^{-1}\cdot \left(\frac{\varepsilon\cdot (1+\gamma)^2}{4\alpha}\cdot \frac{\widehat{\OPT}}{k\cdot \left(\widehat{|C|}\right)^2} + \alpha\cdot\widehat{\Delta_C}\right).
\end{equation}
A point $p$ is $\beta^{-1}$-close for $C$ iff it is not $\beta^{-1}$-far for $C$.
\end{definition}
\begin{hypothesis}\label{hyp:approx}
We have, for every $C\in {\cal C}$, an approximate center $\widehat{\mu_{C}}$ that satisfies
Equation~\eqref{eq:cover}, an approximation $\widehat{|C|}$ of $|C|$ that satisfies 
Equation~\eqref{eq:card}, and an approximation $\widehat{\Delta_{C}}$ of $\Delta_C$ that
satisfies Equation~\eqref{eq:size}. We also have an approximate value $\widehat{\OPT}$ that
satisfies Equation~\eqref{eq:opt}.
\end{hypothesis}
\begin{definition}\label{def:partial clustering}
A set $\mathcal{P} = \{P_1,P_2,\dots,P_k\}$ is called a {\em partial clustering} of $A$
iff $P_1,P_2,\dots,P_k$ are mutually disjoint, and $P_1\sqcup P_2\sqcup\cdots\sqcup P_k\subseteq A$.
\end{definition}
Stage 2 computes an assignment of far points to clusters, providing the guarantee of
the following theorem.
\begin{theorem}\label{thm:farassignment}
Let $\gamma^{-1} = \poly(k,\eps^{-1})$ (to be specified later).
There exists an algorithm that computes $n^{o(1)}\cdot 2^{2^{O(k\log(k/\varepsilon))}}$ partial 
clustering such that, assuming Hypothesis~\ref{hyp:approx}, at least one of the partial 
assignments satisfies that it can be extended to a clustering ${\cal C}'$ for which the following
assertions hold: 
($i$) All $\gamma^{-1}$-far points for all $C\in {\cal C}'$ have been assigned in the partial
      clustering; and,
($ii$) the cost of ${\cal C}'$ is at most $\left(1+\frac{\varepsilon}{2k}\right)\cdot \OPT$.
The algorithm runs in time $n^{1+o(1)}\cdot 2^{2^{O(k\log(k/\varepsilon))}}$ .
\end{theorem}

In the next stage, we need approximate centers, cardinalities, and cluster costs for the
clustering $\mathcal{C}'$ rather than for $\mathcal{C}$. This requires a further enumeration
over $(n/\varepsilon)^{o(1)}\cdot \gamma^{-2}\cdot 2^{O(k\log k)}$ options. Keep in mind
that we do not need to invoke the algorithm in Theorem~\ref{thm:approxcenters} more than
twice.

\paragraph{Stage 3: assignment of close points.}
This stage completes the partition of the input dataset into $k$ clusters. Assuming for
$\mathcal{C}'$ the same assumptions that are specified in Stage 2 for $\mathcal{C}$,
Stage 3 computes a near-optimal partition. The implementation of this stage requires 
grouping ``neighboring" clusters, defined roughly as clusters for which their sets of 
$\gamma^{-1}$-close points intersect. The neighbors binary relation defines a graph 
over the node set ${\cal C}'$. Stage 3 first assigns points to each connected component 
of this graph. It then applies in each connected component with more than one cluster a 
sublinear time approximation scheme for metric MAX $h$-CUT which we design and 
analyze (here $h$ is the number of clusters in the connected component). The outcome 
of Stage 3 is summarized in the following theorem.
\begin{theorem}\label{thm:finalguarantee}
There exists a randomized algorithm that, given a partial assignment satisfying the conclusion 
of Theorem~\ref{thm:farassignment} and a failure probability $\delta\in (0,1)$, and assuming 
Hypothesis~\ref{hyp:approx} for $\mathcal{C}'$, computes a $(1+\varepsilon)$-approximate 
clustering in time 
$n\cdot 2^{O(k\log(k/\varepsilon))}\cdot \log(n/\delta) + 2^{2^{O(k\log(k/\varepsilon))}}$.
\end{theorem}

\paragraph{Stage 4: clustering cost estimation.}
In this stage we evaluate the cost of the clustering computed by Stage 3, to enable
selecting the (nearly) best partition in the enumeration over the choice of parameters
feeding the loop over Stages 2 and 3. Computing the cost of a clustering exactly is
trivial in linear time, but we require a good sub-linear time approximation. Our
guarantee in Stage 4 is the following theorem.

\begin{theorem}\label{thm:costguarantee}
There exists a randomized algorithm with the following guarantee. Given explicitly 
an $n$ point metric space $(X,d)$, failure probability $\delta\in (0,1)$, approximation
guarantee $\varepsilon\in (0,1)$, the algorithm, using $O(\varepsilon^{-2}\cdot n\cdot \log(1/\delta))$ 
operations, outputs an estimate $\widehat{\Cost(X)}$ that satisfies with probability at least $1-\delta$
\[(1-\varepsilon)\cdot \Cost(X) \leq\widehat{\Cost(X)} \leq (1+\varepsilon)\cdot \Cost(X).\]
\end{theorem}

In what follows, we present and analyze efficient algorithms for all four steps, thus proving
Theorems~\ref{thm:approxcenters},~\ref{thm:farassignment},~\ref{thm:finalguarantee},
and~\ref{thm:costguarantee}. Before doing that, we put together the four ingredients to 
derive Theorem~\ref{thm: main}.
\medskip

\begin{proofof}{Theorem~\ref{thm: main}}
Given the input dataset $(A,d)$ and the desired approximation parameter $\varepsilon$,
we first run the center approximation algorithm of Theorem~\ref{thm:approxcenters} twice.
One run is intended to approximate an optimal clustering $\mathcal{C}$, and the other
run is intended to approximate a near-optimal clustering $\mathcal{C}'$, an extension
of the partial clustering computed by the algorithm of Theorem~\ref{thm:farassignment}.
We shall run the latter algorithm several times, but we need to guarantee good approximate
values only for a single partial clustering that will eventually yield a near-optimal clustering.
Poor approximations will yield some clustering, but we shall evaluate the cost of all the
clusterings we generate, and we shall output the best of them.

After invoking Theorem~\ref{thm:approxcenters} twice, we enumerate over the 
$(n/\varepsilon)^{o(1)}\cdot 2^{O(k\log k+\log(k/\varepsilon))}$ choices of approximate 
centers, approximate cluster sizes, and approximate cluster costs for both $\mathcal{C}$ 
and $\mathcal{C}'$, and approximate optimal value. The rest of the algorithm is executed 
for each such choice, each yielding a set of candidate clusterings, each endowed with an 
approximate cost. We output a clustering of minimum approximate cost, among all the 
candidates.\footnote{Note that not all steps need to be executed for every choice of
approximate values for both $\mathcal{C}$ and $\mathcal{C}'$; however the double
enumeration does not affect the asymptotic bounds as stated here, so in order to simplify 
the presentation, we ignore this possible improvement.}

In the outer enumeration loop, we first invoke the algorithm assigning far points
of Theorem~\ref{thm:farassignment} (this is where we need the approximate values
for $\mathcal{C}$ to get the guarantees stated in the theorem). We get
$n^{o(1)}\cdot 2^{2^{O(k\log(k/\varepsilon))}}$ candidate partial assignments that we 
need to enumerate over in an inner loop. Thus, the next two steps will be executed over 
a total of $(n/\varepsilon)^{o(1)}\cdot 2^{2^{O(k\log(k/\varepsilon))}}$ candidate partial 
clusterings (taking into account the enumeration over approximate values, and the
enumeration over the resulting list of partial clusterings). For each of those, we execute 
the algorithm of Theorem~\ref{thm:finalguarantee} that completes the clustering (this is
where we need the approximate values for $\mathcal{C}'$, for the correct partial clustering,
to achieve the guarantees of the theorem; we can take $1-\delta$ to be a constant as we
only need it to succeed for the one run on the right choice of partial clustering). Finally, we 
compute the cost of each clustering thus computed using the algorithm of 
Theorem~\ref{thm:costguarantee}. Here, we need success in all choices (otherwise
we may underestimate the cost of a poor clustering), so we need to cap the failure
probability at $\delta' \le \frac{\delta}{(n/\varepsilon)^{o(1)}\cdot 2^{2^{O(k\log(k/\varepsilon))}}}$.

The algorithm's correctness follows immediately from the correctness of the four 
algorithms that compose it. The total running time, for a fixed target failure 
probability $\delta$, is as follows. Theorem~\ref{thm:approxcenters} is invoked
twice. Theorem~\ref{thm:farassignment} is invoked 
$(n/\varepsilon)^{o(1)}\cdot 2^{O(k\log k+\log(k/\varepsilon))}$ times (in fact, we
only need to invoke it $(n/\varepsilon)^{o(1)} + 2^{O(k\log k+\log(k/\varepsilon))}$
times, as it doesn't require the approximate values for $\mathcal{C}'$).
Theorems~\ref{thm:finalguarantee} and~\ref{thm:costguarantee} are invoked
$(n/\varepsilon)^{o(1)}\cdot 2^{2^{O(k\log(k/\varepsilon))}}$ times, where in
Theorem~\ref{thm:costguarantee} we require failure probability
$O((n/\varepsilon)^{-o(1)}\cdot 2^{-2^{O(k\log(k/\varepsilon))}})$. This gives 
(assuming we're aiming for fixed failure probability $\delta\in (0,1)$) a total 
running time of:
$$
\left(n^{1+o(1)}+2^{O(k\log k)}\right)\cdot (k/\varepsilon)^{O(1)} +
$$
$$
\left((n/\varepsilon)^{o(1)} + 2^{O(k\log k+\log(k/\varepsilon))}\right)\cdot 
n^{1+o(1)}\cdot 2^{2^{O(k\log(k/\varepsilon))}} +
$$
$$
(n/\varepsilon)^{o(1)}\cdot 2^{2^{O(k\log(k/\varepsilon))}}\cdot 
\left((n\log n)\cdot 2^{O(k\log(k/\varepsilon))} + 2^{2^{O(k\log(k/\varepsilon))}}\right) +
$$
$$
(n/\varepsilon)^{o(1)}\cdot 2^{2^{O(k\log(k/\varepsilon))}}\cdot 
\varepsilon^{-2}\cdot n\cdot \left(o(\log(n/\varepsilon)) + 2^{O(k\log(k/\varepsilon))}\right).
$$
The running time is clearly dominated by $n^{1+o(1)}\cdot 2^{2^{O(k\log(k/\varepsilon))}}$.
\end{proofof}


\subsection{Center Approximation}
\label{sec:center:approx}

This section is devoted to the proof of Theorem~\ref{thm:approxcenters}. We first define
the algorithm, and then we analyze it. In this section, to simplify the notation, we use
$\mathcal{C}$ to denote an arbitrary target clustering, not necessarily an optimal one.

\paragraph{The algorithm.}
We construct a rooted tree. The tree has depth $k-1$ and is $10\log(nk/\varepsilon)$-ary, meaning that 
we have exactly $\sum_{t=0}^{k-1} \left(10\log(nk/\varepsilon)\right)^t$ nodes. The label $v$ of a node
is a sequence that encodes the path to it from the root, where the root is labeled $[0]$. Thus, the label $v$
of a depth $t$ node is a sequence of length $t+1$, where the first entry is $0$, and the other entries are
all in $\{1,2,\dots, 10\log(nk/\varepsilon)\}$. A node $v$ is associated with a set of points $M_v$. The
root is associated with $M_{[0]}$, initialized with a $k$-median approximation plus a uniformly random
sample (see the details below). Given $M_v$ and $i\in \{1,\ldots ,10\log(nk/\varepsilon)\}$, we construct 
$M_{[v,i]}$ for its child $[v,i]$ by sampling a set of points $S$ according to a distribution that depends on 
$M_v$ and $i$, then setting $M_{[v,i]} = S\cup M_v$. The sampled points are used as proxies for cores of 
the clusters of an optimal solution. Our goal is to show that there exists at least one leaf $M_v$, with 
$v\in [0,\{1,\ldots ,10\log(nk/\varepsilon)\}^{k-1}]$, such that the core of every cluster of an optimal 
solution is in some sense well approximated by $M_v$. Ideally for cluster $C_i$, we would like to show that 
there exists a point $p\in M_v$ such that $\|e_p-\mu_{C_i}\|_{{\cal W}_i}\in O(1)\cdot \Delta_{C_i}$, but this 
is too strong a property to enforce, hence the weaker guarantee in Definition~\ref{def:cover}. The following
pseudocode defines the algorithm.
\begin{algorithm}[H]
\begin{algorithmic}[1]
        \State Initialize tree with $[0]$ and root node  represented by points $M_{[0]} = \emptyset$
        \State Compute a $k$-median $O(1)$-approximation and add centers to $M_{[0]}$
        \State Sample $m\geq \eta\cdot (k/\varepsilon)\cdot \log(k/\delta)$ i.i.d. point u.a.r. on $A$ 
                   and add them to $M_{[0]}$
        \While{$\exists$ a node $v$ at level $< k$ with no children}
            \For{$i=1$ to $10\log(nk/\varepsilon)$} 
                \State Add all the points in $M_{v}$ to $M_{[v,i]}$
                \State Sample $\eta \cdot (k/\varepsilon)\cdot \log(k/\delta)$ points i.i.d. from $\mathbb{P}_{i,M_{v}}$ 
                           and add them to $M_{[v,i]}$
            \EndFor
        \EndWhile
    \end{algorithmic}
    \caption{Core Generation via Pruned Distance Sampling}
\end{algorithm}

We now elaborate:
\begin{description}
\item[Initialization:] We initialize the points in the root node via uniform sampling, as well as a 
    constant factor approximation to the $k$-median problem. The probability that we sample points 
    from a core of a cluster of size $n/k$ is, at this stage of the algorithm, large.
\item[Adaptive Sampling:] Consider a node of the tree labeled with a non-empty set of candidate 
    core proxies $M$. For any point $q$, denote $m_q = \arg\min_{m\in M} d(q,m)$.
    For $i\in \{0\}\cup\NN$, let $\Gamma_i = 2^{-i}\cdot\sum_{q\in A} d(q,m_q)$,
    and denote $A_{i,M} = \{q\in A\colon d(q,m_q)\le \Gamma_i\}$.  
    Let $\mathbb{P}_{i,M}$ denote the probability distribution on $A_{i,M}$ induced by setting, for 
    each $p\in A_{i,M}$ (assuming $A_{i,M}\ne\emptyset$), 
    $$
    \mathbb{P}_{i,M}[p]=\frac{d(p,m_p)}{\sum_{q\in A_{i,M}} d(q,m_q)}.
    $$ 
    We will write $\mathbb{P}_i$ if $M$ is clear from the context.
    For each $i\le 10\log(nk/\varepsilon)$, we sample a sufficient (polynomial in $k$ and $\varepsilon$, 
    but independent of $n$) number of points i.i.d. from $\mathbb{P}_{i}$. Let $S$ denote the set of sampled 
    points. 
\end{description}

We shall now build our way towards a proof of Theorem~\ref{thm:approxcenters}.

\paragraph{Approximate covers via uniform sampling.}
For the initialization of the root, namely $M_{[0]}$, we only have to analyze uniform sampling:
\begin{lemma}
\label{lem:uniform2}
Let $M$ be a set of $m$ points sampled independently and uniformly at random with replacement. Then if 
$m\geq 2k\cdot \ln(k/\delta)$, we have with probability at least $1-\frac{\delta}{k}$ that $M$ $(0,3)$-covers 
$C_1$.
\end{lemma}

\begin{proof}
Recall that $C_1$ is a largest cluster, so $|C_1|\ge \frac{n}{k}$. Thus, the probability that a point is sampled 
from $C_1$ is at least $\frac 1 k$, and the probability that a point is sampled from the $C'=\core_{2}(C_1)$ 
is at least $\frac{1}{2k}$.

We have $M\cap C' \neq \emptyset$ with probability at least $1-(1-1/(2k))^m$, which by our choice of $m$ 
is at least $1-\delta/k$. By definition, for any point $p\in M\cap C'$ it holds that 
$\| \mu_{C_1} - e_p \|_{{\cal W}_1}\le 2\Delta_{C_1}$. Therefore, conditioned on $M\cap C'\ne\emptyset$,
there exists $p\in M\cap C'$, and $p$ satisfies
\begin{align*}
            |C_1|\cdot \sum_{x\in C_1} \| e_x - e_p\|_{{\cal W}_1} 
&\leq |C_1|\cdot \sum_{x\in C_1}\left(\|e_x - \mu_{C_1}\|_{{\cal W}_1} + \| \mu_{C_1} - e_p \|_{{\cal W}_1}\right) \\
&\leq  \Cost(C_1)   + |C_1|^2 \cdot 2 \Delta_{C_1} = 3 \Cost(C_1).\qedhere
\end{align*}
\end{proof}

\paragraph{Approximate covers via non-uniform sampling.}
In this section, we consider the subsequent nodes of the tree, where the points are sampled 
proportionate to distances. The following two lemmas prove that the probability of sampling 
points from $\core_{2}(C_j)$ of the largest uncovered cluster $C_j$ is large.
\begin{lemma}\label{lem:samplehit}
There exists a small absolute constant $\zeta > 0$ such that the following holds. Let $M$ be a 
set of points, and let $\varepsilon > 0$ be sufficiently small ($\varepsilon < 1$ definitely suffices). 
Suppose that there exists $j\in \{1,2,\dots, k\}$ such that $C_j$ is not $(\varepsilon,5)$-covered 
by $M$, and consider the smallest such $j$ (a largest uncovered cluster). Then, for every set 
$A_{i,M}$ containing $\core_{2}(C_j)$, we have
$$
\mathbb{P}_i\left[p\in \core_{2}(C_j)\colon p\in C_{< j}\cup C_j\right] \geq \frac{\varepsilon\cdot\zeta}{k},
$$
where $C_{< j}\triangleq \bigcup_{h < j} C_h$.
\end{lemma}

\begin{proof}
Throughout this proof, we use $\beta = \frac{\varepsilon\cdot\zeta}{k}$ to simplify the notation.
Consider any $i$ such that all points in $\core_{2}(C_j)$ are in the support of $\mathbb{P}_i$.
Assume for contradiction that $M\cap \core_{2}(C_j) = \emptyset$
(otherwise, by definition, already $M$ covers $C_j$) and 
$\mathbb{P}_i\left[p\in \core_{2}(C_j)\colon p\in C_{<j}\cup C_j\right] < \beta$. 
We have:
\begin{eqnarray}
\nonumber
& & \sum_{p\in \core_{2}(C_j)} d(p,m_p) \\
\nonumber
    & < & \beta\cdot \left(\sum_{p\in C_{<j}} d(p,m_p) + \sum_{p\in C_j} d(p,m_p)\right) \\
\nonumber
    &\leq & \beta  \cdot \left(\sum_{h<j}\left(\frac{\varepsilon}{20k\cdot |C_h|}\cdot \OPT + 5\cdot|C_h|\cdot \Delta_{C_h}\right) + 
                 \sum_{p\in C_j} d(p,m_p)\right) \\ 
\label{eq: core}
    &\leq & \beta  \cdot \left(\sum_{h<j}\left(\frac{\varepsilon}{20k\cdot |C_h|}\cdot \OPT + 5\cdot |C_h|\cdot \Delta_{C_h}\right) + 
                |C_j|\cdot \Delta_{C_j} + 2\sum_{p\in \core_{2}(C_j)} d(p,m_p)\right),
\end{eqnarray}
where the final inequality is explained as follows. Define, for every $p\in C_j$ a point $q_p\in \core_{2}(C_j)$
as follows. If $p\in \core_{2}(C_j)$, then $q_p = p$. Otherwise, for the points in $C_j\setminus \core_{2}(C_j)$,
the $q_p$-s form a matching. This is possible as $|\core_{2}(C_j)|\geq \frac{1}{2}\cdot |C_j|$. Now, by
the minimality of $m_p$ and the triangle inequality, $d(p,m_p) \leq d(p,m_{q_p})\leq d(p,q_p) + d(q_p,m_{q_p})$. 
Using the triangle inequality again, $d(p,q_p)\le \|e_p-\mu_{C_j}\|_{{\cal W}_1} + \|\mu_{C_j}-e_{q_p}\|_{{\cal W}_1}$.
We have
\begin{eqnarray*}
\sum_{p\in C_j} d(p,q_p) & = & \sum_{p\in C_j\setminus \core_{2}(C_j)} d(p,q_p) \\
& \le & \sum_{p\in C_j\setminus \core_{2}(C_j)} \left(\|e_p-\mu_{C_j}\|_{{\cal W}_1} + \|\mu_{C_j}-e_{q_p}\|_{{\cal W}_1}\right) \\
& \le & \sum_{p\in C_j} \|e_p-\mu_{C_j}\|_{{\cal W}_1} \\
& = & |C_j|\cdot \Delta_{C_j},
\end{eqnarray*}
and,
\begin{eqnarray*}
\sum_{p\in C_j} d(q_p,m_{q_p}) & = & 
\sum_{p\in\core_2(C_j)} d(p,m_p) + \sum_{p\in C_j\setminus\core_2(C_j)} d(q_p,m_{q_p}) \\
&\le& 2\cdot \sum_{p\in \core_{2}(C_j)} d(p,m_p).
\end{eqnarray*}

Now, given that $C_j$ is not $(\varepsilon,5)$-covered by $M$ (and therefore also not
$(0,5)$-covered by $M$), we have that for every
$q\in\core_2(C_j)$ and for every $m\in M$,
\begin{eqnarray*}
5\cdot\Delta_{C_j} & < & \frac{1}{|C_j|}\sum_{p\in C_j} d(p,m) \\
&\le& d(q,m) + \frac{1}{|C_j|}\sum_{p\in C_j} d(p,q) \\
&\le& d(q,m) + \frac{1}{|C_j|}\sum_{p\in C_j} \left(\|e_p-\mu_{C_j}\|_{{\cal W}_1} + \|e_q-\mu_{C_j}\|_{{\cal W}_1}\right) \\
&\le& d(q,m) + \Delta_{C_j} + 2\cdot\Delta_{C_j}.
\end{eqnarray*}
Therefore, $2\cdot\Delta_{C_j} < d(q,m)$. Hence,
\begin{equation}\label{eq:delta-core}
|C_j|\cdot \Delta_{C_j}\le 2\cdot |\core_2(C_j)|\cdot\Delta_{C_j} < \sum_{p\in \core_{2}(C_j)} d(p,m_p).
\end{equation}

Rearranging the terms in Equation~\eqref{eq: core}, we get
$$
    (1-3\beta)\cdot \sum_{p\in \core_{2}(C_j)} d(p,m_p) \le 
    \beta  \cdot \sum_{h<j}\left(\frac{\varepsilon}{20k\cdot |C_h|}\cdot \OPT + 5\cdot |C_h|\cdot \Delta_{C_h} \right).
$$
This implies, as $|C_j|\le |C_h|$ for all $h < j$ and $j\le k$,
\begin{eqnarray}
              |C_j|\cdot \sum_{p\in \core_{2}(C_j)} d(p,m_p) 
&\leq &  \frac{\beta}{1-3\beta}  \cdot 
              \sum_{h<j}\left(\frac{\varepsilon}{20k}\cdot \frac{|C_j|}{|C_h|}\cdot \OPT + 
              5\cdot |C_j|\cdot |C_h|\cdot \Delta_{C_h} \right)\nonumber \\
&\leq & \frac{\beta}{1-3\beta}  \cdot \frac{\varepsilon}{20}\cdot \OPT + 5\cdot \sum_{h<j} |C_h|^2 \cdot \Delta_{C_h} \nonumber\\
&\leq & \frac{\beta}{1-3\beta}  \cdot 6\cdot \OPT \nonumber \\
&\leq & 12\cdot\beta\cdot  \OPT. \label{eq:sum-dist-core}
\end{eqnarray}

Let $m' = \arg\min_{m\in M} \sum_{p\in C_j} d(p,m)$. Note that for every point $q$,
$$
\sum_{p\in C_j} d(p,m')\le \sum_{p\in C_j} d(p,m_q)\le \sum_{p\in C_j} (d(p,q) + d(q,m_q)).
$$
Therefore,  we have
\begin{eqnarray*}
|C_j|\cdot \sum_{p\in C_j} d(p,m')
   &\leq& |C_j|\cdot \sum_{p\in C_j} \frac{1}{|\core_{2}(C_j)|}\cdot \sum_{q\in \core_{2}(C_j)} (d(p,q) + d(q,m_q)) \\
   &\leq& 2\cdot \sum_{p\in C_j} \sum_{q\in \core_{2}(C_j)} d(p,q) + 2\cdot \sum_{p\in C_j} \sum_{q\in \core_{2}(C_j)} d(q,m_q) \\
   &\leq& 4\cdot \cost(C_j) + 2\cdot |C_j|\cdot \sum_{q\in \core_{2}(C_j)} d(q,m_q)  \\
   & <  & \frac{1}{5}\cdot |C_j|\cdot \sum_{p\in C_j} d(p,m') + 2\cdot |C_j|\cdot \sum_{q\in \core_{2}(C_j)} d(q,m_q),
\end{eqnarray*}
where the last inequality uses the fact that $M$ does not $(0,5)$-cover $C_j$. This implies that
\begin{equation}\label{eq:covering}
|C_j|\cdot \sum_{p\in C_j} d(p,m') < \frac{5}{2}\cdot |C_j|\cdot \sum_{q\in \core_{2}(C_j)} d(q,m_q),
\end{equation}
Therefore, using Equation~\eqref{eq:sum-dist-core} and taking $\zeta\le \frac{1}{600}$,
$$
|C_j|\cdot \sum_{p\in C_j} d(p,m') < \frac{5}{2}\cdot |C_j|\cdot \sum_{q\in \core_{2}(C_j)} d(q,m_q) 
\leq 30\cdot\beta\cdot  \OPT
\leq \frac{\varepsilon}{20k}\cdot \OPT,
$$
in contradiction with the assumption that $M$ does not $(\varepsilon,5)$-cover $C_j$.
\end{proof}

Our next step is to show that sampling from one of the distributions $\mathbb{P}_i$ unconditionally
picks points from the core of the largest uncovered cluster $C_j$ with substantial probability.
\begin{lemma}\label{lem:sample}
Using the same notation and assumptions of Lemma~\ref{lem:samplehit}, the following assertion
holds. Let $c\le n^9$ be a constant.
If $\sum_{p\in A} d(p,m_p) \leq c\cdot \OPT$, then there exists an $i\le 10\log(nk/\varepsilon)$ 
such that sampling a point $p$ from $\mathbb{P}_i$ satisfies
$$
\mathbb{P}_i\left[p\in \core_{2}(C_j)\right] \geq \frac{\varepsilon\cdot\zeta}{2 k}.
$$
\end{lemma}

\begin{proof}
Let $i$ be the largest integer such that $\core_{2}(C_j)\subset A_{i,M}$. Note that if we assume for
contradiction that $i>10\log(nk/\varepsilon)$, 
then Equation~\eqref{eq:covering} (which is derived using only the assumptions of Lemma~\ref{lem:samplehit})
gives
\begin{eqnarray*}
          \sum_{p\in C_j} d(p,m') 
& < & \frac{5}{2}\cdot |C_j|\cdot \sum_{q\in \core_{2}(C_j)} d(q,m_q) \\
& < & \frac{5}{2}\cdot |C_j|\cdot \left(\frac{\varepsilon}{nk}\right)^{10} \cdot \sum_{q\in A} d(q,m_q) \\
&\le& \frac{5}{2}\cdot n\cdot\left(\frac{\varepsilon}{nk}\right)^{10} \cdot n^2 \cdot\OPT \\
&\le& \frac{\varepsilon}{20 k} \cdot \OPT,
\end{eqnarray*}
implying that $C_j$ is $(\varepsilon,5)$-covered by $M$, a contradiction.

Let $p^* = \arg\max_{p\in\core_{2}(C_j)} d(p,m_p)$. Then we must have $d(p^*,m_{p^*}) > \frac 1 2\Gamma_i$,
otherwise $\core_2(C_j)\subset A_{(i+1),M}$, in contradiction to the maximality of $i$. Thus, for any point $q\in A_{i,M}$, 
we have 
$d(q,m_q) \leq \Gamma_i < 2\cdot d(p^*,m_{p^*})$. Therefore, using $|C_j|\leq 2|\core_{2}(C_j)|$,
\begin{eqnarray}
            \sum_{q\in C_{j}\cap A_{i,M}} d(q,m_q) 
&\leq& 4\cdot |\core_{2}(C_j)|\cdot d(p^*,m_{p^*}) \nonumber \\
&\leq& 4\cdot \sum_{q\in \core_{2}(C_j)} (d(p^*,q) + d(q,m_q)) \nonumber \\
&\leq& 4\cdot \sum_{q\in \core_{2}(C_j)} (2\Delta_{C_j} + d(q,m_q)) \nonumber \\
&  =  & 8\cdot |\core_2(C_j)|\cdot\Delta_{C_j} + 4\cdot \sum_{q\in \core_{2}(C_j)} d(q,m_q) \nonumber \\
&\leq& 8\cdot \sum_{q\in \core_{2}(C_j)} d(q,m_q), \label{eq:sample}
\end{eqnarray}
where the last inequality uses Equation~\eqref{eq:delta-core}.

Define $C_{\le j} = \bigcup_{h\leq j} C_j$ and $C_{\ge j} = \bigcup_{h\geq j}C_j$. By Lemma~\ref{lem:samplehit},
\begin{equation}
\label{eq:samplehit2}
    \mathbb{P}_i[p\in \core_2(C_j)\colon p\in C_{\le j} \cap A_{i,M}] \geq \frac{\varepsilon}{600 k},
\end{equation}
and Equation~\eqref{eq:sample} implies
\begin{equation}
\label{eq:samplehit3}
    \mathbb{P}_i[p\in \core_2(C_j)\colon p\in C_{\ge j} \cap A_{i,M}] \geq \frac{1}{8 k}.
\end{equation}
As 
$$
\max\left\{\mathbb{P}_i[p\in C_{\le j} \cap A_{i,M}],\mathbb{P}_i[p\in C_{\ge j} \cap A_{i,M}]\right\} \geq \frac{1}{2},
$$
we have by the law of total probability
\begin{eqnarray*}
           \mathbb{P}_i[p\in \core_2(C_j)] 
&\ge& \max_{C\in \{C_{\le j},C_{\ge j}\}} \mathbb{P}_i[p\in \core_2(C_j)\colon C \cap A_{i,M}]\cdot  \mathbb{P}_i[p\in C \cap A_{i,M}]\\
&\ge& \frac{1}{2}\cdot \min_{C\in \{C_{\le j},C_{\ge j}\}} \mathbb{P}_i[p\in \core_2(C_j)\colon C \cap A_{i,M}] \\
&\ge& \frac{\varepsilon}{1200 k},
\end{eqnarray*}
where the last inequality follows from Equations~\eqref{eq:samplehit2} and~\eqref{eq:samplehit3}.
\end{proof}

We are now ready to prove our Theorem~\ref{thm:approxcenters}.
\begin{proof}[Proof of Theorem~\ref{thm:approxcenters}]
Consider the rooted tree constructed by the algorithm, sampling $1200\cdot (k/\varepsilon)\cdot \ln(k/\delta)$
points at each node. By Lemma~\ref{lem:uniform2}, $C_1$ is $(\varepsilon,5)$-covered by $M_{[0]}$ with probability
at least $1-\frac{\delta}{k}$. We will show by induction on $j\le k$ that with probability at least $1 - \frac{j\delta}{k}$ there
exists a level $j-1$ node $v$ for which $C_1,C_2,\dots,C_j$ are $(\varepsilon,5)$-covered by $M_v$. The base of
the induction is the above implication of Lemma~\ref{lem:uniform2}. Now, assume the inductive hypothesis. With
probability at least $1 - \frac{(j-1)\delta}{k}$ there exists a node $u$ in level $j-2$ for which $C_1,C_2,\dots,C_{j-1}$ 
are $(\varepsilon,5)$-covered by $M_u$. Let's condition on this event happening. If $C_j$ is $(\varepsilon,5)$-covered 
by $M_u$, then we are done as every child $v$ of $u$ satisfies $M_u\subset M_v$. Otherwise, $j$ is the smallest
index of a cluster not $(\varepsilon,5)$-covered by $M_u$. Moreover, as $M_{[0]}\subset M_u$ and $M_{[0]}$ contains
an approximate $k$-median solution, the condition $\sum_{p\in A} \min_{m\in M_u} d(p,m) \leq O(1)\cdot \OPT$ holds.
Thus, by Lemma~\ref{lem:sample}, there exists $i\le 10\log(nk/\varepsilon)$ such that the probability that $C_j$ 
is not $(\varepsilon,5)$-covered by $M_{[u,i]}$ is at most
$\left(1 - \frac{\varepsilon}{1200 k}\right)^{1200\cdot(k/\varepsilon)\cdot\ln(k/\delta)}\le \frac{\delta}{k}$.
If this bad event does not happen, we can take $v = [u,i]$. By the union bound, the probability that $u$ does not exist 
or that $u$ exists but $v$ does not exist is at most $\frac{(j-1)\delta}{k} + \frac{\delta}{k} = \frac{j\delta}{k}$.

For $j=k$ we get that with probability at least $1-\delta$ there exists a leaf $v$ of the tree such that all the clusters
in $\mathcal{C}$ are $(\varepsilon,5)$-covered by $M_v$. We also have 
$|M_v| = k + 1200\cdot (k^2/\varepsilon)\cdot\ln(k/\delta)$ (the first term is on account of the $k$-median 
approximation). The degree of each internal node of the tree is $10\log(nk/\varepsilon)$, and its depth 
is $k-1$, so the number of leaves (which is the number of sets we construct) is $\left(10\log(nk/\varepsilon)\right)^{k-1}$, 
and the total number of nodes (affecting the running time) is less than $\left(10\log(nk/\varepsilon)\right)^{k}$. If 
$k = o(\log(n/\varepsilon)/\log\log(n/\varepsilon))$, this is $\left(n/\varepsilon\right)^{o(1)}$. Otherwise, this is 
$2^{O(k\log k)}$.

It remains to bound the running time. A suitable approximation of $k$-median can be found in time $O(2^{O(k)}\cdot n)$ 
(for example, by invoking the bicriteria approximation in~\cite{Ind99a} and then enumerating over the choice of $k$
out of $O(k)$ centers). Computing the random sampling distribution is trivial at the root, and takes
$O(n\cdot |M_u|)$ operations at any other node $v = [u,i]$ (this also dominates generating the sample and
$M_v$ after computing the sampling distribution). We've already established that
$|M_u| = O((k^2/\varepsilon)\cdot\ln(k/\delta))$ and that the size of the tree is at most
$\max\left\{\left(n/\varepsilon\right)^{o(1)},2^{O(k\log k)}\right\}$. Therefore, the total number of operations
needed to construct the entire tree is at most 
$$
O(2^{O(k)}\cdot n) + \left(n/\varepsilon\right)^{1+o(1)}\cdot k^2\cdot \ln(k/\delta) + 
2^{O(k\log k)}\cdot (k^2/\varepsilon)\cdot \ln(k/\delta). 
$$
This is at most
$$
\left(n^{1+o(1)}+2^{O(k\log k)}\right)\cdot ((k/\varepsilon)\cdot\log(1/\delta))^{O(1)},
$$
as claimed.
\end{proof}

\subsection{Assignment of Far Points}\label{sec:farassignment}

Here we assume that we are given, for every cluster $C\in\mathcal{C}$, 
an $(\varepsilon,5)$-approximate center $\widehat{\mu_{C}}$ (see Definition~\ref{def:cover}),  
$\widehat{\Delta_{C}}$ that $\gamma$-approximates $\Delta_{C}$ (see Definition~\ref{def:size}),
and $\widehat{|C|}$ that $\gamma$-approximates $|C|$ (see Definition~\ref{def:card}),
for $\gamma = (\varepsilon/k)^{O(1)}$. For $i\in\{1,2,\dots,k\}$, put 
\begin{equation}\label{eq:Delta hat hat}
\widehat{\widehat{\Delta_{C_i}}} = 
\frac{\varepsilon\cdot(1+\gamma)^2\cdot\widehat{\OPT}}{20k\cdot \left(\widehat{|C_i|}\right)^2}
+ 5\cdot\widehat{\Delta_{C_i}}.
\end{equation}
Thus, by Definition~\ref{def:far} a point $p\in A$ is $\beta^{-1}$-far for $C_i$ iff 
$\|e_p - \widehat{\mu_{C_i}}\|_{\mathcal{W}_1}\ge \beta^{-1}\cdot\widehat{\widehat{\Delta_{C_i}}}$.

Also, let $z$ be the cost of the $c$-approximate $k$-median solution computed in the previous step.
Then, $\OPT\in [z/c,2nz]$, so we can enumerate over powers of $(1+\gamma)$ in this interval to get a
$\gamma$-approximation $\widehat{\OPT}$ of $\OPT$.

Let $0 < \beta\ll 1$ be a parameter to be determined. For any cluster $C$, let $C^{\Close}$ denote 
the set of points that are $\beta^{-2}$-close for $C$, let $C^{\Far}$ denote the set 
of points that are $\beta^{-7}$-far for $C$, let $C^{\Mid} = A\setminus (C^{\Close}\cup C^{\Far})$, 
let $C^{\neg\Far} = C^{\Close}\cup C^{\Mid}$, and let $C^{\neg\Close} = C^{\Mid}\cup C^{\Far}$.
For a cluster $j$, let 
$$
C_{\mapsfrom j} = \left\{p\in C_j\cap C_j^{\Far}\colon
  \exists i,\ \widehat{|C_i|}\ge \beta^4\cdot\widehat{|C_j|}\wedge p\in C_i^{\Close}\right\}.
$$
Let $g: \bigcup_j C_{\mapsfrom j} \rightarrow \{1,2,\dots,k\}$ be defined by choosing, for 
each source point $p$, an arbitrary ``blocking" target cluster $i$ (any one for which
$\widehat{|C_j|}\le\beta^{-4}\cdot\widehat{|C_i|}$ and $p\in C_i^{\Close}$). For all 
$i$, let $C_{i\mapsfrom}$ denote the set of points assigned by $g$ to $i$.
Put $D_i = (C_i\setminus C_{\mapsfrom i}) \cup C_{i\mapsfrom}$, and let
$\mathcal{D} = \{D_1,D_2,\dots,D_k\}$ be this clustering.\footnote{Note that the sets
of the form $C_i^{\Close}$, $C_i^{\Mid}$, $C_i^{\Far}$ are known to the
algorithm, but the sets of the form $C_{\mapsfrom j}$, $C_{i\mapsfrom}$, $D_i$,
and the function $g$, are not known to the algorithm. They will be used only in the 
analysis.}

Let $R = \lfloor 5\cdot \log_{1+\gamma} \frac 1 \beta \rfloor$.
For every point $p\in A$, let 
$\displaystyle j_p = \arg\min_{j\colon p\in C_j^{\Far}} \|e_p - \widehat{\mu_{C_j}}\|_{\mathcal{W}_1}$
(so if no such $j$ exists, then $j_p = 0$). Let $L = \{0,1,2,\dots,k,\Close,\Mid_1,\dots,\Mid_R,\Far\}$.
Consider for every point $p\in A$ the function $h_p:\{1,2,\dots,k\}\rightarrow L$ defined as follows.
$$
h_p(i) = \begin{cases}
             j_p       & \hbox{if }p\in C_i^{\Close}\wedge \left(j_p = 0\vee
                             \|e_p - \widehat{\mu_{C_{j_p}}}\|_{\mathcal{W}_1} < \beta\cdot \widehat{\widehat{\Delta_{C_i}}}\right)\\
             \Close & \hbox{if }p\in C_i^{\Close}\wedge j_p > 0\wedge
                            \|e_p - \widehat{\mu_{C_{j_p}}}\|_{\mathcal{W}_1}\ge \beta\cdot \widehat{\widehat{\Delta_{C_i}}}\\
             \Mid_s & \hbox{if }p\in C_i^{\Mid}\wedge s = \arg\min\{j\colon p\hbox{ is } 
                               \beta^{-2}(1+\gamma)^j\hbox{-close for }C_i\} \\
             \Far     & \hbox{if }p\in C_i^{\Far}.
             \end{cases}
$$
For every $h$ and $L'\subset L$, let 
$h^{-1}(L') = \{i\colon h(i)\in L'\}$, and let $B^h = \{p\in A\colon h_p = h\}$.\footnote{Clearly, the functions
$h_p$ and the sets of the form $B^h$ are known to the algorithm.}
We shall need $\gamma^2$-approximate cardinalities $\widehat{|B^h\cap D_i|}$ for all
$i$ and $h$. This requires enumeration over $\left(\gamma^{-2}\cdot\log n\right)^{k (R+k+1)^k}$ options.

To compute the desired assignment of far points, we set, for every point $p\in A$ and for every
$j\in\{1,2,\dots,k\}$,
$$
d_p(j) = \begin{cases}
             \infty & \hbox{if }p\in C_i^{\Close}\cap C_j^{\Far}\hbox{ for }i\hbox{ satisfying }
                         \widehat{|C_i|}\ge\beta^4\cdot\widehat{|C_j|}, \\
             \widehat{|C_j|}\cdot \|e_p-\widehat{\mu_{C_j}}\|_{\mathcal{W}_1} &
                                     \hbox{if otherwise }p\in C_j^{\Far},\\
             0 & \hbox{if }p\in C_j^{\neg\Far}.
             \end{cases}
$$
We next compute an optimal assignment of the points in $A$ to clusters, with assignment costs 
$\{d_p\}_{p\in A}$, and subject to the constraints that for every $h\in L^k$
and for every $i\in\{1,2,\dots,k\}$, we assign to $i$ at most $\widehat{|B^h\cap D_i|}$ points 
$p\in B^h$. Notice that this assignment can be computed by solving (integrally) a network 
flow problem. We will discuss later how to solve this flow problem efficiently using bucketing. 
Let $f:A\rightarrow\{1,2,\dots,k\}$ denote the computed assignment, and for all $i$ let
$F_i = f^{-1}(i)$ denote the set of points assigned to $i$. To complete the description 
of the algorithm in this section, for all $j$ let $P_j = C_j^{\Far}\cap F_j$.
We return the partial clustering $\mathcal{P} = \{P_1,P_2,\dots,P_k\}$.\footnote{It is crucial
that the assignment costs are defined using estimates for $\mathcal{C}$, whereas the
cardinality constraints are defined using estimates for $\mathcal{D}$. The $\infty$
cost factors in the definition of $\mathcal{D}$, and we want to avoid defining $\mathcal{D}$
using its own estimates.}

The analysis of $\mathcal{P}$ proceeds as follows. Let $\mathcal{F} = \{F_1,F_2,\dots,F_j\}$
denote the clustering induced by $f$. Let $\mathcal{D}'$ denote the partial assignment
derived from $\mathcal{D}$ by removing all of points in $\bigcup_i \left(C_i^{\Far}\cap (D_i\cup F_i)\right)$. 
Let $\mathcal{P}'$ denote the partial clustering that extends $\mathcal{D}'$ by assigning the 
points in $\bigcup_i \left(C_i^{\Far}\cap F_i\right)$ according to $f$ (so all the points added to 
$D'_i$ are in $C_i^{\Far}$). Finally, let $\mathcal{C}'$ denote the clustering that extends $\mathcal{P}'$ 
by choosing for every unassigned point $p\in \bigcup_i (C_i^{\Far}\cap D_i)$ a cluster 
$i\in h_p^{-1}(L')$, for $L' = L\setminus\{\Far\}$, without violating the capacity constraints 
$\widehat{|B^h\cap D_i|}$ for all $i,h$. Note that such an assignment must exist, as for every $h$, 
the set of points $B^h\setminus \bigcup_i \left(C_i^{\Far}\cap F_i\right)$ were assigned by $f$ to 
clusters indexed by $h_p^{-1}(L')$, so their total number is at most
$\sum_{i\in h_p^{-1}(L')} \widehat{|B^h\cap D_i|}$. The only points assigned in $\mathcal{P}'$
to clusters indexed by $h_p^{-1}(L')$ are from this set, so there's room for the rest.\footnote{Note 
that the assignment according to $f$ of these points is used only to reserve sufficient capacity.
Furthermore, as these assignment costs are all $0$, any such extension of the assignment of
$\bigcup_i \left(C_i^{\Far}\cap F_i\right)$ is a minimum cost assignment.} To summarize, the
points in $\bigcup_i \left(C_i^{\Far}\cap F_i\right)$ are assigned according to $\mathcal{F}$. The 
points in $\bigcup_i \left(C_i^{\neg\Far}\cap D_i\right)\setminus \bigcup_i \left(C_i^{\Far}\cap F_i\right)$
are assigned according to $\mathcal{D}$. The points in the remaining set of points
$\bigcup_i \left(C_i^{\Far}\cap D_i\right)\setminus \bigcup_i \left(C_i^{\Far}\cap F_i\right)$ are
assigned arbitrarily (but overall subject to the capacity constraints) to clusters for which they
are not $\beta^{-7}$-far. The output partial clustering $\mathcal{P}$ reflects a commitment to
the assignment of the first set of points, and a commitment to assign the rest of the points only
to clusters for which they are not $\beta^{-7}$-far. Denote the assignment used in the last step 
of the construction by 
$l:\left(\bigcup_i (C_i^{\Far}\cap D_i)\right)\setminus\left(\bigcup_i (C_i^{\Far}\cap F_i)\right)
\rightarrow\{1,2,\dots,k\}$.\footnote{Note that $l$ maps points to clusters for which they are 
not $\beta^{-7}$-far.} Note that we only compute $\mathcal{P}$, not its extension $\mathcal{C}'$
to a full clustering (as it depends on the optimal clustering $\mathcal{C}$). We use
$\mathcal{C}'$ in the analysis of $\mathcal{P}$, to argue that the partial clustering can be
extended to a near-optimal clustering.

We will show in Lemma~\ref{lem:smallfartoclose} that $\mathcal{D}$ is a near-optimal 
clustering. Then we will show in Lemma~\ref{lm:cost C'} that the cost of the
clustering $\mathcal{C}'$ is close to the cost of $\mathcal{D}$. 
We will make use of the following claim.
\begin{claim}\label{claim:markov}
For any cluster  $C_i$, there exist at most $\beta\cdot |C_i|$ many $\beta^{-1}$-far points. 
\end{claim}

\begin{proof}
We have
$$
\frac{1}{|C_i|}\cdot\sum_{x\in C_i} \|e_x-\widehat{\mu_{C_i}}\|_{\cal W_1} \le 
\frac{\varepsilon}{20}\cdot\frac{\OPT}{k\cdot |C_i|^2}+5\cdot\frac{\Cost(C_i)}{|C_i|^2}\le
\frac{\varepsilon\cdot(1+\gamma)^2}{20}\cdot
\frac{\widehat{\OPT}}{k\cdot \left|\widehat{C_i}\right|^2}+5\cdot\widehat{\Delta_{C_i}},
$$
where the inequalities follow from Definitions~\ref{def:cover} through~\ref{def:size}.
Using Definition~\ref{def:far} and Markov's Inequality, the claim follows.
\end{proof}

We first prove that the clustering $\mathcal{D}$ is near-optimal.
\begin{lemma}\label{lem:smallfartoclose}
For every $i=1,2,\dots,k$ we have that
\begin{enumerate}
\item $\displaystyle \Cost(D_i) < \Cost(C_i) + 36\cdot (k-1)\cdot\beta\cdot \OPT$; and,
\item $\displaystyle \Cost(C_i^{\neg\Far}\cap D_i) <  \Cost(C_i^{\neg\Far}\cap C_i) + 36\cdot (k-1)\cdot\beta\cdot \OPT$.
\end{enumerate}
\end{lemma}

\begin{proof}
Notice that due to Claim~\ref{claim:markov}, for every $j$ we have
$\left|C_j^{\Far}\cap  C_j\right| \leq \beta^7\cdot |C_j|$. As 
$C_{i\mapsfrom}$ contains only points from $C_j^{\Far}\cap  C_j$
for which $|C_j|\le\beta^{-4}\cdot |C_i|$, we have
$|C_{i\mapsfrom}|\le (k-1)\cdot\beta^3\cdot |C_i|$. Therefore
\begin{equation}\label{eq:D card}
\left(1-\beta^7\right)\cdot |C_i|\le |D_i|\le 
\left(1+(k-1)\cdot\beta^3\right)\cdot |C_i|.
\end{equation}
Moreover, for all $p\in C_i^{\Close}$ we have 
\begin{equation}\label{eq:close dist}
\|p-\widehat{\mu_{C_i}}\|_{\mathcal{W}_1} <
\beta^{-2}\cdot\widehat{\widehat{\Delta_{C_i}}}\le 6\cdot\beta^{-2}\cdot \frac{\OPT}{|C_i|^2}.
\end{equation}
We also have
\begin{equation}\label{eq:C to D}
\sum_{p\in D_i}\sum_{q\in D_i} d(p,q)  \le   \sum_{p\in C_i}\sum_{q\in C_i} d(p,q) + 
           2\sum_{p\in C_{i\mapsfrom}}\sum_{q\in C_i} d(p,q) +
           \sum_{p\in C_{i\mapsfrom}}\sum_{q\in C_{i\mapsfrom}} d(p,q).
\end{equation}
The first term is simply $\Cost(C_i)$. To bound the second and third terms, note that
by the triangle inequality, 
$d(p,q)\le \|p-\widehat{\mu_{C_i}}\|_{\mathcal{W}_1} + \|q-\widehat{\mu_{C_i}}\|_{\mathcal{W}_1}$.
Therefore, for the second term we have
\begin{eqnarray*}
        2\sum_{p\in C_{i\mapsfrom}}\sum_{q\in C_i} d(p,q)
& \le & 2\cdot |C_i|\cdot\sum_{p\in C_{i\mapsfrom}} \|p-\widehat{\mu_{C_i}}\|_{\mathcal{W}_1} +
        2\cdot |C_{i\mapsfrom}|\cdot \sum_{q\in C_i} \|q-\widehat{\mu_{C_i}}\|_{\mathcal{W}_1} \\
& \le & 12\cdot (k-1)\cdot\beta\cdot\OPT +
        12\cdot (k-1)\cdot\beta^3\cdot\OPT,
\end{eqnarray*}
where in the last inequality we have used the assumption that $\widehat{\mu_{C_i}}$ is
an $(\varepsilon,5)$-approximate center of $C_i$, so
$$
\sum_{q\in C_i} \|q-\widehat{\mu_{C_i}}\|_{\mathcal{W}_1} \le 
\frac{\varepsilon}{20}\cdot \frac{\OPT}{k\cdot |C_i|} + 5\cdot\frac{\Cost(C_i)}{|C_i|}
\le 6\cdot\frac{\OPT}{|C_i|}.
$$
For the third term we get, using Equation~\eqref{eq:close dist},
\begin{eqnarray*}
        \sum_{p\in C_{i\mapsfrom}}\sum_{q\in C_{i\mapsfrom}} d(p,q)
& \le & 2\cdot |C_{i\mapsfrom}|\cdot\sum_{p\in C_{i\mapsfrom}} \|p-\widehat{\mu_{C_i}}\|_{\mathcal{W}_1}\\
& \le & 12\cdot \beta^{-2}\cdot |C_{i\mapsfrom}|^2\cdot \frac{\OPT}{|C_i|^2} \\
& \le & 12\cdot (k-1)^2\cdot\beta^4\cdot\OPT.
\end{eqnarray*}
Summing up the terms (and recalling that $\beta\ll \frac 1 k$) proves 
the first assertion of the lemma. 

The proof of the second assertion of the lemma is identical, except that in
Equation~\eqref{eq:C to D} the summations on points in $C_i$ and $D_i$, respectively,
need to be replaced with summations on $C_i^{\neg\Far}\cap C_i$ and $C_i^{\neg\Far}\cap D_i$,
respectively.
\end{proof}

Let $f_{\mathcal D}$ denote the solution to the assignment problem that is induced by the
clustering $\mathcal{D}$. Notice that this is a feasible solution of finite cost. We next bound 
the cost of this assignment. We do it through comparison with the assignment $f_{\mathcal{C}}$
that is induced by the optimal clustering $\mathcal{C}$. Notice that this assignment may
have infinite cost, so we introduce an auxiliary assignment cost function to aid the comparison.
Specifically, for all $j$ let
$$
d'_p(j) = \begin{cases}
              \widehat{|C_j|}\cdot \|e_p-\widehat{\mu_{C_j}}\|_{\mathcal{W}_1} & 
                                     \hbox{if }p\in C_j^{\Far},\\
             0 & \hbox{if }p\in C_j^{\neg\Far}.
             \end{cases}
$$
\begin{claim}\label{lem:assgnuppbnd}
The cost of the assignment $f_{\mathcal{D}}$ satisfies that $2\cdot \sum_{p\in A} d_p(f_{\mathcal{D}}(p))$
is at most
$$
\left(1 + 2\beta^2 + \gamma\right)\cdot \sum_{i=1}^k \left(\sum_{p,q\in C_i} d(p,q)
               - \sum_{p,q\in C_i^{\neg\Far}\cap C_i} d(p,q) 
               - \sum_{p\in C_i^{\Far}\cap C_i} \sum_{q\in C_i^{\Mid}\cap C_i} d(p,q) \right) 
 + 6\cdot\beta^7\cdot\OPT.
$$
\end{claim}

\begin{proof}
Note that $\sum_{p\in A} d_p(f_{\mathcal{D}}(p)) = \sum_{p\in A} d'_p(f_{\mathcal{D}}(p))\le 
\sum_{p\in A} d'_p(f_{\mathcal{C}}(p))$. The equality follows as $f_{\mathcal{D}}$ uses no
point assignment of $d_p$-cost $\infty$, and the inequality follows as the change between 
$\mathcal{C}$ and $\mathcal{D}$ only switches point assignments of positive $d'_p$-cost
to point assignments of zero $d'_p$-cost. Now,
\begin{eqnarray}
\nonumber
            2\cdot \sum_{p\in A} d'_p(f_{\mathcal{C}}(p))
&  =  & 2\cdot \sum_{i=1}^k \sum_{p\in C_i^{\Far}\cap C_i} \widehat{|C_i|}\cdot 
            \|e_p-\widehat{\mu_{C_i}}\|_{\mathcal{W}_1} \\
\nonumber
& \le  & 2\cdot (1+\gamma)\cdot \sum_{i=1}^k \sum_{p\in C_i^{\Far}\cap C_i} |C_i|\cdot 
            \|e_p-\widehat{\mu_{C_i}}\|_{\mathcal{W}_1} \\
\nonumber
&  =  & (1+\gamma)\cdot\sum_{i=1}^k \left(\sum_{p\in C_i^{\Far}\cap C_i} 
                       \left(|C_i|+|C_i^{\neg\Close}\cap C_i|\right)\cdot 
                       \|e_p-\widehat{\mu_{C_i}}\|_{\mathcal{W}_1} \right. \\
\nonumber
&      & + \left.\sum_{p\in C_i^{\Far}\cap C_i} |C_i^{\Close}\cap C_i|\cdot 
                    \|e_p-\widehat{\mu_{C_i}}\|_{\mathcal{W}_1} \right) \\
\nonumber
& \le  & (1+\gamma)\cdot\left(1+\beta^2\right)\cdot\sum_{i=1}^k \sum_{p\in C_i^{\Far}\cap C_i} 
                     |C_i|\cdot \|e_p-\widehat{\mu_{C_i}}\|_{\mathcal{W}_1} \\
\nonumber
&      & + (1+\gamma)\cdot\sum_{i=1}^k \sum_{p\in C_i^{\Far}\cap C_i} |C_i^{\Close}\cap C_i|\cdot 
                    \|e_p-\widehat{\mu_{C_i}}\|_{\mathcal{W}_1}  \\
\nonumber
& \le  & (1+\gamma)\cdot\left(1+\beta^2\right)\cdot\sum_{i=1}^k \sum_{p\in C_i^{\Far}\cap C_i} 
                     |C_i|\cdot \left(\|e_p-\mu_{C_i}\|_{\mathcal{W}_1} + 
                            \|\mu_{C_i}-\widehat{\mu_{C_i}}\|_{\mathcal{W}_1}\right) \\
\nonumber
&      & + (1+\gamma)\cdot\sum_{i=1}^k \sum_{p\in C_i^{\Far}\cap C_i} |C_i^{\Close}\cap C_i|\cdot 
                    \|e_p-\widehat{\mu_{C_i}}\|_{\mathcal{W}_1}  \\
\nonumber
& \le  & (1+\gamma)\cdot\left(1+\beta^2\right)\cdot \sum_{i=1}^k \sum_{p\in C_i^{\Far}\cap C_i} 
                     \sum_{q\in C_i} d(p,q) + 6\cdot\beta^7\cdot\OPT\\
&      & + (1+\gamma)\cdot\sum_{i=1}^k \sum_{p\in C_i^{\Far}\cap C_i} |C_i^{\Close}\cap C_i|\cdot 
                    \|e_p-\widehat{\mu_{C_i}}\|_{\mathcal{W}_1},
\label{eq:sum f expr}
\end{eqnarray}
where we have used Claim~\ref{cl:approx dist} to get
\begin{eqnarray*}
\sum_{i=1}^k \sum_{p\in C_i^{\Far}\cap C_i}  |C_i|\cdot 
             \|\mu_{C_i}-\widehat{\mu_{C_i}}\|_{\mathcal{W}_1} 
&  =  & \sum_{i=1}^k |C_i^{\Far}\cap C_i|\cdot |C_i|\cdot 
             \|\mu_{C_i}-\widehat{\mu_{C_i}}\|_{\mathcal{W}_1} \\
& \le &  \sum_{i=1}^k\beta^7\cdot |C_i|^2\cdot\left(\frac{\varepsilon\cdot\OPT}{20k\cdot |C_i|^2} +
                5\cdot\Delta_{C_i}\right) \\
&  =  & \frac{\varepsilon\beta^7}{20}\cdot\OPT + 5\beta^7\cdot \sum_{i=1}^k \Cost(C_i) \\
&  =  & \left(5 + \frac{\varepsilon}{20}\right)\cdot\beta^7\cdot \OPT.
\end{eqnarray*}
Bounding 
$(1+\gamma)\cdot \left(1+\beta^2\right)\cdot \left(5 + \frac{\varepsilon}{20}\right)\le 6$
gives the stated inequality. 

To complete the proof, we need to bound the last term of the
right-hand side of Equation~\eqref{eq:sum f expr}, namely
$$
(1+\gamma)\cdot\sum_{i=1}^k \sum_{p\in C_i^{\Far}\cap C_i} |C_i^{\Close}\cap C_i|\cdot 
                    \|e_p-\widehat{\mu_{C_i}}\|_{\mathcal{W}_1}.
$$
To do that, note that if $p\in C_i^{\Far}\cap C_i$ and $q\in C_i^{\Close}\cap C_i$, then
$\|e_q-\widehat{\mu_{C_i}}\|_{\mathcal{W}_1} \le \beta^5\cdot \|e_p-\widehat{\mu_{C_i}}\|_{\mathcal{W}_1}$,
and therefore
$d(p,q)\ge \|e_p-\widehat{\mu_{C_i}}\|_{\mathcal{W}_1} - \|e_q-\widehat{\mu_{C_i}}\|_{\mathcal{W}_1}\ge
\left(1 - \beta^5\right)\cdot \|e_p-\widehat{\mu_{C_i}}\|_{\mathcal{W}_1}$.
This gives
$$
(1+\gamma)\cdot\sum_{i=1}^k \sum_{p\in C_i^{\Far}\cap C_i} |C_i^{\Close}\cap C_i|\cdot 
                    \|e_p-\widehat{\mu_{C_i}}\|_{\mathcal{W}_1}
 \le \frac{1+\gamma}{1-\beta^5}\cdot\sum_{i=1}^k \sum_{p\in C_i^{\Far}\cap C_i} 
             \sum_{q\in C_i^{\Close}\cap C_i} d(p,q).
$$
Plugging this bound into Equation~\eqref{eq:sum f expr}, we get
\begin{eqnarray*}
&      & 2\cdot \sum_{p\in A} d_p(f_{\mathcal{D}}(p)) \\
& \le & 2\cdot \sum_{p\in A} d'_p(f_{\mathcal{C}}(p)) \\
& \le & (1+\gamma)\cdot\left(1+\beta^2\right)\cdot \sum_{i=1}^k \sum_{p\in C_i^{\Far}\cap C_i} 
                     \sum_{q\in C_i} d(p,q) \\
&      & + \frac{1+\gamma}{1-\beta^5}\cdot\sum_{i=1}^k \sum_{p\in C_i^{\Far}\cap C_i} 
             \sum_{q\in C_i^{\Close}\cap C_i} d(p,q) + 
             6\cdot\beta^7\cdot\OPT\\
& \le & \left(1 + 2\beta^2 + \gamma\right)\cdot \sum_{i=1}^k \left(\sum_{p,q\in C_i} d(p,q)
               - \sum_{p,q\in C_i^{\neg\Far}\cap C_i} d(p,q) 
               - \sum_{p\in C_i^{\Far}\cap C_i} \sum_{q\in C_i^{\Mid}\cap C_i} d(p,q) \right) \\
&      &  + 6\cdot\beta^7\cdot\OPT,
\end{eqnarray*}
where the last inequality assumes that 
$(1+\gamma)\cdot\frac{\beta^5}{1-\beta^5}\le (1+\gamma)\cdot\beta^2\le 2\beta^2$.
\end{proof}

We will compare the cost of the clustering $\mathcal{C}'$ to the baseline of the
cost of the clustering $\mathcal{D}$, which is close to $\OPT$.
For $i\in\{1,2,\dots,k\}$ and $r\in\{0,1,2,\dots,k,\Close,\Mid_1,\dots,\Mid_R\}$, 
let $H_{i,r} = \{h\colon h(i) = r\}$. Also let $H_{i,\#} = \bigcup_{j=1}^k H_{i,j}$,
and let $H_{i,\Mid} = \bigcup_{s=1}^R H_{i,\Mid_s}$. We further put 
$C_{i,r} = \bigcup_{h\in H_{i,r}} B^h$,
for $i\in\{1,2,\dots,k\}$ and $r\in \{0,1,2,\dots,k,\#,\Close,\Mid_1,\dots,\Mid_R\}$.
Note that the points in $\bigcup_{i=1}^k C_{i,0}$ are not far for any cluster, 
so they are assigned in $\mathcal{C}'$ according to $\mathcal{C}$.
\begin{lemma}\label{lm:cost C'}
The clustering $\mathcal{C}'$ satisfies
$$
\Cost(\mathcal{C}')\le \left(1 + 44k^2\beta + \frac{14\beta^2}{\gamma} + 5\gamma\right)\cdot\OPT.
$$
\end{lemma}

\begin{proof}
We write $\Cost(\mathcal{C}')$ as
\begin{eqnarray*}
            \sum_{i=1}^k \sum_{p\in C'_i}\sum_{q\in C'_i} d(p,q)
&  =  & 2\cdot\sum_{i=1}^k \sum_{p\in C_i^{\Far}\cap C'_i} \sum_{q\in C'_i\setminus C_i^{\Far}}  d(p,q)
            + \sum_{i=1}^k \sum_{p\in C_i^{\Far}\cap C'_i} \sum_{q\in C_i^{\Far}\cap C'_i} d(p,q) \\         
&      & + \sum_{i=1}^k \sum_{p\in C_i^{\neg\Far}\cap C'_i} \sum_{q\in C_i^{\neg\Far}\cap C'_i} d(p,q)
\end{eqnarray*}
Claim~\ref{cl:C' far cost} upper bounds the first two terms. We further expand the last term as
follows.
\begin{eqnarray*}
&      & \sum_{i=1}^k \sum_{p\in C_i^{\neg\Far}\cap C'_i} \sum_{q\in C_i^{\neg\Far}\cap C'_i} d(p,q) \\
&  =  & 2\cdot \sum_{i=1}^k \sum_{p\in C_{i,\Close}\cap C'_i} \sum_{q\in (C_{i,\Mid}\cup C_{i,\#}\cup C_{i,0})\cap C'_i} d(p,q)
            +  \sum_{i=1}^k \sum_{p\in C_{i,\Close}\cap C'_i}  \sum_{q\in C_{i,\Close}\cap C'_i} d(p,q) \\
&      & + 2\cdot\sum_{i=1}^k \sum_{p\in C_i^{\Mid}\cap C'_i} \sum_{q\in (C_{i,\#}\cup C_{i,0})\cap C'_i} d(p,q)
            + \sum_{i=1}^k \sum_{p\in C_i^{\Mid}\cap C'_i} \sum_{q\in C_i^{\Mid}\cap C'_i} d(p,q) \\
&      & +  \sum_{i=1}^k \sum_{p\in (C_{i,\#}\cup C_{i,0})\cap C'_i} \sum_{q\in (C_{i,\#}\cup C_{i,0})\cap C'_i} d(p,q).            
\end{eqnarray*}
The first two terms are the tricky part of the argument. Claim~\ref{cl:C' close cost} upper bounds these
terms. Claim~\ref{cl:C' mid cost} upper bounds the next two terms. Finally, Claim~\ref{cl:C' sharp cost}
upper bounds the last term. Summing up these upper bounds, we get
\begin{eqnarray*}
            \Cost(\mathcal{C}')
& \le & (1+3\beta^2+2\gamma)\cdot \sum_{i=1}^k \left(\sum_{p,q\in C_i} d(p,q)
               - \sum_{p,q\in C_i^{\neg\Far}\cap C_i} d(p,q) \right) + 18\beta^7\cdot\OPT \\
&      & + 2(1+2\gamma)\cdot\sum_{i=1}^k \sum_{p\in C_{i,\Close}\cap D_i} \sum_{q\in C_{i,\Mid}\cap C'_i} d(p,q) \\
&      & + \sum_{i=1}^k \sum_{p\in C_{i,\Close}\cap D_i}  \sum_{q\in C_{i,\Close}\cap D_i} d(p,q)
             + 6\beta\cdot \OPT \\
&      & + 2(1+3k\beta^3+5\gamma)\cdot \sum_{i=1}^k \sum_{p\in C_i^{\Mid}\cap D_i} \sum_{q\in C_i^{\Close}\cap D_i} d(p,q)
           + \left(\frac{13\beta^2}{\gamma} + 6\beta\right)\cdot \OPT \\
&      & - 2(1+2\gamma)\cdot\sum_{i=1}^k \sum_{p\in C_{i,\Mid}\cap C'_i} \sum_{q\in C_{i,\Close}\cap D_i} d(p,q) \\
&      & + (1+\gamma^2)^2\cdot \sum_{i=1}^k \sum_{p\in (C_{i,\#}\cup C_{i,0})\cap D_i} 
               \sum_{q\in (C_{i,\#}\cup C_{i,0})\cap D_i} d(p,q) + 25\beta\cdot\OPT \\
& \le & (1+3\beta^2+5\gamma)\cdot \sum_{i=1}^k \sum_{p\in C_i} \sum_{q\in C_i} d(p,q)
            + \left(\frac{13\beta^2}{\gamma} + 44k^2\beta\right)\cdot\OPT \\
& \le &  \left(1 + 44k^2\beta + \frac{14\beta^2}{\gamma} + 5\gamma\right)\cdot\OPT.
\end{eqnarray*}
where the penultimate inequality uses the second part of Lemma~\ref{lem:smallfartoclose}.
\end{proof}

\begin{claim}\label{cl:avg}
Let $A_1,A_2\subset A$, and let $p,q\in A$. Then,
$$
d(p,q)\le \frac{1}{|A_1|\cdot |A_2|}\cdot \sum_{p'\in A_1}\sum_{q'\in A_2} \left(d(p,p')+d(p',q')+d(q',q)\right).
$$
\end{claim}

\begin{proof}
For every $p'\in A_1$ and $q'\in A_2$, the triangle inequality implies that
$d(p,q) \le d(p,p')+d(p',q')+d(q',q)$. Averaging over both sides yields the claim.
\end{proof}

\begin{claim}\label{cl:H-isharp}
Let $i\in\{1,2\dots,k\}$, let $j,j'\in \{1,2,\dots,k\}\setminus\{i\}$, and let $h\in H_{i,j}$ 
and $h'\in H_{i,j'}$. Then
\begin{enumerate}
\item If $p,q\in B^h$, then 
         $$
           d(p,q)\le 2\cdot\beta\cdot\widehat{\widehat{\Delta_{C_i}}}.
         $$
\item If $p\in B^h$ and $q\in B^{h'}$, then
         $$
           d(p,q)\le \frac{1}{|B^h\cap D_i|\cdot |B^{h'}\cap D_i|}\cdot
           \sum_{p'\in B^h\cap D_i}\sum_{q'\in B^{h'}\cap D_i} \left(d(p',q') + 
           4\cdot\beta\cdot\widehat{\widehat{\Delta_{C_i}}}\right).
          $$
\item For any $q\in A$, if $p\in B^h$, then
         $$
           d(p,q)\le \frac{1}{|B^h\cap D_i|}\cdot \sum_{p'\in B^h\cap D_i} \left(d(p',q) + 
            2\cdot\beta\cdot\widehat{\widehat{\Delta_{C_i}}}\right).
         $$
\end{enumerate}
\end{claim}

\begin{proof}
The first claim is a simple application of the triangle inequality:
$$
d(p,q)\le \|e_p - \widehat{\mu_{C_j}}\|_{\mathcal{W}_1(A,d)} + 
\|e_q - \widehat{\mu_{C_j}}\|_{\mathcal{W}_1(A,d)} \le 2\cdot\beta\cdot\widehat{\widehat{\Delta_{C_i}}}.
$$
The second and third claims apply the first claim and the mean-value property. For the second
claim:
\begin{eqnarray*}
           d(p,q)
& \le & \min\{d(p',q')\colon p'\in B^h\cap D_i\wedge q'\in B^{h'}\cap D_i\} + 
           4\cdot\beta\cdot\widehat{\widehat{\Delta_{C_i}}} \\
& \le & \frac{1}{|B^h\cap D_i|\cdot |B^{h'}\cap D_i|}\cdot
           \sum_{p'\in B^h\cap D_i}\sum_{q'\in B^{h'}\cap D_i} \left(d(p',q') + 
           4\cdot\beta\cdot\widehat{\widehat{\Delta_{C_i}}}\right).
\end{eqnarray*}
A similar argument proves the third claim.
\end{proof}

\begin{claim}\label{cl:C' far cost}
The contribution to $\Cost(\mathcal{C}')$ of distances involving at least one point in 
$C_i^{\Far}\cap C'_i$ satisfies
$$
       2\cdot\sum_{i=1}^k \sum_{p\in C_i^{\Far}\cap C'_i} \sum_{q\in C'_i\setminus C_i^{\Far}}  d(p,q)
           + \sum_{i=1}^k \sum_{p\in C_i^{\Far}\cap C'_i} \sum_{q\in C_i^{\Far}\cap C'_i} d(p,q) \\
\le 2(1+\beta^5)\cdot \sum_{p\in A} d_p(f_{\mathcal{D}}(p)) + 11\beta^7\cdot\OPT.
$$
\end{claim}

\begin{proof}
We use the triangle inequality to write:
\begin{eqnarray*}
&      & 2\cdot\sum_{i=1}^k \sum_{p\in C_i^{\Far}\cap C'_i} \sum_{q\in C'_i\setminus C_i^{\Far}}  d(p,q)
           +  \sum_{i=1}^k \sum_{p\in C_i^{\Far}\cap C'_i} \sum_{q\in C_i^{\Far}\cap C'_i} d(p,q) \\
& \le & 2\cdot\sum_{i=1}^k \sum_{p\in C_i^{\Far}\cap C'_i} \sum_{q\in C'_i\setminus C_i^{\Far}}  
                    \left(\|e_p - \widehat{\mu_{C_i}}\|_{\mathcal{W}_1} + \|\widehat{\mu_{C_i}} - e_q\|_{\mathcal{W}_1}\right) \\
&      & +  \sum_{i=1}^k \sum_{p\in C_i^{\Far}\cap C'_i} \sum_{q\in C_i^{\Far}\cap C'_i} 
                   \left(\|e_p - \widehat{\mu_{C_i}}\|_{\mathcal{W}_1} + \|\widehat{\mu_{C_i}} - e_q\|_{\mathcal{W}_1}\right) \\
&  =  & 2\cdot\sum_{i=1}^k \sum_{p\in C_i^{\Far}\cap C'_i} \sum_{q\in C_i^{\Mid}\cap C'_i}  
                    \left(\|e_p - \widehat{\mu_{C_i}}\|_{\mathcal{W}_1} + \|\widehat{\mu_{C_i}} - e_q\|_{\mathcal{W}_1}\right) \\
&      & +  2\cdot\sum_{i=1}^k \sum_{p\in C_i^{\Far}\cap C'_i} \sum_{q\in C_i^{\Close}\cap C'_i}  
                    \left(\|e_p - \widehat{\mu_{C_i}}\|_{\mathcal{W}_1} + \|\widehat{\mu_{C_i}} - e_q\|_{\mathcal{W}_1}\right) \\
&      & +  \sum_{i=1}^k \sum_{p\in C_i^{\Far}\cap C'_i} \sum_{q\in C_i^{\Far}\cap C'_i} 
                   \left(\|e_p - \widehat{\mu_{C_i}}\|_{\mathcal{W}_1} + \|\widehat{\mu_{C_i}} - e_q\|_{\mathcal{W}_1}\right) \\
& \le & 2\cdot\sum_{i=1}^k \sum_{s=1}^R \sum_{h\in H_{i,\Mid_s}} \sum_{q\in B^{h}\cap C'_i}
                |C_i^{\Far}\cap C'_i|\cdot \|\widehat{\mu_{C_i}} - e_q\|_{\mathcal{W}_1} \\
&      & + 2\cdot\sum_{i=1}^k \sum_{p\in C_i^{\Far}\cap C'_i} |C_i^{\Mid}\cap C'_i|\cdot
                    \|e_p - \widehat{\mu_{C_i}}\|_{\mathcal{W}_1} \\
&     & + 2\cdot\sum_{i=1}^k \sum_{p\in C_i^{\Far}\cap C'_i} 
                    (1+ \beta^5)\cdot |C_i^{\Close}\cap C'_i|\cdot \|e_p - \widehat{\mu_{C_i}}\|_{\mathcal{W}_1} \\
&      & +  \sum_{i=1}^k \sum_{p\in C_i^{\Far}\cap C'_i} \sum_{q\in C_i^{\Far}\cap C'_i} 
                   \left(\|e_p - \widehat{\mu_{C_i}}\|_{\mathcal{W}_1} + \|\widehat{\mu_{C_i}} - e_q\|_{\mathcal{W}_1}\right) \\
& \le & 2\cdot \widehat{|C_i^{\Far}\cap D_i|}\cdot 
               \sum_{s=1}^R \frac{\beta^2}{(1+\gamma)^{s-1}}\cdot \widehat{|D_i|}\cdot 
                  \beta^{-2}(1+\gamma)^s\cdot\widehat{\widehat{\Delta_{C_i}}} \\
&      & +  2(1+\beta^5)\cdot \sum_{i=1}^k \sum_{p\in C_i^{\Far}\cap C'_i}  \widehat{|D_i|}\cdot 
                    \|e_p - \widehat{\mu_{C_i}}\|_{\mathcal{W}_1} \\
& \le & 11\beta^7\cdot\OPT + 2(1+\beta^5)\cdot\sum_{p\in A} d_p(f_{\mathcal{C}'}(p)) \\
& \le & 11\beta^7\cdot\OPT + 2(1+\beta^5)\cdot \sum_{p\in A} d_p(f_{\mathcal{D}}(p)),        
\end{eqnarray*} 
where the second inequality uses the fact that if $p$ is $\beta^{-7}$-far for $C_i$ and
$q$ is $\beta^{-2}$-close for $C_i$, then 
$\|\widehat{\mu_{C_i}} - e_q\|_{\mathcal{W}_1} \le \beta^5 \cdot \|e_p - \widehat{\mu_{C_i}}\|_{\mathcal{W}_1}$,
the third inequality uses the fact that $|C_i^{\Far}\cap C'_i|\le \widehat{|C_i^{\Far}\cap D_i|}$
and $|C'_i|\le \widehat{|D_i|}$, and the last inequality follows from the optimality of the assignment 
$f$ with respect to the assignment costs $d_p$ (note that for all $i$ we have 
$f(p) = f_{\mathcal{C}'}(p) = i$ for $p\in C_i^{\Far}\cap C'_i$).
\end{proof}

\begin{claim}\label{cl:C' close cost}
The contribution to $\Cost(\mathcal{C}')$ of distances involving at least one point in 
$C_{i,\Close}\cap C'_i$ and no point in $C_i^{\Far}\cap C'_i$ satisfies
\begin{eqnarray*}
&      & 2\cdot \sum_{i=1}^k \sum_{p\in C_{i,\Close}\cap C'_i} 
                             \sum_{q\in (C_{i,\Mid}\cup C_{i,\#}\cup C_{i,0})\cap C'_i} d(p,q)
           + \sum_{i=1}^k \sum_{p\in C_{i,\Close}\cap C'_i} 
                  \sum_{q\in C_{i,\Close}\cap C'_i} d(p,q) \\
& \le & 4\beta\cdot\sum_{p\in A} d_p(f_{\mathcal{D}}(p))
             + 2(1+2\gamma)\cdot\sum_{i=1}^k \sum_{p\in C_{i,\Close}\cap D_i}\sum_{q\in (C_{i,\#}\cup C_{i,0})\cap D_i} d(p,q) \\
&      & + 2(1+2\gamma)\cdot\sum_{i=1}^k \sum_{p\in C_{i,\Close}\cap D_i} \sum_{q\in C_{i,\Mid}\cap C'_i} d(p,q) 
             + \sum_{i=1}^k \sum_{p\in C_{i,\Close}\cap D_i}  \sum_{q\in C_{i,\Close}\cap D_i} d(p,q)
             + 6\beta\cdot \OPT.
\end{eqnarray*}
\end{claim}

\begin{proof}
Denote $G_i = C_{i,\Close}\cap D_i$ and $G'_i = C_{i,\Close}\cap C'_i$, consider 
$p\in G'_i\setminus G_i$ and let $i_p\ne i$ (not to be confused with $j_p$) be such 
that $p\in C_{i_p}^{\Far}\cap D_{i_p}$. Note that $i_p$ must exist by the construction 
of $\mathcal{C}'$, as $p\not\in D_i$. By the definition of $C_{i,\Close}$, we have that
\begin{equation}\label{eq:cls cost}
\widehat{|C_i|}\cdot \|e_p - \widehat{\mu_{C_i}}\|_{\mathcal{W}_1}
\le \widehat{|C_i|}\cdot \beta^{-2}\cdot\widehat{\widehat{\Delta_{C_i}}}
\le \beta^4\cdot \widehat{|C_{i_p}|}\cdot \beta^{-2}\cdot\beta^{-1}\cdot 
     \|e_p - \widehat{\mu_{C_{i_p}}}\|_{\mathcal{W}_1}
= \beta\cdot \widehat{|C_{i_p}|}\cdot \|e_p - \widehat{\mu_{C_{i_p}}}\|_{\mathcal{W}_1}.
\end{equation}
Further note that for every $i=1,2,\dots,k$, $C_{i,0}\cap C'_i = C_{i,0}\cap D_i$ (or
equivalently, for every $h\in H_{i,0}$, $B^h\cap C'_i = B^h\cap D_i$).

We need to consider two cases:

\paragraph{\boldmath Case I: $|G_i\setminus G'_i|\ge \gamma\cdot |G_i|$.}
Notice that in this case, 
$|G'_i\setminus G_i|\le |G_i\setminus G'_i| + \gamma^2\cdot |G_i|\le (1+\gamma)\cdot |G_i\setminus G'_i|$.
We first bound.
\begin{eqnarray}
\nonumber
&      & 2\cdot\sum_{p\in G'_i} \sum_{q\in (C_{i,\#}\cup C_{i,0})\cap C'_i} d(p,q) \\
\nonumber
& \le & 2\cdot\sum_{p\in G'_i\setminus G_i} \sum_{h'\in H_{i,\#}} \sum_{q\in B^{h'}\cap C'_i} 
            \frac{1}{|G_i\setminus G'_i|\cdot |B^{h'}\cap D_i|}\cdot \sum_{p'\in G_i\setminus G'_i}\sum_{q'\in B^{h'}\cap D_i} 
            \left(d(p,p') + d(p',q') + d(q',q)\right) \\
\nonumber
&     & + 2\cdot\sum_{p\in G'_i\setminus G_i} \sum_{q\in C_{i,0}\cap C'_i} 
            \frac{1}{|G_i\setminus G'_i|}\cdot \sum_{p'\in G_i\setminus G'_i} \left(d(p,p') + d(p',q)\right) \\
\nonumber
&      & + 2\cdot\sum_{p\in G'_i\cap G_i}
            \sum_{h'\in H_{i,\#}} \sum_{q\in B^{h'}\cap C'_i} 
            \frac{1}{|B^{h'}\cap D_i|}\cdot \sum_{q'\in B^{h'}\cap D_i} 
            \left(d(p,q') + d(q',q)\right) + 2\cdot\sum_{p\in G'_i\cap G_i} \sum_{q\in C_{i,0}\cap C'_i} d(p,q) \\
\nonumber
& \le & 2\cdot\sum_{p\in G'_i\setminus G_i} |(C_{i,\#}\cup C_{i,0})\cap C'_i|\cdot
              2\cdot\beta^{-2}\cdot\widehat{\widehat{\Delta_{C_i}}}
           + 2\cdot\frac{|G'_i\setminus G_i|}{|G_i\setminus G'_i|}\cdot \sum_{h\in H_{i,\#}\cup H_{i,0}} 
               \frac{|B^h\cap C'_i|}{|B^h\cap D_i|}\cdot \sum_{p\in G_i}\sum_{q\in B^h\cap D_i} d(p,q) \\
\nonumber
&      & + 2\cdot\sum_{h\in H_{i,\#}\cup H_{i,0}} \frac{|B^h\cap C'_i|}{|B^h\cap D_i|}\cdot
                \sum_{p\in G'_i\cap G_i}\sum_{q\in B^h\cap D_i} d(p,q)
            + 2\cdot\sum_{p\in G'_i\cap G_i} |C_{i,\#}\cap C'_i|\cdot 2\cdot\beta\cdot \widehat{\widehat{\Delta_{C_i}}} \\            
\nonumber
& \le & 4\cdot\sum_{p\in G'_i\setminus G_i} |(C_{i,\#}\cup C_{i,0})\cap C'_i|\cdot\beta^{-2}\cdot\widehat{\widehat{\Delta_{C_i}}}
           + 2(1+\gamma^2)^2\cdot\sum_{p\in G_i}\sum_{q\in (C_{i,\#}\cup C_{i,0})\cap D_i} d(p,q) 
           + \beta\cdot |C'_i|^2\cdot \widehat{\widehat{\Delta_{C_i}}} \\   
\label{eq:terms dist cls sharp}
& \le & 4\cdot\sum_{p\in G'_i\setminus G_i} |(C_{i,\#}\cup C_{i,0})\cap C'_i|\cdot\beta^{-2}\cdot\widehat{\widehat{\Delta_{C_i}}}
           + 2(1+\gamma)\cdot\sum_{p\in G_i}\sum_{q\in (C_{i,\#}\cup C_{i,0})\cap D_i} d(p,q) 
           + \beta\cdot |C'_i|^2\cdot \widehat{\widehat{\Delta_{C_i}}}, 
\end{eqnarray}
where the first inequality uses Claim~\ref{cl:avg}, the second inequality uses the fact that
$p,p'\in C_i^{\Close}$ and $q,q'\in C_{i,j}$ for some $j$, and the third inequality uses bounds 
on the cardinalities of the sets.

Next, we bound
\begin{eqnarray}
\nonumber
&      & 2\cdot\sum_{p\in G'_i} \sum_{q\in C_i^{\Mid}\cap C'_i} d(p,q) \\
\nonumber
& \le & 2\cdot\sum_{p\in G'_i\setminus G_i} \frac{1}{|G_i\setminus G'_i|}\cdot\sum_{p'\in G_i\setminus G'_i} 
                    \sum_{q\in C_i^{\Mid}\cap C'_i} \left(d(p,p') + d(p',q)\right) 
                  + 2\cdot\sum_{p\in G'_i\cap G_i} \sum_{q\in C_i^{\Mid}\cap C'_i} d(p,q) \\
\nonumber
&  =  & 2\cdot\sum_{p\in G'_i\setminus G_i} \frac{|C_i^{\Mid}\cap C'_i|}{|G_i\setminus G'_i|}\cdot
                     \sum_{p'\in G_i\setminus G'_i} d(p,p')
               + 2\cdot\frac{|G'_i\setminus G_i|}{|G_i\setminus G'_i|} \cdot 
                       \sum_{p\in G_i\setminus G'_i}  \sum_{q\in C_i^{\Mid}\cap C'_i} d(p,q) \\
\nonumber
&      &  + 2\cdot\sum_{p\in G_i\cap G'_i} \sum_{q\in C_i^{\Mid}\cap C'_i} d(p,q) \\
\label{eq:terms dist cls mid}
& \le & 4\cdot\sum_{p\in G'_i\setminus G_i} |C_i^{\Mid}\cap C'_i|\cdot\beta^{-2}\cdot\widehat{\widehat{\Delta_{C_i}}}
                + 2(1+\gamma)\cdot\sum_{p\in G_i}  \sum_{q\in C_i^{\Mid}\cap C'_i} d(p,q).
\end{eqnarray}

\paragraph{\boldmath Case II: $|G_i\setminus G'_i| < \gamma\cdot |G_i|$.} 
In this case, 
$|G'_i\setminus G_i|\le |G_i\setminus G'_i| + \gamma^2\cdot |G_i| < (\gamma+\gamma^2)\cdot |G_i|$.
We repeat the above arguments, averaging over $G_i$ instead of $G_i\setminus G'_i$
(which may be empty). We get
\begin{eqnarray*}
&      & 2\cdot\sum_{p\in G'_i} \sum_{q\in (C_{i,\#}\cup C_{i,0}\cap C'_i} d(p,q) \\
& \le & 2\cdot\sum_{p\in G'_i\setminus G_i} \sum_{h'\in H_{i,\#}} \sum_{q\in B^{h'}\cap C'_i} 
            \frac{1}{|G_i|\cdot |B^{h'}\cap D_i|}\cdot \sum_{p'\in G_i}\sum_{q'\in B^{h'}\cap D_i} 
            \left(d(p,p') + d(p',q') + d(q',q)\right) \\
&     & + 2\cdot\sum_{p\in G'_i\setminus G_i} \sum_{q\in C_{i,0}\cap C'_i} 
            \frac{1}{|G_i|}\cdot \sum_{p'\in G_i} \left(d(p,p') + d(p',q)\right) \\
&      & + 2\cdot\sum_{p\in G'_i\cap G_i}
            \sum_{h'\in H_{i,\#}} \sum_{q\in B^{h'}\cap C'_i} 
            \frac{1}{|B^{h'}\cap D_i|}\cdot \sum_{q'\in B^{h'}\cap D_i} 
            \left(d(p,q') + d(q',q)\right) 
            + 2\cdot\sum_{p\in G'_i\cap G_i} \sum_{q\in C_{i,0}\cap C'_i} d(p,q)\\
& \le & 2\cdot\sum_{p\in G'_i\setminus G_i} |(C_{i,\#}\cup C_{i,0})\cap C'_i|\cdot
              2\cdot\beta^{-2}\cdot\widehat{\widehat{\Delta_{C_i}}}
           + 2\cdot\frac{|G'_i\setminus G_i|}{|G_i|}\cdot \sum_{h\in H_{i,\#}\cup H_{i,0}} 
               \frac{|B^h\cap C'_i|}{|B^h\cap D_i|}\cdot \sum_{p\in G_i}\sum_{q\in B^h\cap D_i} d(p,q) \\
&      & + 2\cdot\sum_{h\in H_{i,\#}\cup H_{i,0}} \frac{|B^h\cap C'_i|}{|B^h\cap D_i|}\cdot
                \sum_{p\in G'_i\cap G_i}\sum_{q\in B^h\cap D_i} d(p,q)
            + 2\cdot\sum_{p\in G'_i\cap G_i} |C_{i,\#}\cap C'_i|\cdot 2\cdot\beta\cdot \widehat{\widehat{\Delta_{C_i}}} \\            
& \le & 4\cdot\sum_{p\in G'_i} |(C_{i,\#}\cup C_{i,0})\cap C'_i|\cdot\beta^{-2}\cdot\widehat{\widehat{\Delta_{C_i}}}
           + 2(1+2\gamma)\cdot\sum_{p\in G_i}\sum_{q\in (C_{i,\#}\cup C_{i,0})\cap D_i} d(p,q) 
           + \beta\cdot\sum_{i=1}^k  |C'_i|^2\cdot \widehat{\widehat{\Delta_{C_i}}},
\end{eqnarray*}
similarly to Equation~\eqref{eq:terms dist cls sharp}. Also
\begin{eqnarray}
\nonumber
&      & 2\cdot\sum_{p\in G'_i} \sum_{q\in C_i^{\Mid}\cap C'_i} d(p,q) \\
\nonumber
& \le & 2\cdot\sum_{p\in G'_i\setminus G_i} \frac{1}{|G_i|}\cdot\sum_{p'\in G_i} 
                    \sum_{q\in C_i^{\Mid}\cap C'_i} \left(d(p,p') + d(p',q)\right) 
                  + 2\cdot\sum_{p\in G'_i\cap G_i} \sum_{q\in C_i^{\Mid}\cap C'_i} d(p,q) \\
\nonumber
&  =  & 2\cdot\sum_{p\in G'_i\setminus G_i} \frac{|C_i^{\Mid}\cap C'_i|}{|G_i|}\cdot
                     \sum_{p'\in G_i} d(p,p')
               + 2\cdot\frac{|G'_i\setminus G_i|}{|G_i|} \cdot 
                       \sum_{p\in G_i}  \sum_{q\in C_i^{\Mid}\cap C'_i} d(p,q) \\
\nonumber
&      &  + 2\cdot\sum_{p\in G_i\cap G'_i} \sum_{q\in C_i^{\Mid}\cap C'_i} d(p,q) \\
\nonumber
& \le & 4\cdot\sum_{p\in G'_i\setminus G_i} |C_i^{\Mid}\cap C'_i|\cdot\beta^{-2}\cdot\widehat{\widehat{\Delta_{C_i}}}
                + 2(1+2\gamma)\cdot\sum_{p\in G_i}  \sum_{q\in C_i^{\Mid}\cap C'_i} d(p,q),
\end{eqnarray}
similarly to Equation~\eqref{eq:terms dist cls mid}.

\paragraph{Both cases.}
Finally, we bound
\begin{eqnarray}
\nonumber
&      & \sum_{p\in G'_i} \sum_{q\in G'_i} d(p,q) = 
            2\cdot\sum_{p\in G'_i\setminus G_i} \sum_{q\in G'_i\cap G_i} d(p,q) 
            + \sum_{p\in G'_i\setminus G_i} \sum_{q\in G'_i\setminus G_i} d(p,q) 
            + \sum_{p\in G'_i\cap G_i} \sum_{q\in G'_i\cap G_i} d(p,q) \\
\label{eq:terms dist cls cls}
& \le & 2\cdot \sum_{p\in G'_i\setminus G_i} |G'_i\cap G_i|\cdot\beta^{-2}\cdot\widehat{\widehat{\Delta_{C_i}}}
            + \sum_{p\in G'_i\setminus G_i} |G'_i\setminus G_i|\cdot\beta^{-2}\cdot\widehat{\widehat{\Delta_{C_i}}}
            + \sum_{p\in G'_i\cap G_i} \sum_{q\in G'_i\cap G_i} d(p,q).
\end{eqnarray}

To complete the proof, we sum the right-hand sides of 
Equations~\eqref{eq:terms dist cls sharp} and~\eqref{eq:terms dist cls mid} (or the corresponding Case II
equations), and Equation~\eqref{eq:terms dist cls cls} on all $k$ clusters to get an upper bound of
\begin{eqnarray*}
\nonumber
&      & 4\cdot\sum_{i=1}^k \sum_{p\in G'_i\setminus G_i} |C'_i|\cdot\beta^{-2}\cdot\widehat{\widehat{\Delta_{C_i}}}
            + 2(1+2\gamma)\cdot\sum_{i=1}^k \sum_{p\in G_i}\sum_{q\in (C_{i,\#}\cup C_{i,0})\cap D_i} d(p,q) \\
&      & + 2(1+2\gamma)\cdot\sum_{i=1}^k \sum_{p\in G_i}\sum_{q\in C_i^{\Mid}\cap C'_i} d(p,q)
            + \sum_{i=1}^k \sum_{p\in G_i}\sum_{q\in G_i} d(p,q)
            + \beta\cdot \sum_{i=1}^k |C'_i|^2\cdot \widehat{\widehat{\Delta_{C_i}}} \\
& \le & 4\beta \cdot\sum_{i=1}^k \sum_{p\in G'_i\setminus G_i} |C'_{i_p}|\cdot\|e_p - \widehat{\mu_{C_{i_p}}}\|_{\mathcal{W}_1}
            + 2(1+2\gamma)\cdot\sum_{i=1}^k \sum_{p\in G_i}\sum_{q\in (C_{i,\#}\cup C_{i,0})\cap D_i} d(p,q) \\
&      & + 2(1+2\gamma)\cdot\sum_{i=1}^k \sum_{p\in G_i}\sum_{q\in C_i^{\Mid}\cap C'_i} d(p,q)
            + \sum_{i=1}^k \sum_{p\in G_i}\sum_{q\in G_i} d(p,q)
            + 6\beta\cdot\OPT \\
& \le & 4\beta\cdot\sum_{p\in A} d_p(f_{\mathcal{D}}(p))
             + 2(1+2\gamma)\cdot\sum_{i=1}^k \sum_{p\in G_i}\sum_{q\in (C_{i,\#}\cup C_{i,0})\cap D_i} d(p,q) \\
&      & + 2(1+2\gamma)\cdot\sum_{i=1}^k \sum_{p\in G_i}\sum_{q\in C_i^{\Mid}\cap C'_i} d(p,q)
             + \sum_{i=1}^k \sum_{p\in G_i}\sum_{q\in G_i} d(p,q)
             + 6\beta\cdot\OPT,
\end{eqnarray*}
where the first inequality uses Equations~\eqref{eq:cls cost} and~\eqref{eq:Delta hat hat}.
\end{proof}

\begin{claim}\label{cl:C' mid cost}
The contribution to $\Cost(\mathcal{C}')$ of distances involving at least one point in 
$C_i^{\Mid}\cap C'_i$ and no point in $(C_i^{\Far}\cup C_{i,\Close})\cap C'_i$ satisfies
\begin{eqnarray*}
&      & 2\cdot\sum_{i=1}^k \sum_{p\in C_i^{\Mid}\cap C'_i} \sum_{q\in (C_{i,\#}\cup C_{i,0})\cap C'_i} d(p,q)
            + \sum_{i=1}^k \sum_{p\in C_i^{\Mid}\cap C'_i} \sum_{q\in C_i^{\Mid}\cap C'_i} d(p,q) \\
& \le & 2(1+3k\beta^3+5\gamma)\cdot \sum_{i=1}^k \sum_{p\in C_i^{\Mid}\cap D_i} \sum_{q\in C_i^{\Close}\cap D_i} d(p,q)
           + \left(\frac{13\beta^2}{\gamma} + 6\beta\right)\cdot \OPT \\
&      & - 2(1+2\gamma)\cdot\sum_{i=1}^k \sum_{p\in C_{i,\Mid}\cap C'_i} \sum_{q\in C_{i,\Close}\cap D_i} d(p,q).
\end{eqnarray*}
\end{claim}

\begin{proof}
Note that for $s,t\in \{1,\dots,R\}$, for $h\in H_{i,\Mid_s}$ and $h'\in H_{i,\Mid_t}$, and for $p\in B^h$ 
and $q\in B^{h'}$, we have that
\begin{equation}\label{eq:mid dist}
d(p,q)\le \beta^{-2}\cdot \left((1+\gamma)^s + (1+\gamma)^t\right)\cdot\widehat{\widehat{\Delta_{C_i}}}.
\end{equation}
Therefore, 
\begin{eqnarray}
\nonumber
&      & \sum_{i=1}^k \sum_{p\in C_i^{\Mid}\cap C'_i} \sum_{q\in C_i^{\Mid}\cap C'_i} d(p,q) \\
\nonumber
&  =  & \sum_{i=1}^k \sum_{s=1}^R \sum_{t=1}^R \sum_{h\in H_{i,\Mid_s}} \sum_{h'\in H_{i,\Mid_t}}
                  \sum_{p\in B^h\cap C'_i} \sum_{q\in B^{h'}\cap C'_i} d(p,q) \\
\nonumber
& \le & \beta^{-2}\cdot\sum_{i=1}^k \widehat{\widehat{\Delta_{C_i}}}\cdot 
                 \sum_{s=1}^R \sum_{t=1}^R  \left((1+\gamma)^s + (1+\gamma)^t\right)
                \sum_{h\in H_{i,\Mid_s}} |B^h\cap C'_i|\cdot \sum_{h'\in H_{i,\Mid_t}} |B^{h'}\cap C'_i| \\
\nonumber
& \le & \beta^{-2}\cdot\sum_{i=1}^k \widehat{\widehat{\Delta_{C_i}}}\cdot 
                 \sum_{s=1}^R \sum_{t=1}^R  \left((1+\gamma)^s + (1+\gamma)^t\right)\cdot
                 \frac{\beta^2\cdot |\widehat{C_i}|}{(1+ \gamma)^{s-1}}\cdot
                 \frac{\beta^2\cdot |\widehat{C_i}|}{(1+ \gamma)^{t-1}} \\
\label{eq:mid-mid}
& \le & 2(1+\gamma)\cdot \frac{\beta^2}{\gamma}\cdot 
               \sum_{i=1}^k |\widehat{C_i}|^2\cdot \widehat{\widehat{\Delta_{C_i}}}\le \frac{12\beta^2}{\gamma}\cdot \OPT,
\end{eqnarray}
where the first inequality uses Equation~\eqref{eq:mid dist}, the second inequality uses 
Claim~\ref{claim:markov} along with the fact that $p$ is $\beta^{-2}\cdot (1+\gamma)^{s-1}$-far 
for $C_i$ (and the similar fact for $q$), and the third inequality uses the bound 
$\sum_{s=0}^{\infty} (1+\gamma)^{-s} = \frac{1}{\gamma}$.

Further note that for every $i\in\{1,2,\dots,k\}$,
\begin{eqnarray}
\nonumber
&      & \sum_{p\in C_i^{\Mid}\cap C'_i} \sum_{q\in (C_{i,\#}\cup C_{i,0})\cap C'_i} d(p,q) \\
\nonumber
& \le & \sum_{p\in C_i^{\Mid}\cap C'_i} \sum_{h\in H_{i,\#}}
              \sum_{q\in B^h\cap C'_i} \frac{1}{|B^h\cap D_i|}\cdot\sum_{q'\in B^h\cap D_i} 
              \left(d(p,q') + 2\beta\cdot \widehat{\widehat{\Delta_{C_i}}}\right) \\
\nonumber
&     & \sum_{p\in C_i^{\Mid}\cap C'_i} \sum_{q\in C_{i,0}\cap C'_i} d(p,q) \\
\label{eq:mid-sharp}
& \le & (1+\gamma^2)\cdot\sum_{p\in C_i^{\Mid}\cap C'_i} \sum_{q\in (C_{i,\#}\cup C_{i,0})\cap D_i} d(p,q)
             + \frac{\beta}{2}\cdot |C'_i|^2\cdot \widehat{\widehat{\Delta_{C_i}}}
\end{eqnarray}
where the first inequality uses Claim~\ref{cl:H-isharp}, and the second inequality
uses bounds on the cardinalities of the relevant sets.

In order to complete the proof, we need to upper bound the distances involving a point in 
$C_i^{\Close}\cap D_i$. To do this, we first argue that for every $i\in\{1,2,\dots,k\}$,
\begin{eqnarray*}
            \|\mu_{C_i^{\Close}\cap C_i} - \widehat{\mu_{C_i}}\|_{\mathcal{W}_1} 
&  =  & \frac{1}{|C_i^{\Close}\cap C_i|}\cdot \sum_{p\in C_i^{\Close}\cap C_i} \|e_p - \widehat{\mu_{C_i}}\|_{\mathcal{W}_1} \\
& \le & \frac{1}{|C_i^{\Close}\cap C_i|}\cdot \sum_{p\in C_i} \|e_p - \widehat{\mu_{C_i}}\|_{\mathcal{W}_1} \\
& \le & \frac{\varepsilon\cdot\OPT}{20k\cdot |C_i^{\Close}\cap C_i|\cdot |C_i|} + 
                   \frac{5\cdot \Cost(C_i)}{|C_i^{\Close}\cap C_i|\cdot |C_i|},
\end{eqnarray*}
where we have used the fact that $\widehat{\mu_{C_i}}$ is an $(\varepsilon,5)$-approximate center for $C_i$.
Thus, we have that for every $i=,1,2\dots,k$,
\begin{eqnarray}
\nonumber
&      & \sum_{p\in C_i^{\Mid}\cap C'_i} \sum_{q\in C_i^{\Close}\cap D_i} d(p,q) \\
\nonumber
&  =  & \sum_{s=1}^R \sum_{h\in H_{i,\Mid_s}} \sum_{p\in B^h\cap C'_i} \sum_{q\in C_i^{\Close}\cap D_i} d(p,q) \\
\nonumber
& \le & \sum_{s=1}^R \sum_{h\in H_{i,\Mid_s}} \sum_{p\in B^h\cap C'_i} 
            \left(\sum_{q\in C_i^{\Close}\cap C_i} d(p,q) + \sum_{q\in D_i\setminus C_i} d(p,q)\right) \\
\nonumber
& \le & \sum_{s=1}^R \sum_{h\in H_{i,\Mid_s}} \sum_{p\in B^h\cap C'_i}
              |C_i^{\Close}\cap C_i|\cdot \|e_p - \mu_{C_i^{\Close}\cap C_i}\|_{\mathcal{W}_1} \\
\nonumber
&      & + \sum_{s=1}^R \sum_{h\in H_{i,\Mid_s}} \sum_{p\in B^h\cap C'_i}
                \sum_{q\in D_i\setminus C_i} \left(\|e_p - \widehat{\mu_{C_i}}\|_{\mathcal{W}_1} + 
                   \|\widehat{\mu_{C_i}} -e_q\|_{\mathcal{W}_1}\right) \\
\nonumber
& \le & \sum_{s=1}^R \sum_{h\in H_{i,\Mid_s}} \sum_{p\in B^h\cap C'_i}
              |C_i^{\Close}\cap C_i|\cdot \left(\|e_p - \widehat{\mu_{C_i}}\|_{\mathcal{W}_1} +
               \|\widehat{\mu_{C_i}} - \mu_{C_i^{\Close}\cap C_i}\|_{\mathcal{W}_1} \right) \\
\nonumber
&      & + \sum_{s=1}^R \sum_{h\in H_{i,\Mid_s}} \sum_{p\in B^h\cap C'_i}
                \sum_{q\in D_i\setminus C_i} \left(\|e_p - \widehat{\mu_{C_i}}\|_{\mathcal{W}_1} + 
                \beta^{-2}\cdot \widehat{\widehat{\Delta_{C_i}}}\right) \\
\nonumber
& \le & \left(|C_i| + |D_i\setminus C_i|\right)\cdot \sum_{s=1}^R \sum_{h\in H_{i,\Mid_s}} \sum_{p\in B^h\cap C'_i}
                  \|e_p - \widehat{\mu_{C_i}}\|_{\mathcal{W}_1} \\
\nonumber
&      & + |C_i^{\Mid}\cap C'_i|\cdot\left(\frac{\varepsilon\OPT}{20k\cdot |C_i|} + \frac{5\cdot\Cost(C_i)}{|C_i|}\right)
            + \beta^{-2}\cdot |C_i^{\Mid}\cap C'_i|\cdot |D_i\setminus C_i|\cdot\widehat{\widehat{\Delta_{C_i}}} \\
\nonumber
& \le & (1+\gamma)(1+\gamma^2)(1+k\beta^3)\cdot |C_i| \cdot \sum_{s=1}^R \sum_{h\in H_{i,\Mid_s}} \sum_{p\in B^h\cap D_i}
                  \|e_p - \widehat{\mu_{C_i}}\|_{\mathcal{W}_1} \\
\nonumber
&      & + (1+\gamma^2)\beta^2\cdot\left(\frac{\varepsilon\OPT}{20k} + 5\cdot\Cost(C_i)\right)
            + (1+\gamma^2)k\beta^3\cdot |C_i|^2\cdot\widehat{\widehat{\Delta_{C_i}}} \\
\nonumber
& \le & (1+\gamma)(1+\gamma^2)(1+k\beta^3)\cdot |C_i| \cdot \sum_{p\in C_i^{\Mid}\cap D_i}
                  \|e_p - \mu_{C_i^{\Close}\cap C_i}\|_{\mathcal{W}_1} \\
\nonumber
&      &  + (1+\gamma)(1+\gamma^2)(1+k\beta^3)\cdot 
                  |C_i^{\Mid}\cap D_i|\cdot |C_i|\cdot \|\mu_{C_i^{\Close}\cap C_i} - \widehat{\mu_{C_i}}\|_{\mathcal{W}_1}\\
\nonumber
&      & + (1+\gamma^2)\beta^2\cdot\left(\frac{\varepsilon\OPT}{20k} + 5\cdot\Cost(C_i)\right)
            + (1+\gamma^2)k\beta^3\cdot |C_i|^2\cdot\widehat{\widehat{\Delta_{C_i}}} \\
\label{eq:mid-close cluster cost}
& \le & (1+2k\beta^3+2\gamma)\cdot \sum_{p\in C_i^{\Mid}\cap D_i} \sum_{q\in C_i^{\Close}\cap C_i} d(p,q)
              + \beta^2\cdot\left(\frac{\varepsilon\OPT}{9k} + 11\cdot\Cost(C_i)\right).
\end{eqnarray}
Summing up Equation~\eqref{eq:mid-close cluster cost} over the $k$ clusters, we get that
\begin{eqnarray}
\nonumber
&      & 2\cdot \sum_{i=1}^k \sum_{p\in C_i^{\Mid}\cap C'_i} \sum_{q\in (C_{i,\#}\cup C_{i,0})\cap D_i} d(p,q) \\
\nonumber
& \le & 2(1+2\gamma)\cdot \sum_{i=1}^k \sum_{p\in C_i^{\Mid}\cap C'_i} \sum_{q\in C_i^{\Close}\cap D_i} d(p,q)
            - 2(1+2\gamma)\cdot \sum_{i=1}^k \sum_{p\in C_i^{\Mid}\cap C'_i} \sum_{q\in C_{i,\Close}\cap D_i} d(p,q) \\
\nonumber
& \le & 2(1+3k\beta^3+5\gamma)\cdot \sum_{i=1}^k \sum_{p\in C_i^{\Mid}\cap D_i} \sum_{q\in C_i^{\Close}\cap D_i} d(p,q)
            + 23\beta^2\cdot\OPT \\
\label{eq:mid-sharp aux}
&      & - 2(1+2\gamma)\cdot \sum_{i=1}^k \sum_{p\in C_i^{\Mid}\cap C'_i} \sum_{q\in C_{i,\Close}\cap D_i} d(p,q).
\end{eqnarray}
Plugging Equation~\eqref{eq:mid-sharp aux} into Equation~\eqref{eq:mid-sharp}, then summing up 
Equations~\eqref{eq:mid-mid} ($\times 2$) and~\eqref{eq:mid-sharp} completes the proof.
\end{proof}

\begin{claim}\label{cl:C' sharp cost}
The contribution to $\Cost(\mathcal{C}')$ of distances between the pairs of points in
$(C_{i,\#}\cup C_{i,0})\cap C'_i$ satisfies
\begin{eqnarray*}
&      & \sum_{i=1}^k \sum_{p\in (C_{i,\#}\cup C_{i,0})\cap C'_i} \sum_{q\in (C_{i,\#}\cup C_{i,0})\cap C'_i} d(p,q) \\
& \le & (1+\gamma^2)^2\cdot \sum_{i=1}^k \sum_{p\in (C_{i,\#}\cup C_{i,0})\cap D_i} 
                 \sum_{q\in (C_{i,\#}\cup C_{i,0})\cap D_i} d(p,q) + 25\beta\cdot\OPT.
\end{eqnarray*}
\end{claim}

\begin{proof}
This is a simple application of Claim~\ref{cl:H-isharp}, which gives
\begin{eqnarray*}
&      & \sum_{i=1}^k \sum_{p\in (C_{i,\#}\cup C_{i,0})\cap C'_i} \sum_{q\in (C_{i,\#}\cup C_{i,0})\cap C'_i} d(p,q) \\
& \le & \sum_{i=1}^k \sum_{h,h'\in H_{i,\#}} 
               \frac{|B^h\cap C'_i|\cdot |B^{h'}\cap C'_i|}{|B^h\cap D_i|\cdot |B^{h'}\cap D_i|}\cdot 
               \sum_{p\in B^h\cap D_i} \sum_{q\in B^{h'}\cap D_i}
            \left(d(p,q) + 4\beta\cdot\widehat{\widehat{\Delta_{C_i}}}\right) \\
&      & + 2\cdot \sum_{i=1}^k \sum_{h\in H_{i,\#}} \frac{|B^h\cap C'_i|}{|B^h\cap D_i|}\cdot 
                    \sum_{p\in B^h\cap D_i} \sum_{q\in C_{i,0}\cap D_i} \left(d(p,q) + 2\beta\cdot\widehat{\widehat{\Delta_{C_i}}}\right) \\
&      & + \sum_{i=1}^k \sum_{p\in C_{i,0}\cap D_i} \sum_{q\in C_{i,0}\cap D_i} d(p,q) \\
& \le & (1+\gamma^2)^2\cdot \sum_{i=1}^k \sum_{p\in (C_{i,\#}\cup C_{i,0})\cap D_i} 
                 \sum_{q\in (C_{i,\#}\cup C_{i,0})\cap D_i} d(p,q) + 25\beta\cdot \OPT,
\end{eqnarray*}
thus completing the proof.
\end{proof}

\begin{lemma}\label{lem:computelargefartoclose}
Given approximate cluster centers, cluster sizes, cluster costs, and optimum value that
satisfy Hypothesis~\ref{hyp:approx}, there exists an algorithm that computes in time 
$n^{1+o(1)}\cdot 2^{2^{O(k\log(k/\varepsilon))}}$ a set of at most
$n^{o(1)}\cdot 2^{2^{O(k\log(k/\varepsilon))}}$ partial clusterings with the 
following property. At least one of these partial clusterings 
$\mathcal{P} = \{P_1,P_2,\dots,P_k\}$ can be extended to a clustering 
$\mathcal{C}'= \{C'_1,C'_2,\dots,C'_k\}$ that satisfies Lemma~\ref{lm:cost C'},
and for every $i=1,2,\dots,k$, we have that $C_i^{\Far}\cap C'_i = P_i$.
\end{lemma}

\begin{proof}
Note that the number of possible values of $h$ is at most $(R+k+3)^k$, with 
$R\le 5\gamma^{-1}\log \beta^{-1}$. For every $h$ and $i$ we need
the estimate $\widehat{|B^h\cap D_i|}$. This requires enumerating over
$$
(\gamma^{-2}\cdot\log n)^{k\cdot(R+k+3)^k} = n^{o(1)}\cdot 2^{2^{O(k\log(k/\varepsilon))}}
$$ 
values, assuming that both $\beta$ and $\gamma$ are $\left(\frac{\varepsilon}{k}\right)^{O(1)}$.

To compute the desired partial clustering, we need to compute, for each of the above guesses,
an optimal assignment $f$. This can be done by computing a minimum cost flow in a flow network.
Constructing a flow network with each individual point being a source node would result in a
flow computation with prohibitive time complexity (if we're aiming at sub-linear time). Therefore,
we must partition the points into buckets of interchangeable options. The set of values $h$ is
insufficient as points in $B^h$ where $h(i) = \Far$ are not interchangeable with respect to an
assignment to $i$ because the costs could vary widely. So we need to refine the bucketing to
indicate approximate distances (just for far points; rounding to powers of $1+\gamma$ which
would inflate the entire approximation guarantee by this factor). The number of buckets is at
most the order of
$$
(R+k+3)^k\cdot (\gamma^{-1}\log n)^k = n^{o(1)}\cdot 2^{O(k\log(k/\varepsilon))},
$$
under all of the above assumptions.

Assigning the points to buckets is done once, and requires $O(nk)$ operations. Solving the
assignment problem requires time polynomial in the size of the flow network. Extracting
the solution $\mathcal{P}$ from the assignment takes $O(n)$ time. These two steps are
executed for each of the $n^{o(1)}\cdot 2^{2^{O(k\log(k/\varepsilon))}}$ cardinality guesses.
\end{proof}

Lemma~\ref{lem:computelargefartoclose} implies Theorem~\ref{thm:farassignment}.
\begin{proof}[Proof of Theorem \ref{thm:farassignment}]
Use Lemma~\ref{lem:computelargefartoclose} to compute a set of partial clusterings,
one of which can be extended to satisfy Lemma~\ref{lm:cost C'}. To complete the proof, 
assume that $\gamma\le \frac{\varepsilon}{30 k}$, and use $\beta\le \frac{\varepsilon}{300 k^3}$
(leaving plenty of room to spare).
\end{proof}
\subsection{Assignment of Close Points}
\label{sec:closeassignment}

Following Theorem~\ref{thm:farassignment}, we now have to deal with assigning
to clusters all the points that were not assigned by the (correct) partial 
clustering $\mathcal{P}$. For simplicity, we abuse notation slightly and denote 
by $\mathcal{C}$ (instead of $\mathcal{C}'$) the extension of $\mathcal{P}$ to a 
full clustering (which is w.h.p. near-optimal for at least one of the partial
clustering options enumerated upon). We also assume that we have approximate values 
$\widehat{\mu_{C}}$, $\widehat{\Delta_{C}}$, and $\widehat{|C|}$, as per
Hypothesis~\ref{hyp:approx} applied to the extension $\mathcal{C}$ of
$\mathcal{P}$ (with $\gamma = (\varepsilon/k)^{O(1)}$).

Notice that the set of points unassigned by $\mathcal{P}$ excludes, in particular,
for every cluster $i$, all the points in $C_i$ that are $\gamma^{-1}$-far for $C_i$.
Thus, define for every $i$ the set 
\begin{equation}\label{eq:L_i}
L_i = \left\{p\in A\colon d(p,\widehat{\mu_{C_i}}) < \gamma^{-1}\cdot\widehat{\widehat{\Delta_{C_i}}}\right\},
\end{equation}
the set of points that are $\gamma^{-1}$-close for $C_i$, where we recall Equation~\eqref{eq:Delta hat hat}
that puts
$$
\widehat{\widehat{\Delta_{C_i}}} = 5\cdot\widehat{\Delta_{C_i}} + 
\frac{\varepsilon\cdot (1+\gamma)^2}{20k}\cdot \frac{\widehat{\OPT}}{\left(\widehat{|C_i|}\right)^2}
\le (1+\gamma)\cdot \left(5\Delta_{C_i} + \frac{\varepsilon(1+\gamma)^2}{20k}\cdot \frac{\OPT}{|C_i|^2}\right)
$$
Ideally, the set $L_i$ only 
contains points of $C_i$ and no other points, but we cannot assume this to be the case 
in general. In particular, for two clusters $i,j$ it may be the case that 
$L_i\cap L_j\neq \emptyset$. Therefore, we first apply easy, local assignment steps 
similar to those in Section~\ref{sec:farassignment}. The remaining assignments cannot 
be resolved locally. For these we apply approximation schemes solving MAX $h$-CUT, 
for varying choices of $h$.

\paragraph{Local assignments.}
In the following, we use a parameter $\beta = \left(\frac{\varepsilon}{k}\right)^{O(1)}$ which 
we will specify later. Consider two clusters $C_i$ and $C_j$ with $L_i\cap L_j\neq \emptyset$.
The following lemma shows that if $\widehat{|C_j|}$ is significantly smaller than $\widehat{|C_i|}$, 
then we can assign all the points in $L_i\cap L_j\cap C_j$ to cluster $i$ without increasing 
the cost substantially.
\begin{lemma}
\label{lem:localassignsize}
For a cluster $C_i$, let 
$C'_i = C_i\bigcup \left(L_i\cap 
  \left(\bigcup_{j\colon\widehat{|C_j|} < \beta\cdot \widehat{|C_i|}} \{C_j\cap L_j\}\right)\right)$.
Then 
$$
\Cost(C_i') < \Cost(C_i) + 25\cdot k^2\cdot\beta\cdot \gamma^{-1}\cdot 
             \left(\Cost(C_i) + \frac{\varepsilon}{100k}\cdot \OPT\right).
$$
\end{lemma}

\begin{proof}
The number of points assigned to $C_i$ from all clusters in $\mathcal{C}$ is at most 
$k\cdot \beta\cdot \widehat{|C_i|} < k\cdot\beta\cdot (1+\gamma)\cdot |C_i|$. Consider 
$$
\Cost(C_i') = \Cost(C_i)  + 2\cdot \sum_{p\in C_i} \sum_{q\in C_i'\setminus C_i}d(p,q)  
   + \sum_{p\in C_i'\setminus C_i} \sum_{q\in C_i'\setminus C_i}d(p,q).
$$
We will bound the latter two terms individually.
For the second term, we have 
\begin{eqnarray*}
       d(p,q) 
&\leq& \|e_p-\mu_{C_i}\|_{{\cal W}_1} + \|\widehat{\mu_{C_i}}-\mu_{C_i}\|_{{\cal W}_1} + d(q,\widehat{\mu_{C_i}}) \\
& <  & \|e_p-\mu_{C_i}\|_{{\cal W}_1} + 
       \left(1 + \frac{1+\gamma}{\gamma}\right)\cdot 
       \left(5\Delta_{C_i} + \frac{\varepsilon(1+\gamma)^2}{20k}\cdot \frac{\OPT}{|C_i|^2}\right),
\end{eqnarray*}
where we used Claim~\ref{cl:approx dist} and Equation~\eqref{eq:L_i}.
We get
\begin{eqnarray*}
       2 \sum_{p\in C_i} \sum_{q\in C_i'\setminus C_i}d(p,q) 
&\leq& 2\cdot |C_i'\setminus C_i|\cdot\left(\frac{\Cost(C_i)}{|C_i|} + 
       \left(1+\frac{1+\gamma}{\gamma}\right)\cdot\left(5\cdot\frac{\Cost(C_i)}{|C_i|} + 
       \frac{\varepsilon(1+\gamma)^2}{20k}\cdot \frac{\OPT}{|C_i|}\right)\right) \\
& <  & 14\cdot k\cdot\beta\cdot \gamma^{-1}\left(\Cost(C_i) + \frac{\varepsilon}{130k}\cdot \OPT\right).
\end{eqnarray*}

Now, for every $p,q\in C_i'\setminus C_i$, we have that 
$d(p,q)\le 2\cdot\frac{1+\gamma}{\gamma}\cdot \left(5\Delta_{C_i} + 
\frac{\varepsilon\cdot (1+\gamma)^2}{20k}\cdot \frac{\OPT}{|C_i|^2}\right)$, so for the third term,
\begin{eqnarray*}
       \sum_{p\in C_i'\setminus C_i} \sum_{q\in C_i'\setminus C_i}d(p,q) 
&\leq& |C_i'\setminus C_i|^2 \cdot 2\cdot \frac{1+\gamma}{\gamma}\cdot \left(5\Delta_{C_i} + 
       \frac{\varepsilon\cdot (1+\gamma)^2}{20k}\frac{\OPT}{|C_i|^2}\right) \\
& <  & 11\cdot k^2\cdot \beta^2\cdot \gamma^{-1}\cdot \left(\Cost(C_i)+ \frac{\varepsilon}{100k}\cdot \OPT\right).
\end{eqnarray*}
Summing up the excess error for both cases yields the claim.
\end{proof}

We now use Lemma~\ref{lem:localassignsize} to revise the sets $L_i$ as follows. Choose
$\beta = \frac{\varepsilon}{100k^3}\cdot \gamma$. For all pairs of clusters $i,j$ with 
$\widehat{|C_j|} < \beta\cdot\widehat{|C_i|}$, in lexicographic order of non-increasing 
size, remove $L_i\cap L_j$ from $L_j$. (We'll keep denoting the modified sets, during
and following this process, as $L_i$ etc.) Now define a conflict graph $K$ over the set
of clusters. We connect every pair $i,j$ of clusters by an edge in $K$ iff 
$L_i\cap L_j\ne\emptyset$. The following lemma shows some properties of $K$.
\begin{lemma}
\label{lem:uniform}
Consider a connected component $K'$ of $K$, and let $i,j$ be two clusters in $K'$.
Then, it holds that $|\widehat{C_i}|\leq \beta^{-k+1}\cdot |\widehat{C_j}|$. Moreover,
$$
d(\widehat{\mu_{C_i}},\widehat{\mu_{C_j}})\leq 7\cdot \gamma^{-1} \cdot \beta^{-4k+6}\cdot 
\frac{\widehat{\OPT}}{\left(\max\left\{\widehat{|C_i|},\widehat{|C_j|}\right\}\right)^2}.
$$
\end{lemma}

\begin{proof}
Notice that if $i,j$ are in the same connected component of $K$, then they are connected
by a path of length at most $k-1$, and every edge along this path connects a pair of clusters
such that the smaller size is at least $\beta$ times the bigger size. This proves the first
claim. For the second claim, use induction on the length of the path. The base case is $i,j$
connected by an edge. Let $p\in L_i\cap L_j$ (the existence of the edge implies that the 
intersection is non-empty). We have
\begin{eqnarray}
\nonumber
      d(\widehat{\mu_i},\widehat{\mu_j})
&\le& d(\widehat{\mu_i},p) + d(p,\widehat{\mu_j}) \\
\nonumber
& < & \gamma^{-1}\cdot 
      \left(\frac{\varepsilon\cdot (1+\gamma)^2\cdot \widehat{\OPT}}{20k}\cdot
      \left(\frac{1}{\widehat{|C_i|}^2} + \frac{1}{\widehat{|C_j|}^2}\right)
      + 5\cdot\left(\widehat{\Delta_{C_i}} + \widehat{\Delta_{C_j}}\right)\right) \\
\label{eq:inductive}
&\le& \gamma^{-1}\cdot \frac{1}{\left(\min\left\{\widehat{|C_i|},\widehat{|C_j|}\right\}\right)^2}\cdot
       \left(\frac{\varepsilon\cdot \widehat{\OPT}}{9k} + 6\cdot\left(\Cost(C_i)+\Cost(C_j)\right) \right) \\
\nonumber
&\le& 7\cdot\gamma^{-1}\cdot \frac{\widehat{\OPT}}{\left(\min\left\{\widehat{|C_i|},\widehat{|C_j|}\right\}\right)^2} \\
\nonumber
&\le& 7\cdot\gamma^{-1}\cdot\beta^{-2k+2}\cdot
      \frac{\widehat{\OPT}}{\left(\max\left\{\widehat{|C_i|},\widehat{|C_j|}\right\}\right)^2}.
\end{eqnarray}
The induction sums up Equation~\eqref{eq:inductive} over at most $k-1$ pairs of adjacent nodes of $K$.
Using the fact that the minimum size cluster $C$ along this path satisfies
$\widehat{|C|}\ge \beta^{k-2}\cdot \min\left\{\widehat{|C_i|},\widehat{|C_j|}\right\}$
and the fact that 
$\max\left\{\widehat{|C_i|},\widehat{|C_j|}\right\}\le 
\beta^{-k+1}\cdot\min\left\{\widehat{|C_i|},\widehat{|C_j|}\right\}$
completes the proof.
\end{proof}

With this in place, we can now prove the main claim of this section.
\begin{proof}[Proof of Theorem~\ref{thm:finalguarantee}]
Firsly, let's sum the excess error due to Lemma~\ref{lem:localassignsize}. 
Specifically, for every cluster $C_i$, we incur an increase in cost of at most 
$25\cdot k^2\cdot\beta\cdot \gamma^{-1}\cdot 
             \left(\Cost(C_i) + \frac{\varepsilon}{100k}\cdot \OPT\right)$,
due to Lemma~\ref{lem:localassignsize}.
Summing over all clusters yields a total excess cost of
$25\cdot k^3\cdot\beta\cdot \gamma^{-1}\cdot \left(\frac{1}{k} + \frac{\varepsilon}{100k}\right)\cdot\OPT.$
Plugging in the aforementioned value of $\beta$ implies that the excess cost is at most 
$\frac{\varepsilon}{2k}\cdot \OPT$.

Adding this to the excess cost incurred by the partial assignment satisfying the
bound of Theorem~\ref{thm:farassignment}, which in particular assigns to clusters
all the points not in the union of all the $L_i$-s, yields that there exists a
a clustering $\mathcal{C}' = \{C'_1,C'_2,\dots,C'_k\}$ with the following
properties: ($i$) the cost of this clustering is at most 
$\left(1+\frac{\varepsilon}{k}\right)\cdot\OPT$; and, ($ii$) we can identify,
for every connected component $K'$ of the graph $K$, the set of points 
$A_{K'} = \bigcup_{i\in V(K')} C'_i$ (i.e., the points that belong
to the $\mathcal{C}'$ clusters in $K'$).

To complete the computation of the clustering, use the the algorithm of
Theorem~\ref{thm:maxcut} to compute for every connected component $K'$
a nearly-optimal MAX $|V(K')|$-CUT partition of $A_{K'}$. The approximation
parameter will be defined as we analyze the resulting clustering.

We first relate the cost of MAX $|V(K')|$-CUT to the cost of $|V(K')|$-minsum 
clustering of $A_{K'}$. The MAX $|V(K')|$-CUT cost is at most 
$\sum_{p\in A_{K'}}\sum_{q\in A_{K'}}d(p,q)$. 
Due to Lemma \ref{lem:uniform}, we have for every $i\in V(K')$, 
$|C'_i| \geq \frac{\beta^{k-1}}{k}\cdot |A_{K'}|$, as the largest
such cluster has size at least $\frac 1 k \cdot |A_{K'}|$.

Likewise, consider an arbitrary $\widehat{\mu_{C_i}}$, $i\in V(K')$. Then
\begin{eqnarray*}
        \sum_{p,q\in A_{K'}} d(p,q)
&\leq&  |A_{K'}|\cdot \sum_{p\in A_{K'}} d(p,\widehat{\mu_{C_i}}) \\
&\leq & |A_K'|\cdot\sum_{j\in V(K')}\sum_{p\in C_j'}d(p,\widehat{\mu_{C_j}}) +
        d(\widehat{\mu_{C_j}},\widehat{\mu_{C_i}})\\
&\leq & \gamma^{-1}\cdot |A_K'|\cdot\sum_{j\in V(K')} 
        \left(\frac{\varepsilon\cdot (1+\gamma)^2}{20}\cdot \frac{\widehat{\OPT}}{k\cdot \widehat{|C_j|}} +
        5\cdot\widehat{\Delta_C}\cdot \widehat{|C_j|} + 7\beta^{-4k+6}\cdot
        \frac{\widehat{\OPT}}{\widehat{|C_j|}}\right) \\
&\leq & 8\gamma^{-1} \cdot \beta^{-5k+7}\cdot k\cdot \widehat{OPT}.
\end{eqnarray*}
This is an upper bound on the optimal value of a MAX $|V(K')|$-CUT. If we now
set the approximation parameter for MAX $|V(K')|$-CUT to be
$\varepsilon' = \frac{\varepsilon\cdot\gamma\cdot\beta^{5k-7}}{k^2}$,
this yields a MAX $|V(K')|$-CUT that deviates from the optimum by at most
an $\frac{\varepsilon}{k}\cdot\OPT$ additive term. Thus, the complemet
$|V(K')|$-minsum clustering objective deviates from the optimum value
by at most the same additive term. Summing over at most $k$ connected
components yields a $k$-minsum clustering solution of cost at most
$\varepsilon\cdot\OPT$ more than the optimum, which is at most the
cost of the clustering $\mathcal{C}'$.

The running time is dominated the calls to the metric MAX $|V(K')|$-CUT
approximation scheme, for at most $k$ connected components $K'$, with
$\varepsilon' = \frac{\varepsilon\cdot\gamma\cdot\beta^{5k-7}}{k^2}$.
We can take $\gamma = \left(\frac{\varepsilon}{k}\right)^{O(1)}$,
so $\beta$ is just a higher constant power of $\frac{\varepsilon}{k}$.
With these parameters, and failure probability $\delta$,
we get a running time bounded by
$(k/\varepsilon)^{O(k)}\cdot n\cdot\log(n/\delta) + 2^{(k/\varepsilon)^{O(k)}}$.
\end{proof}

\subsection{A Sublinear Algorithm for Metric MAX \texorpdfstring{$k$}{k}-CUT}
\label{sec:maxcut}

In the metric MAX $k$-CUT problem we are given a finite metric space $(A,d)$, 
where $A$ is the set of points, and $d:A\times A\rightarrow \mathbb{R}_{\geq 0}$ 
is the distance function, and we must compute a partition 
$A = C_1\sqcup C_2\sqcup\cdots\sqcup C_k$ of $A$ into $k$ clusters that maximizes
the sum of inter-cluster distances. Thus, the objective function is precisely the
complement of the objective of the $k$-minsum clustering problem, and the two
problems on the same input have the same set of optimal partitions. However,
approximate solutions may be drastically different. In this section we will prove 
the following theorem.
\begin{theorem}
\label{thm:maxcut}
There exists a randomized algorithm that computes a $(1-\varepsilon)$-approximation 
to the optimal MAX $k$-CUT solution in time 
$O(\varepsilon^{-2}\cdot n\cdot\log(n/\delta)) + 
O\left(\varepsilon^{-3}\cdot n\cdot k^4\right) + k^{O(\varepsilon^{-2})}$,
where $\delta$ is the failure probability.
\end{theorem}

We will use the following result giving fast, sublinear algorithms for dense constraint 
satisfaction problems (CSP). An $r$-CSP over an alphabet of size $k$ is a collection of 
constraints, each acting on at most $r$ variables from a set of $n$ variables, such that 
each variable can assume $k$ values. For example, MAX CUT is a special case of a $2$-CSP 
over Boolean variables. 
\begin{theorem}[Theorem 1.2 \cite{Yaroslavtsev14}]
\label{thm:maxcutdense}
There is an algorithm which approximates any $2$-CSP problem over alphabet of size $k$ 
within additive error $\varepsilon\cdot n^2$ and runs in time 
$O\left(\varepsilon^{-2}\cdot n\cdot k^4\right) + k^{O(\varepsilon^{-2})}$.
\end{theorem}
As MAX $k$-CUT is a special case of $2$-CSP over an alphabet of size $k$, we immediately 
get the following corollary.
\begin{corollary}
\label{cor:maxcutdense}
Given a weighted undirected graph with maximum weight $M$, there exists an algorithm 
that computes MAX $k$-CUT with additive error $\varepsilon\cdot M\cdot n^2$ in time 
$O\left(\varepsilon^{-2}\cdot n\cdot k^4\right)+ k^{O(\varepsilon^{-2})}$.
\end{corollary}

We ultimately wish to use Corollary~\ref{cor:maxcutdense} to solve metric MAX $k$-CUT. 
Unfortunately, merely interpreting the distances between nodes as weights is not sufficient, 
as $\max_{p,q} d(p,q)\cdot n^2$ may be significantly larger than the value of the optimum 
cut which is close to $\sum_{p,q\in A} d(p,q)$ (note that a random partition guarantees in
expectation a fraction of $\frac{k-1}{k}$ of the sum of all pairwise distances).

In what follows, we present a sublinear time reduction of the metric case to the dense 
weighted case, largely inspired by the reduction from~\cite{FK98} for MAX CUT.
We first outline the construction.
For every point $p\in A$, let $w(p) = \sum_{q\in A} d(p,q)$. We assume we are given 
estimates $\widehat{w(p)}\in (1\pm \varepsilon) w(p)$. We also assume that the distances
and estimates are scaled so that 
$W:=\sum_{p\in P} \widehat{w(p)} = n\cdot 2\varepsilon^{-1}$. Now, construct a graph
$G = (A',E)$ with edge weights $v:E\rightarrow\mathbb{R}_+$ as follows.
\begin{itemize}
    \item For each $p\in A$, add $\lfloor\widehat{w(p)}\rfloor$ copies of $p$ to $A'$.
    \item For each pair of points $p,q$ and for each copy $p'$ of $p$ and $q'$ of $q$, 
          add an edge of weight $v(p',q')=\frac{d(p,q)}{\widehat{w(p)}\cdot \widehat{w(q)}}$
          to $E$.
\end{itemize}
We show that: (i) $|A'| = O(n)$; (ii) the graph is dense, i.e. the average edge weight is 
within a constant factor of the maximum edge weight; (iii) the cost of any metric $k$-cut 
is approximated by the cost of the corresponding $k$-cut in $G$; and (iv) any $k$-cut in 
$G$ can be mapped to a metric $k$-cut without loss in approximation ratio.
We begin by analyzing the quality of the cut. We first discuss some properties of the
values $\widehat{w(p)}$ and their rounded versions.
\begin{lemma}
\label{lem:maxcutweights}
The following assertions are true.
\begin{itemize}
\item $\frac{1}{2}\cdot W < |A'| = \sum_{p\in A} \lfloor \widehat{w(p)}\rfloor \leq n\cdot 2\varepsilon^{-1}$.
\item For any $p\in A$, $W \leq \frac{1+\varepsilon}{1-\varepsilon}\cdot 2n \cdot \widehat{w(p)}$.
\item For any $p\in A$, $\lfloor\widehat{w(p)}\rfloor \geq 
      \left(\frac{1-\varepsilon}{1+\varepsilon}\right)^2\cdot \widehat{w(p)}$.
\end{itemize}
\end{lemma}

\begin{proof}
We begin with the second claim. Let $p\in A$.
$$
n\cdot 2\varepsilon^{-1} = W =\sum_{q\in A} \widehat{w(q)} \leq  
\frac{1}{1-\varepsilon}\cdot \sum_{q\in A}\sum_{r\in A} d(q,r) \leq 
\frac{1}{1-\varepsilon}\cdot 2n\cdot w(p) \leq  
\frac{1+\varepsilon}{1-\varepsilon}\cdot 2n \cdot \widehat{w(p)},
$$
where the penultimate inequality follows from the triangle inequality:
for all $q\in A$, $\sum_{r\in A} d(q,r)\le \sum_{r\in A} (d(p,q) + d(p,r))$.

As for the third claim, the above inequality implies that
$\widehat{w(p)}\geq \varepsilon^{-1}\cdot \frac{1-\varepsilon}{1+\varepsilon}$, 
which implies
$$
\lfloor\widehat{w(p)}\rfloor > \widehat{w(p)}-1\ge 
\left(1 - \varepsilon\cdot \frac{1+\varepsilon}{1-\varepsilon}\right)\cdot \widehat{w(p)}\ge
\left(\frac{1-\varepsilon}{1+\varepsilon}\right)^2\cdot \widehat{w(p)},
$$
where the last inequality assumes $\varepsilon$ is sufficiently small ($\le\frac 1 3$ would
suffice).

Finally, the first inequality follows from the above bound that implies that
$\lfloor\widehat{w(p)}\rfloor\ge 1$ for all $p$, so $|A'|\ge n$, and as
$\widehat{w(p)} < \lfloor\widehat{w(p)} + 1$, we get $W < |A'|+n\leq 2|A'|$.
\end{proof}

Let $\Val(C_1,C_2,\dots,C_k)$ denote the sum of inter-cluster pairwise distances
(or edge weights, in the case of a graph).
\begin{lemma}
\label{lem:cutapprox}
Let $A = C_1\sqcup C_2\sqcup\cdots\sqcup C_k$ be a partition of a metric space, and let 
$C_1',\ldots C_k'$ be the corresponding partition of $A'$, where $C'_i$ consists of all 
duplicates of $C_i$. Then,
$$
\Val(C_1,\ldots C_k) \geq \Val(C_1',\ldots C_k')\geq 
\left(\frac{1-\varepsilon}{1+\varepsilon}\right)^4\cdot \Val(C_1,\ldots C_k).
$$ 
\end{lemma}

\begin{proof}
The upper bound follows immediately. For the lower bound, let $E_C\subset E$ be the 
the set of edges that are cut by the partition of $A'$. Then,
\begin{eqnarray*}
       \sum_{e\in E_C} v(e) 
& =  & \sum_{i\ne j}\sum_{p\in C_i\wedge j\in C_j} \frac{\lfloor \widehat{w(p)}\rfloor}{\widehat{w(p)}}\cdot 
       \frac{\lfloor \widehat{w(q)}\rfloor}{\widehat{w(q)}}\cdot d(p,q)\\
&\geq& \left(\frac{1-\varepsilon}{1+\varepsilon}\right)^4\cdot\sum_{i\ne j}\sum_{p\in C_i\wedge j\in C_j} d(p,q),
\end{eqnarray*}
where the inequality uses Lemma~\ref{lem:maxcutweights}.
\end{proof}

Next, we argue that any partition assigning duplicates of a point to different clusters can be improved by 
assigning all the duplicates of the same point to the same cluster.
\begin{lemma}
Let $A' = C'_1\sqcup C'_2\sqcup\cdots\sqcup C'_k$ be an arbitrary partition. Then, there exists 
a partition $A' = C''_1\sqcup C''_2\sqcup\cdots\sqcup C''_k$ that satisfies that for all $p\in A$,
all the duplicates $p'\in A'$ of $p$ are in the same cluster, and
$\Val(C''_1,\dots,C''_k)\le\Val(C'_1,\dots,C'_k)$.
\end{lemma}

\begin{proof}
Let $p\in A$. Consider a cluster $C'_i$ containing a duplicate $p'$ of $p$, such that the weight of the cut edges
adjacent to $p'$ is maximized among all duplicates. As there are no edges between duplicates, and all duplicates have the same neighbor set, moving all duplicates to $C'_i$ cannot decrease the value of the partition. Repeating this
change for all $p$ completes the proof.
\end{proof}
\begin{lemma}\label{lem:frac}
Let $P\colon V_H\to[k]$ be an arbitrary labeling of the clone graph, and for $x\in X$ and $\ell\in[k]$ let $p_x(\ell)=|\{a\in\lfloor\widehat{w(x)}\rfloor\colon P(x,a)=\ell\}|/\lfloor\widehat{w(x)}\rfloor$ be the fraction of copies of $x$ labeled $\ell$. Then
\begin{equation}\label{eq:frac}
\text{val}_h(P)\;=\;\sum_{\{x,y\}\subseteq X}d(x,y)\Big(1-\sum_{\ell\in[k]}p_x(\ell)\,p_y(\ell)\Big).
\end{equation}
Let $c_P\colon X\to[k]$ be the random labeling obtained by choosing, independently for every $x\in X$, a uniformly random index $a_x\in[m(x)]$ and setting $c_P(x)=P(x,a_x)$. Then
\[
\E\big[\text{cut}(c_P)\big]\;=\;\text{val}_h(P).
\]
\end{lemma}
\begin{proof}
For $x\neq y$ and $\ell\in[k]$ let $B_x^\ell=\{(x,a)\colon P(x,a)=\ell\}$, so $|B_x^\ell|=p_x(\ell)\,m(x)$. The number of pairs $(u,v)\in B_x\times B_y$ with $P(u)\neq P(v)$ is
\[
\lfloor\widehat{w(x)}\rfloor\lfloor\widehat{w(y)}\rfloor-\sum_{\ell}|B_x^\ell|\,|B_y^\ell|=\lfloor\widehat{w(x)}\rfloor\lfloor\widehat{w(y)}\rfloor\Big(1-\sum_\ell p_x(\ell)p_y(\ell)\Big),
\]
and each such pair has weight $d(x,y)/(\lfloor\widehat{w(x)}\rfloor\lfloor\widehat{w(y)}\rfloor)$; pairs of copies of the same point have weight $0$. Summing over $\{x,y\}$ gives~\eqref{eq:frac}.

For the rounding, $\Pr[c_P(x)=\ell]=p_x(\ell)$ for every $x$ and $\ell$, and the choices for different points are independent, so $\Pr[c_P(x)\neq c_P(y)]=1-\sum_\ell p_x(\ell)p_y(\ell)$ for $x\neq y$. Hence, by linearity of expectation, $\E[\text{cut}(c_P)]=\sum_{\{x,y\}}d(x,y)\Pr[c_P(x)\neq c_P(y)]$ equals the right-hand side of~\eqref{eq:frac}.
\end{proof}

Finally, we show a lemma that implies that the graph $G$ is dense.
\begin{lemma}
\label{lem:maxcutdense}
For any edge $e\in E$, we have 
$$
v(e)\leq 16\cdot\left(\frac{1+\varepsilon}{1-\varepsilon}\right)^6\cdot\frac{\sum_{e'\in E}v(e')}{|A'|^2}.
$$
\end{lemma}

\begin{proof}
Consider an edge $\{p',q'\}\in E$, where $p'$ and $q'$, respectively, are duplicates of
$p$ and $q$, respectively. We have
\begin{align*}
    v(p',q') &= \frac{d(p,q)}{\widehat{w(p)}\cdot \widehat{w(q)}} \leq \frac{\frac{1}{n}\sum_{r\in A} (d(p,r) + d(q,r))}{\widehat{w(p)}\cdot \widehat{w(q)}} \leq 
    (1+\varepsilon)\frac{\widehat{w(p)} + \widehat{w(q)}}{n\cdot \widehat{w(p)}\cdot \widehat{w(q)}} \\
    &\leq \left(\frac{1+\varepsilon}{1-\varepsilon} \right)^2 \cdot \frac{4}{W} \leq 
    \left(\frac{1+\varepsilon}{1-\varepsilon} \right)^2 \cdot \frac{8}{|A'|},
\end{align*}
where the penultimate inequality uses the second bullet point of Lemma~\ref{lem:maxcutweights},
and the last inequality uses the first bullet point of the same lemma.
Now, the average weight of an edge is
\begin{align*}
    \frac{\sum_{e'\in E}v(e')}{|A'|^2} &=    
    \frac{\sum_{p,q\in A}\frac{\lfloor\widehat{w(p)}\rfloor\cdot\lfloor\widehat{w(q)}\rfloor\cdot d(p,q)}{\widehat{w(p)}\cdot \widehat{w(q)}}}{|A'|^2} \ge
    \left(\frac{1-\varepsilon}{1+\varepsilon}\right)^4 \cdot \frac{W}{|A'|^2} \geq 
    \frac{1}{2}\cdot \left(\frac{1-\varepsilon}{1+\varepsilon}\right)^4\cdot \frac{1}{|A'|},
\end{align*}
where we used the first bullet point of Lemma~\ref{lem:maxcutweights} in the final inequality.
\end{proof}

We now are ready to prove Theorem~\ref{thm:maxcut}.
\begin{proof}[Proof of Theorem~\ref{thm:maxcut}]
We first compute the values $\widehat{w(p)}$ for all $p$, using Lemma~\ref{lem:estimatingCosts}.
This takes $O(n\cdot \varepsilon^{-2}\cdot \log(n/\delta))$ time, and succeeds with probability
at least $1 - \delta$.

Subsequently, we consider the graph $G$ (we don't need to construct it explicitly). Due to
Lemma~\ref{lem:maxcutweights}, $G$ has $O(n\cdot \varepsilon^{-1})$ nodes. Due to 
Lemma~\ref{lem:maxcutdense}, the graph is dense with maximum edge weight 
$8\cdot\left(\frac{1+\varepsilon}{1-\varepsilon}\right)^6\frac{\sum_{e\in A'} v(e)}{|A'|^2}$.
Thus, letting $\eta$ be the parameter controlling the error term (to be determined
shortly), we may use Corollary~\ref{cor:maxcutdense} to compute, in time 
$O\left(\eta^{-2}\cdot |A'|\cdot k^4\right) + k^{O(\eta^{-2})}$, a $k$-cut 
$A' = C'_1\sqcup C'_2\sqcup\cdots\sqcup C'_k$ with
\begin{eqnarray*}
       \Val(C'_1,\dots,C'_k)
&\geq& \OPT - 16\cdot\eta\cdot \left(\frac{1+\varepsilon}{1-\varepsilon}\right)^6\cdot 
              \frac{\sum_{e'\in E}v(e')}{|A'|^2}\cdot |A'|^2 \\
&\geq& \left(1-16\cdot\eta\cdot \left(\frac{1+\varepsilon}{1-\varepsilon}\right)^6\right) \cdot \OPT,
\end{eqnarray*}
as $\OPT\le\sum_{e'\in E}v(e')$.
Now, by Lemma~\ref{lem:cutapprox}, we have for the corresponding metric $k$-cut
$A = C_1\sqcup C_2\sqcup\cdots\sqcup C_k$ that
$$
\Val(C_1,C_2,\dots,C_k)\ge 
\left(\left(\frac{1-\varepsilon}{1+\varepsilon}\right)^4 - 
16\cdot\eta\cdot \left(\frac{1+\varepsilon}{1-\varepsilon}\right)^{2}\right) \cdot \OPT,
$$
where $\OPT$ here indicates the metric optimum.
For sufficiently small $\varepsilon$, it suffices to set $\eta \leq \frac{\varepsilon}{32}$
to yield additive error less than $10\cdot\varepsilon\cdot\OPT$. 
Using Lemma \ref{lem:frac} allows us to round the fractional cut into a true cut, where we may use repeated rounding instances to obtain a cut with desired error with high probability. 
The afore-mentioned running
time amounts to
$O\left(\varepsilon^{-3}\cdot n\cdot k^4\right) + k^{O(\varepsilon^{-2})}$.
Adding the time complexity of computing the $\widehat{w(p)}$-s gives the theorem's
stated running time.
\end{proof}

\subsection{Computing the Cost of a Cluster}
\label{sec:cluster:cost}

In this section, we will prove Theorem~\ref{thm:costguarantee}. We first require the following 
lemma, which is also used in Section~\ref{sec:maxcut}.
\begin{lemma}
\label{lem:estimatingCosts}
Let $A,B$ be two sets of points in some metric space and let $\varepsilon, \delta > 0$.
We can compute, using a randomized $O(n\cdot\log \delta^{-1})$-time algorithm, a set 
$S\subset A$ endowed with weights $v:S\rightarrow \mathbb{R}_{\geq 0}$ such that 
$|S| = O(\varepsilon^{-2}\cdot \log(|B|/\delta))$ and, with probability at least
$1-\delta$, for every $p\in B$,
$$
(1-\varepsilon)\cdot \sum_{q\in A} d(p,q) \le \sum_{q\in S}v(q)\cdot d(p,q) \leq 
(1+\varepsilon)\cdot \sum_{q\in A} d(p,q).
$$
\end{lemma}

\begin{proof}
Denote $D(p) =\sum_{q\in A} d(p,q)$. We need to estimate the values $D(p)$
for all $p\in B$.

Let $c = \arg\min_{r\in A} \sum_{q\in A} d(q,r)$ be a median point of $A$.
We first argue that we can compute quickly a suitable proxy for $c$. Let
$c'\in A$ be a point sampled uniformly at random from $A$. By the triangle
inequality, any point $c'\in A$ satisfies
$\sum_{q\in A} d(q,c') \leq \sum_{q\in A} d(q,c) + |A|\cdot d(c,c')$.
By Markov's Inequality, $\Pr[d(c,c') > \frac{2}{n}\cdot \sum_{q\in A} d(q,c)] < \frac 1 2$.
If $d(c,c') \le \frac{2}{n}\cdot \sum_{q\in A} d(q,c)$, we get that
$\sum_{q\in A} d(q,c') \leq 3\cdot \sum_{q\in A} d(q,c)$. 
Sampling a set of $\log \delta^{-1}$ points from $A$ uniformly 
and independently at random guarantees that one of the sampled
points satisfies this inequality with probability at least $1-\delta$.

Next, we construct a set $S$ by sampling points in $A$ independently 
at random from the probability distribution $\mathbb{P}$, defined by 
$\mathbb{P}[q]=\frac{1}{2}\cdot \left(\frac{d(q,c')}{\sum_{q'\in A} d(q',c')} + \frac{1}{|A|}\right)$.
For $q\in S$, we set $v(q) = \frac{1}{|S|\cdot \mathbb{P}[q]}$.
Our estimator of $D(p)$ is therefore
$$
\widehat{D(p)}=\frac{1}{|S|}\sum_{q\in S} \frac{1}{\mathbb{P}[q]}d(p,q).
$$
Clearly, $\mathbb{E}[\widehat{D(p)}] = \sum_{q\in A} d(p,q) = D(p)$.
To choose $|S|$ and to complete the proof, we use concentration bounds.
First, we bound the maximum contribution of any sampled point to the value the 
estimator. We have 
\begin{eqnarray*}
       \frac{d(p,q)}{\mathbb{P}[q]} 
&\leq& \frac{\left(d(p,c') + d(q,c')\right)}{\mathbb{P}[q]} \\
& =  & 2\cdot \frac{\left(d(p,c') + d(q,c')\right)\cdot|A|\cdot \sum_{q'\in A}d (q',c')}
                   {|A|\cdot d(q,c') + \sum_{q'\in A} d(q',c')} \\
&\leq& 2 \cdot \left(|A|\cdot d(p,c') + \sum_{q'\in A} d(q',c') \right) \\
&\leq& 2 \cdot \left(\sum_{q\in A} (d(p,q) + d(q,c')) + \sum_{q'\in A} d(q',c')\right).
\end{eqnarray*}
Now, by definition, $2\cdot \sum_{q\in A} d(p,q) = 2\cdot D(p)$. Letting
$c'' = \arg\min_{q\in A} d(p,q)$, we have
$$
4\cdot \sum_{q\in A} d(q,c') \leq 12\cdot \sum_{q\in A} d(q,c)\leq
12\cdot \sum_{q\in A} d(q,c'') \le 12\cdot\left(\sum_{q\in A} d(p,q) + |A|\cdot d(p,c'')\right)\le
24\cdot D(p).
$$
Combining everything, we get that $\frac{d(p,q)}{\mathbb{P}[q]}\le 26\cdot D(p)$.
Thus, using Hoeffding's inequality, we have
$$
\mathbb{P}\left[|\widehat{D(p)} - D(p)| > \varepsilon\cdot D(p)\right] \leq 
2\cdot\exp\left(-\frac{2\cdot\varepsilon^2\cdot D(p)^2\cdot |S|}{(26\cdot D(p))^2}\right) \leq 
\exp\left(-\frac{\varepsilon^2\cdot |S|}{26^2}\right).
$$
Setting $|S|\geq 26^2 \varepsilon^{-2} \log(|B|/\delta)$ and applying a union bound for all $p\in B$ 
yields the claim.
\end{proof}

\begin{proof}[Proof of Theorem \ref{thm:costguarantee}]
Due to Lemma \ref{lm: center based}, it is sufficient to estimate 
$\sum_{p\in X} \|e_p-\mu_X\|_{{\cal W}_1}$ in order to estimate the sum of pairwise 
distances of the cluster. Using $B=\{\mu_X\}$ and $A=X$, we may invoke 
Lemma~\ref{lem:estimatingCosts} and identify, in time $O(|X|\cdot\log \delta^{-1})$, a 
weighted set of at most $O(\varepsilon^{-2}\cdot\log \delta^{-1})$ points $S$ such that 
evaluating $\sum_{p\in S} v_p \|e_p-\mu_X\|_{{\cal W}_1}$ yields a 
$(1\pm \varepsilon)$-approximate estimate of $\sum_{p\in X} \|e_p-\mu_X\|_{{\cal W}_1}$, 
with probability at least $1-\delta$. Evaluating $\|e_p-\mu_X\|$ takes time $O(|X|)$ 
for every $p\in S$, so the total running time is $O(|X|\cdot \varepsilon^{-2}\cdot\log \delta^{-1})$.
\end{proof}

\section{Conditional Lower Bounds}\label{sec:lb}

In this section, we prove conditional lower bounds for the parameterized $k$-minsum and balanced 2-median problems. In Section~\ref{sec:lbprelims}, we list the relevant hypotheses and results in literature that we will use. In Section~\ref{sec:lbSETH}, we prove the conditional lower bounds for balanced 2-median (i.e, formal version of Theorem~\ref{thm:SETHinformal}). Finally, in Section~\ref{sec:lbgapeth}, we prove the conditional lower bounds for the parameterized   $k$-minsum problem.

\subsection{Preliminaries}\label{sec:lbprelims} In this subsection, we first list the hypotheses on the SAT problem that are used to prove our lower bounds.

\begin{hypothesis}[Strong Exponential Time Hypothesis (\SETH)~\cite{IPZ98,IP01}] \label{hyp:seth}
For every $\eps > 0$, there exists $k = k(\eps) \in \N$ such that  no $2^{(1-\eps) n}$-time algorithm that takes as input a k-SAT formula on $n$ variables can decide if it is satisfiable. 
\end{hypothesis}

The below is a version of Gap-ETH, where we have already applied the ``expander replacement'' transformation \cite{PY88} to the starting instance.

\begin{hypothesis}[Gap Exponential Time Hypothesis (Gap-ETH)~\cite{D16,MR17}] \label{hyp:gap-eth}
There exists some $\eps,\delta>0$ such that no $2^{\eps n}$-time algorithm that takes as input a 3-SAT formula on $n$ variables  where each variable appears in at most $O(1)$ clauses, can decide if it is satisfiable or if every assignment violates at least $\delta$ fraction of the clauses. 
\end{hypothesis}

Next, we recall the \maxcover problem introduced in \cite{CCKLMNT17} which turned out to be  the centerpiece of many results in parameterized inapproximability. 

\paragraph{\boldmath $k$-\textsf{MaxCover} problem.}
A $k$-\maxcover instance $\Gamma$ consists of a bipartite graph $(V\dot\cup W, E)$ such that $V$ is partitioned into $V=V_1\dot\cup \cdots \dot\cup V_k$ and $W$ is partitioned into $W=W_1\dot\cup \cdots \dot\cup W_\ell$. We sometimes refer to $V_{i}$'s and $W_j$'s as {\em left super-nodes} and {\em right super-nodes} of $\Gamma$, respectively. 

A solution to $k$-\maxcover is called a {\em labeling},
which is a subset of vertices $v_1\in V_1,\dots v_k\in V_k$.
We say that a labeling $v_1,\dots v_k$ {\em covers} a right super-node $W_{i}$, if there exists a vertex $w_{i} \in W_{i}$ which is a joint neighbor of all $v_1,\dots v_k$, i.e., $(v_j,w_i)\in E$ for every $j\in[k]$.
We denote by $\maxcover(\Gamma)$ the maximal fraction of right super-nodes that can be simultaneously covered, i.e.,
\begin{align*}
\maxcover(\Gamma) = \frac{1}{\ell} \left(\max_{\text{labeling } v_1,\dots v_k} \bigl\lvert\bigl\{i \in [\ell] \mid W_i \text{ is covered by } v_1,\dots v_k\bigr\}\bigr\rvert\right). 
\end{align*}

The following inapproximability for \maxcover result was implicitly proved by Rubinstein~\cite{Rub18} (however, in terms of the bichromatic maximum inner product problem).

\begin{theorem}[Essentially from~\cite{Rub18}]\label{thm:maxcover}
   Assuming SETH, for every $\eps>0$ there exist $\delta=\delta(\eps)>0$, $c=c(\eps)>1$, and $D=D(\eps)>1$ such that no $n^{2-\eps}$-time algorithm that takes as input a $2$-$\maxcover$ instance $\Gamma=(V\dot\cup W,E)$ with $V=V_1\dot\cup V_2$, $W=W_1\dot\cup\cdots\dot\cup W_\ell$, $|V_1|=|V_2|=n$, $\ell\le c\cdot \log n$, and $|W_1|=\cdots = |W_\ell|\le D$ for all $i\in[\ell]$, can distinguish whether $\maxcover(\Gamma)=1$ or $\maxcover(\Gamma)\le 1-\delta$.
\end{theorem}

However, for the sake of completeness, we include the  proof outline of the above theorem.  
We find it more convenient to prove Theorem~\ref{thm:maxcover} which is implicit in~\cite{Rub18}, using the notations and framework provided in  \cite{KLM18}. 

We refer the reader to~\cite{KLM18} for the definitions of the ``$(w,r,\ell,s)$-efficient protocol'' in the  
Simultaneous Message Passing Model. 
We recall below the parameters of an efficient protocol for the set-disjointness function given in~\cite{KLM18}.

\begin{theorem}[Karthik, Laekhanukit, Manurangsi; Rewording of Theorem 6.1 in~\cite{KLM18}]\label{thm:CC} 
There is a polynomial function $\hat{\ell}:\mathbb N\times \mathbb N\to [1,\infty)$ such that 
for every  $\alpha\in\mathbb N$, there is a protocol  in the Simultaneous Message Passing Model for two players to compute if the two input strings are disjoint, which is an $( \nicefrac{m}{\alpha},\log_2 m,\hat{\ell}(\alpha),\nicefrac{1}{2})$-efficient protocol, where each player is given $m$ bits as input. 
\end{theorem}
 
By plugging in the above protocol parameters to the below theorem we obtain Theorem~\ref{thm:maxcover}.

  \begin{theorem}[Karthik, Laekhanukit, Manurangsi; Rewording of Corollary 5.3 in~\cite{KLM18}]\label{cor:labSETH}
For every $\delta>0$, suppose there exists a $(w,r,\ell,s)$-efficient protocol in the Simultaneous Message Passing Model for two players to compute if the two input strings are disjoint  for every input size $m\in\mathbb N$, such that $w\le \delta m$ and $r+\ell =o(m)$.
Then, assuming \SETH, for every $\eps>0$,   no $O(N^{ 2 - \eps})$-time algorithm  can distinguish between $\maxcover(\Gamma) = 1$
and $\maxcover(\Gamma) \leq s$ for any 2-\maxcover instance  $\Gamma=(V\dot\cup W,E)$ with $V=V_1\dot\cup V_2$, $W=W_1\dot\cup\cdots\dot\cup W_t$, $|V_1|=|V_2|=N$, $t=2^r$, and $|W_1|=\cdots =|W_t|=2^{2\ell}$,  for all $N \in \N$.
\end{theorem}

\paragraph{Triangle Packing problem.} In the Triangle Packing problem, we are given a graph $G=(V,E)$, and the goal is to determine if there are $|V|/3$ vertex disjoint triangles in $G$.

The following is an immediate consequence from Patrank \cite{Patrank94} (although the step from 3-dimensional matching to triangle packing is not written down). Thus, for the sake of completeness we refer here to the proof in \cite{GuptaLLM019}.

\begin{theorem}[Essentially Lemma 4.4. in \cite{GuptaLLM019}]\label{thm:gapethVC}
  Assuming   Gap-ETH, 
  there exist some $\eps,\delta>0$ and $B>2$, such that no $2^{\eps n}$-time algorithm that takes as input a triangle packing instance $G=(V,E)$ on $n$ vertices with maximum degree $B$ can distinguish between the following:
   \begin{description}
      \item[Completeness:] There are $|V|/3$ vertex disjoint triangles in $G$. 
      \item[Soundness:] $G$ has at most $(1-\delta)|V|/3$ vertex
disjoint triangles.
  \end{description}
\end{theorem}

The theorem above holds because in \cite{GuptaLLM019}, the authors provide a linear-sized, gap-preserving reduction chain from bounded-degree 3-SAT to Max 1-in-3SAT and then to Partitioning into Triangles (PIT), i.e., the triangle packing problem.

\subsection{SETH Based Lower Bound for Balanced 2-Median}\label{sec:lbSETH}

In this section, we prove the formal version of Theorem~\ref{thm:SETHinformal}. 

First, we prove a conditional lower bound on the maxcoverage problem when the size of the optimum solution is just two sets. This is obtained by combining Theorem~\ref{thm:maxcover} with the gadgets introduced by Feige in~\cite{Fei98}. Later we will use this result to prove the lower bound for balanced 2-median.

\begin{theorem}[Subquadratic lower bound for $k$-maxcoverage problem when $k=2$]\label{thm:maxcoverage}
  Assuming SETH, for every $\eps>0$ there exist $\delta=\delta(\eps)>0$ and $c:=c(\eps)$ such that no $n^{2-\eps}$-time algorithm that takes as input a collection of sets $S_1,\ldots ,S_n$ over a universe $U$ of size $c\cdot \log n$ can distinguish between the following two cases:
  \begin{description}
      \item[Completeness:] There are some $i,j\in [n]$ and an equipartition of $U=U_1\dot\cup U_2$ such that $U_1\subseteq S_i$ and $U_2\subseteq S_j$. 
      \item[Soundness:] For every $i,j\in [n]$ we have $|S_i\cup S_j|\le (1-\delta)\cdot |U|$.  
  \end{description}
\end{theorem}
\begin{proof} 
For every $\eps>0$, let $c_{\eps}, D_\eps, $ and $\delta_\eps$ be the constants specified by Theorem~\ref{thm:maxcover}. 
Suppose that for some constant $\eps_0>0$, there is an algorithm $\mathcal{A}$ running in time $n^{2-\eps_0}$ that takes as input $n$ sets over a universe of size $2^{D_{{\eps'}}}\cdot c_{\eps'}\log n$ and distinguishes the completeness and soundness cases in the theorem statement where $\delta$ in the soundness case is equal to $\delta_{\eps'}/2^{D_{\eps'}}$, where $\eps':=\eps_0/2$. We will then be relying on Theorem~\ref{thm:maxcover} to arrive at a contradiction.

  Let   $\Gamma=(V\dot\cup W,E)$ be a $2$-$\maxcover$ instance with $V=V_1\dot\cup V_2$, $W=W_1\dot\cup\cdots\dot\cup W_\ell$, $\tilde{D}:=|W_1|=\cdots =|W_{\ell}|$. We will construct below in $O(|V|\cdot |W|\cdot 2^{\tilde D}\cdot \tilde D)$ time  a collection of sets $S_1,\ldots ,S_{|V|}$ over a universe $U$ of size $\ell\cdot 2^{\tilde D}$  such that if $\maxcover(\Gamma)=1$ then there are some $i,j\in [|V|]$ and an equipartition of $U=U_1\dot\cup U_2$ such that $U_1\subseteq S_i$ and $U_2\subseteq S_j$, and instead if $\maxcover(\Gamma)\le 1-\tilde{\delta}$ (for some $\tilde\delta>0$) then for every $i,j\in [|V|]$ we have $|S_i\cup S_j|\le \left(1-\frac{\tilde{\delta}}{2^{\tilde D}}\right)\cdot |U|$.  

We now show how to decide $\Gamma=(V\dot\cup W,E)$ using $\mathcal{A}$ in time $n^{2-\eps'}$, where $n=|V|/2$. Given $\Gamma$, we first construct the set system in time $O(|V|\cdot |W|\cdot 2^{\tilde D}\cdot \tilde D)=O(n\cdot \log n)$, where we have used the assumption that $\eps_0$ is a fixed constant. Then, we run $\mathcal{A}$ and decide if $\maxcover(\Gamma)=1$ or $\maxcover(\Gamma)<1-\delta_{\eps'}$ in time $O(n^{2-\eps_0})$ (where we set $\tilde\delta$ to be $\delta_{\eps'}$). The total runtime is thus $O(n^{2-\eps_0})<n^{2-\eps'}$, for large $n$, and this contradicts Theorem~\ref{thm:maxcover}.
  
\paragraph{Construction of Set System.}
 Given an instance $\Gamma=(V\dot\cup W,E)$ of $2$-$\maxcover$   with $V=V_1\dot\cup V_2$ and $W=W_1\dot\cup\cdots\dot\cup W_\ell$ with $|V_1|=|V_2|$ and $|W_1|=\cdots = |W_\ell| =\tilde{D}$, we construct the universe $U=\{(i,f): i\in [\ell]\text{ and }f:W_i\to \{1,2\}\}$. The size of this universe is exactly $\ell \cdot 2^{\tilde D}$. For every $j\in\{1,2\}$, and every $v\in V_j$, we construct a set $S_{j,v}$. Thus we have $|V|$ many sets. We have the universe element $(i,f)$ in $S_{j,v}$ if and only if there is some $w\in W_i$ such that $f(w)=j$ and $(v,w)\in E$. The time needed for the construction of a set is $O(|U|\cdot \tilde D)$, and thus the set system can be constructed in time $O(|V|\cdot |U|\cdot \tilde D)$.

\paragraph{Completeness Analysis.}
  Suppose $\maxcover(\Gamma)=1$, i.e., there exists some $(v_1,v_2)\in V_1\times V_2$ such that for all $i\in [\ell]$ there is a node $w_i\in W_i$ such that $(v_1,w_i)$ and $(v_2,w_i)$ are both in $E$. We will now define a partition of $U$ to $U_1$ and $U_2$ as follows. For every $(i,f)\in U$ we assign $(i,f)$ to $U_{f(w_i)}$. 
  
Next we claim that $|U_1|=|U_2|$. For any fixed $i \in [\ell]$, consider the $2^{\tilde{D}}$ functions $f:W_i \to \{1,2\}$. The assignment of $(i,f)$ to either $U_1$ or $U_2$ depends only on the value $f(w_i)$. Exactly half of these functions will have $f(w_i)=1$, and the other half will have $f(w_i)=2$. Thus, for each $i$, we add $2^{\tilde{D}-1}$ elements to $U_1$ and $2^{\tilde{D}-1}$ elements to $U_2$. Summing over all $i \in [\ell]$ gives $|U_1|=|U_2| = \ell \cdot 2^{\tilde{D}-1}$.

Finally, we claim that $U_1\subseteq S_{1,v_1}$ and $U_2\subseteq S_{2,v_2}$. To see this, fix $j\in\{1,2\}$. Let $(i,f)\in U_j$. Then, we know that $f(w_i)=j$. Moreover, we have that $(v_j,w_i)\in E$. Thus, $(i,f)\in S_{j,v_j}$.
  
\paragraph{Soundness Analysis.}
  Suppose that $\maxcover(\Gamma)\le 1-\tilde{\delta}$. Let $S_{j_1,v_1}$ and $S_{j_2,v_2}$ be two sets. 
  Here we consider two cases. 

  \noindent\textbf{Case 1: $j_1=j_2$.} Let $\bar{j}\neq j_1$, i.e., $\bar{j}:=3-j_1=3-j_2$. For every $i\in [\ell]$, consider the function $f_i:W_i\to\{1,2\}$, where $f_i(w):=\bar{j}$ for all $w\in W_i$. By construction of the sets, we have that for all $i\in[\ell]$, we have $(i,f_i)\notin S_{j_1,v_1}\cup S_{j_2,v_2}$, as for all $w\in W_i$ we have $f_i(w)\neq j_1$ and $f_i(w)\neq j_2$. Thus, $|S_{j_1,v_1}\cup S_{j_2,v_2}|\le |U|-\ell= \left(1-\frac{1}{2^{\tilde D}}\right)\cdot |U|$. 

  \noindent\textbf{Case 2: $j_1\neq j_2$.} Without loss of generality, we assume $j_1=1$ and $j_2=2$. Then $\{v_1,v_2\}$ is a labeling for $\Gamma$. By the soundness assumption, there is a subset $X\subseteq [\ell]$ of size $\tilde\delta \cdot \ell$ such that for all $i\in X$ and for all $w\in W_i$, there is some $j_{i,w}\in \{1,2\}$ such that $(v_{j_{i,w}},w)\notin E$. Thus, for every $i\in X$, we construct the function $f_i:W_i\to\{1,2\}$ by defining $f_i(w)=j_{i,w}$ for all $w\in W_i$. It is easy to then see that for all $i\in X$, we have $(i,f_i)$ is not covered by $S_{j_1,v_1}$ or $S_{j_2,v_2}$. Therefore, we have $|S_{j_1,v_1}\cup S_{j_2,v_2}|\le |U|-\tilde\delta\cdot \ell=\left(1-\frac{\tilde{\delta}}{2^{\tilde D}}\right)\cdot |U|$. 
\end{proof}

Now, we are ready to prove our conditional lower bound for balanced 2-median, which also applies to the standard 2-median problem (and even the 2-means problem which follows imemdiately from the proof of the below theorem).

\begin{theorem}[Subquadratic lower bound for $2$-median and balanced 2-median problems]\label{thm:median}
  Assuming SETH, for every $\eps>0$ there exists $\delta=\delta(\eps)>0$ such that no $n^{2-\eps}$-time algorithm that takes as input a collection of facilities $F:=\{f_1,\ldots ,f_n\}$, a collection of clients $C=\{c_1,\ldots ,c_m\}$ (where $m=O_{\eps}(\log n)$), and a distance function $\Delta:C\times F\to \mathbb{R}_{\ge 0}$, can distinguish between the following two cases:
  \begin{description}
      \item[Completeness:] There are some $i,j\in [n]$ and an equipartition of $C=C_1\dot\cup C_2$ such that, $$\left(\sum_{c\in C_1}\Delta(c,f_i)\right)+\left(\sum_{c\in C_2}\Delta(c,f_j)\right)=|C|.$$ 
      \item[Soundness:] For every $i,j\in [n]$ we have, $$ \sum_{c\in C}\ \min\{\Delta(c,f_i),\Delta(c,f_j)\}\ge (1+\delta)\cdot |C|.$$ 
  \end{description}
\end{theorem}
\begin{proof}
    For every $\eps>0$, let  $\delta_\eps$ and $c_\eps$ be the constants specified by Theorem~\ref{thm:maxcoverage}. 
Suppose that for some constant $\eps_0>0$, there is an algorithm $\mathcal{A}$ running in time $n^{2-\eps_0}$ that takes as input $n$ facilities and $c_{\eps'}\log n$ clients and distinguishes the completeness and soundness cases in the theorem statement where $\delta$ in the soundness case is equal to $2\cdot \delta_{\eps'}$, where $\eps'=\eps_0/2$. We will then be relying on Theorem~\ref{thm:maxcoverage} to arrive at a contradiction.

  Let  $S_1,\ldots ,S_n$ over a universe $U$   be a maxcoverage instance. We will construct below in $O(n\cdot |U|)$ time  a collection of facilities $F:=\{f_1,\ldots ,f_n\}$, a collection of clients $C=\{c_1,\ldots ,c_{|U|}\}$ such that if   there are some $i,j\in [n]$ and an equipartition of $U=U_1\dot\cup U_2$ such that $U_1\subseteq S_i$ and $U_2\subseteq S_j$ then there are some $i,j\in [n]$ and an equipartition of $C=C_1\dot\cup C_2$ such that, $\left(\sum_{c\in C_1}\Delta(c,f_i)\right)+\left(\sum_{c\in C_2}\Delta(c,f_j)\right)=|U|$, and instead if   for every $i,j\in [n]$ we have $|S_i\cup S_j|\le \left(1-\tilde\delta\right)\cdot |U|$ then for every $i,j\in [n]$ we have, $ \sum_{c\in C}\ \min\{\Delta(c,f_i),\Delta(c,f_j)\}\ge (1+2\tilde\delta)\cdot |U|$.  

We now show how to decide the maxcoverage instance $S_1,\ldots ,S_n$ over a universe $U$ using $\mathcal{A}$ in time $n^{2-\eps'}$. Given the set system, we first construct the collection of facilities and clients in time $O(n\cdot |U|)=O(n\cdot \log n)$, where we have used the assumption that $\eps_0$ is a fixed constant. Then, we run $\mathcal{A}$ and decide if  there are some $i,j\in [n]$ and an equipartition of $U=U_1\dot\cup U_2$ such that $U_1\subseteq S_i$ and $U_2\subseteq S_j$ or if for every $i,j\in [n]$ we have $|S_i\cup S_j|\le \left(1-\delta_{\eps'}\right)\cdot |U|$  in time $O(n^{2-\eps_0})$ (where we set $\tilde\delta$ to be $\delta_{\eps'}$). The total runtime is thus $O(n^{2-\eps_0})<n^{2-\eps'}$, for large $n$, and this contradicts Theorem~\ref{thm:maxcoverage}.
  
\paragraph{Construction of Facilities and Clients.}
 Given an instance $S_1,\ldots ,S_n$ over a universe $U=\{u_1,\ldots ,u_m\}$, we construct a facility $f_i$ for every set $S_i$ (where $i\in[n]$) and a client $c_j$ for every universe element $u_j$. We construct the distance function $\Delta$ as follows:
 $$\forall j\in[m],i\in[n],\quad \Delta(c_j,f_i)=\begin{cases}1&\text{ if }u_j\in S_i\\ 3&\text{ if }u_j\notin S_i\end{cases}.$$
 The time needed for the construction of the distance function is   $O(mn)$.

\paragraph{Completeness Analysis.}
  Suppose there are some $i^*,j^*\in [n]$ and an equipartition of $U=U_1\dot\cup U_2$ such that $U_1\subseteq S_{i^*}$ and $U_2\subseteq S_{j^*}$, then consider the equipartition of $C=C_1\dot\cup C_2$ such that for all $j\in [m]$ and $u_j\in U_1$ if and only if $c_j\in C_1$. Then, we have that for all $c_j\in C_1$ (resp.\ $c_j\in C_2$), we have $\Delta(c_j,f_{i^*})=1$ (resp.\ $\Delta(c_j,f_{j^*})=1$) as $u_j\in S_{i^*}$ (resp.\ $u_j\in S_{j^*}$). Thus, we have:
  \begin{align*}
      \left(\sum_{c\in C_1}\Delta(c,f_{i^*})\right)+\left(\sum_{c\in C_2}\Delta(c,f_{j^*})\right) = |C_1|+|C_2| = |U|.
  \end{align*}

\paragraph{Soundness Analysis.}
  Suppose that for every $i,j\in [n]$ we have $|S_i\cup S_j|\le \left(1-\tilde\delta\right)\cdot |U|$. Let $f_{i'}$ and $f_{j'}$ be any two facilities. By the soundness assumption, there is a subset $X\subseteq U$ of size $\tilde\delta \cdot |U|$ such that  $X\cap (S_{i'}\cup S_{j'})=\emptyset$. Thus, we have that for all $u_j\in X$, we have $\Delta(c_j,f_{i'})=3$ (as $u_j\notin S_{i'}$) and $\Delta(c_j,f_{j'})=3$ (as $u_j\notin S_{j'}$). Thus, we have:
  \begin{align*}
      \sum_{c\in C}\ \min\{\Delta(c,f_{i'}),\Delta(c,f_{j'})\} &= \left(\sum_{c\in X}\ \min\{\Delta(c,f_{i'}),\Delta(c,f_{j'})\}\right)+\left(\sum_{c\in C\setminus X}\ \min\{\Delta(c,f_{i'}),\Delta(c,f_{j'})\}\right)\\ &\ge 3\cdot |X|+ 1\cdot (|U|-|X|) =|U|+2\cdot |X|\ge (1+2\tilde{\delta})\cdot |U|.\qedhere
  \end{align*}
\end{proof}

\subsection{Parameterized Lower Bound from Gap-ETH for Min-Sum}\label{sec:lbgapeth}

In this section, we prove Theorem~\ref{thm:lbminsum}. The proof relies on Theorem~\ref{thm:gapethVC} as a starting point and follows along the lines of the proof given by \cite{GI03}. We include it here for the sake of completeness.

\begin{theorem}\label{thm:lbminsum}
Suppose there is a $2^{o(k/\eps)}\cdot 2^{o(n)}$ time $(1+\eps)$-approximation algorithm
for metric $k$-minsum (for every $n, k, \eps$ and a metric $d$ on $n$ points). Then Gap-ETH is false.
\qed
\end{theorem}
\begin{proof}
Let $\eps_0,\delta_0, B_0$ be the constants from Theorem~\ref{thm:gapethVC}. We assume without loss of generality that $\delta_0<1/55$.
        Suppose there is an algorithm $\mathcal{A}$ running in time $2^{o(k/\eps)}\cdot 2^{o(N)}$ for $(1+\eps)$-approximating the $k$-Minsum problem on $N$ points for every $k,\eps, N$. We will use $\mathcal{A}$ to contradict Theorem~\ref{thm:gapethVC}. 

        Let $G=(V,E)$ be a triangle packing instance on $n$ vertices and $m$ edges with maximum degree $B$.  
        We construct in $O(n^2)$ time a collection of $n$ points $P$ in a metric space such that if $G$ has $n/3$ vertex disjoint triangles then for $k:=n/3$, the cost of $k$-minsum clustering of $P$ is at most $n$, and instead if there are at most $(1-\delta)\cdot n/3$ vertex disjoint triangles in $G$ then the cost of $k$-minsum clustering of $P$ is more than $(1+\delta^{2})n$.

We now show how to decide the gap triangle packing instance $G=(V,E)$  using $\mathcal{A}$ in time $2^{o(n)}$. Given the graph, we first construct the collection of points in time $O(n^2)$. Then, we run $\mathcal{A}$ and decide if  $G$ has $n/3$ vertex disjoint triangles or if there are at most $(1-\delta)\cdot n/3$ vertex disjoint triangles in $G$, in time $2^{o(n/3\delta_0^2)}\cdot 2^{o(n)}=2^{o(n)}$ (where we used that $\delta_0$  is a positive constant). The total runtime is thus $2^{o(n)}$ and this contradicts Theorem~\ref{thm:gapethVC}.

\paragraph{Construction of Point-set.}
Given an instance $G=(V,E)$ of triangle packing, we construct  a $k$-Minsum instance $P$ of points, where $P = \{p_v \mid v \in V\}$ and $k=|V|/3$. Moreover, the metric $d$ induced over $P$ is given as follows:  For $u \neq v$, $d(p_u, p_{v}) = 1$ if $\{u,v\}\in E$, and $2$ otherwise. Finally, we define for all $v\in V$, $d(p_v, p_v) = 0$. 
The construction time is dominated by computing $O(n^2)$ pairwise distances.

\paragraph{Completeness Analysis.}
If $G$ has a triangle packing of $n/3$ vertex-
disjoint triangles, using these triangles as the clusters gives a
$k$-clustering of total cost $n$.

\paragraph{Soundness Analysis.}
Suppose there is a $k$-clustering $C$ with cost at
most $(1+\delta^{2})n$, then in \cite{GI03}, the authors prove that  as long as $\delta<1/55$ one reaches a contradiction.
\end{proof}


\newpage
\bibliographystyle{plain}
\bibliography{ref} 

\end{document}